\documentclass[graybox]{svmult}

\usepackage{type1cm}
\usepackage{multicol}        
\usepackage[bottom]{footmisc}
\usepackage{newtxtext}       %
\usepackage[varvw]{newtxmath}       

\usepackage[T1]{fontenc}
\usepackage{textcomp}
\usepackage{xcolor}

\usepackage{makeidx}         
\usepackage{multicol}        
\usepackage[bottom]{footmisc}

\usepackage{setspace}
\usepackage{tikz}

\usepackage{doi}

\usepackage{array}
\usepackage{natbib}
\bibpunct{(}{)}{;}{a}{,}{;}
\usepackage{graphicx}
\usepackage{bbm}
\usepackage{mathtools}
\usepackage{booktabs}
\usepackage{icomma}

\usepackage{algorithmic}

\usepackage[labelfont={small,sc},textfont={small,it}, margin=15pt,skip=10pt,position=bottom]{caption}
\usepackage[labelfont={scriptsize,normalfont},textfont={scriptsize,it}, margin=20pt,skip=5pt,position=bottom, labelformat=simple]{subcaption}

\usepackage{algorithmic}
\usepackage[boxed]{algorithm}

\newtheorem{assumption}{Assumption}

\newcommand{\EQ}{\begin{equation}}
\newcommand{\EN}{\end{equation}}
\newcommand{\EQS}{\begin{equation*}}
\newcommand{\ENS}{\end{equation*}}
\newcommand{\ds}{\displaystyle}

\newcommand*\xbar[1]{%
  \hbox{%
    \vbox{%
      \hrule height 0.5pt
      \kern0.5ex%
      \hbox{%
        \kern-0.1em%
        \ensuremath{#1}%
        \kern-0.1em%
      }%
    }%
  }%
}

\def\argmax{\mathop{\rm arg\,max}}

\def\E{\mathcal E}

\def\zz{z_1}
\def\zzz{z_2}

\newcommand{\qq}{q}
\newcommand{\pp}{p}

\newcommand{\plainv}{v}
\newcommand{\tildev}{\tilde{v}}

\newsavebox{\savepar}

\newcommand{\myred}{\color{black}}

\def\argmax{\mathop{\rm arg\,max}}

\usepackage[running,mathlines]{lineno}

\makeindex             

\begin{document}
\title*{Numerical methods for optimal decumulation of a defined contribution pension plan}
\author{Peter A. Forsyth and \\George Labahn}
\institute{Peter A. Forsyth \at David R. Cheriton School of Computer Science,  University of Waterloo, Waterloo ON, Canada N2L 3G1, \email{paforsyt@uwaterloo.ca}
\and George Labahn \at David R. Cheriton School of Computer Science,  University of Waterloo, Waterloo ON, Canada N2L 3G1, \email{glabahn@uwaterloo.ca}}

\maketitle

\abstract*{The decumulation of a defined contribution (DC) pension plan is well known to be one of the hardest problems in finance.
We model this decumulation challenge as an  optimal stochastic control problem.
The control problem is solved, at each  rebalancing date, by alternatively 
solving a linear partial-integro differential  equation (PIDE) followed by an optimization
step.  We solve the  PIDE by using a $\delta$-monotone Fourier method, which ensures that monotonicity holds to $O(\delta)$.  We allow for the use of leverage
(i.e. borrowing to invest in stocks), as well as minimum constraints on bond holdings.  We pay particular attention to minimizing wrap-around error, an issue  which is
endemic for Fourier methods and central to the effective use of these methods for optimal control problems.
Rather unexpectedly, we find that restricting the portfolio equity fraction to a maximum of 50\% does not reduce portfolio efficiency noticeably.
This may be a useful strategy for risk-averse retirees}


\begin{abstract}{
The decumulation of a defined contribution (DC) pension plan is well known to be one of the hardest problems in finance.
We model this decumulation challenge as an  optimal stochastic control problem.
The control problem is solved, at each  rebalancing date, by alternatively 
solving a linear partial-integro differential  equation (PIDE) followed by an optimization
step.  We solve the  PIDE by using a $\delta$-monotone Fourier method, which ensures that monotonicity holds to $O(\delta)$.  We allow for the use of leverage
(i.e. borrowing to invest in stocks), as well as minimum constraints on bond holdings.  We pay particular attention to minimizing wrap-around error, an issue  which is
endemic for Fourier methods and central to the effective use of these methods for optimal control problems.
Rather unexpectedly, we find that restricting the portfolio equity fraction to a maximum of 50\% does not reduce portfolio efficiency noticeably.
This may be a useful strategy for risk-averse retirees.}

\vspace{5pt}
\noindent
\textbf{Keywords:} decumulation,
stochastic control, risk

\noindent
\textbf{JEL codes:} G11, G22\\
\noindent
\textbf{AMS codes:} 91G, 65N06, 65N12, 35Q93

\end{abstract}

\section{Introduction}\label{intro_section}
In developed countries, there is a strong and growing shift from Defined Benefit (DB)  pension plans to Defined Contribution (DC) plans. The private sector, and increasingly
government institutions, are unwilling to take on the balance sheet liabilities of a DB plan \citep{Thinking_2024}.

Under a DC plan, the employer and employee contribute to a (usually) tax advantaged account, which is typically invested in a mix of stocks and bonds. Upon retirement,
the employee is now responsible for both 
choosing an investment strategy and  
deciding on a withdrawal schedule. 
Although it is often advocated that retirees should buy annuities, these are unpopular with DC plan holders \citep{Peijnenburg2016} for numerous valid reasons 
\citep{MacDonald2013}.  In particular, annuities are generally overpriced, 
lack true inflation protection and retirees 
have no access to capital in the event of emergencies.

Since it is usually challenging to fund a reasonable lifestyle investing solely in riskless assets, the retiree is now
exposed to both investment risk and longevity risk.  This problem of funding retirement 
from accumulated DC capital, also known as the
decumulation problem, has been termed {\em``the nastiest, hardest problem in finance,''} 
by Nobel Laureate William Sharpe ( quoted in \citep{Sharpe2017} ).

There have been several suggestions for DC plan {\em spending rules} \citep{Anarkulova_2022_a}. Probably the most ubiquitous rule is the {\em four per cent rule} \citep{Bengen1994}.
This rule suggests that retirees can withdraw four per cent of their initial capital (at age 65) annually, adjusted for inflation.  Based on 
historical US data, an investor who invested in a portfolio of 50\% bonds and 50\% stocks (rebalanced annually),  and followed the four per
cent rule, would have never run out of savings over any thirty year historical period.

However, the four per cent rule has many critics.  For example,  rolling thirty year periods have high correlations.  In addition, a fixed withdrawal schedule,
as well as a constant weight stock allocation is somewhat simplistic. In order to  better stress test spending rules, \citet{Anarkulova_2022_a} use block bootstrap resampling of historical data, to generate
a distribution of outcomes.  We will also use block bootstrap resampling to test our results in this chapter.

Decumulation strategies are naturally posed as a problem in optimal stochastic control.  The controls in this case are  the stock/bond split on each rebalancing date and  the (real) annual withdrawal amounts. In our case we impose maximum and minimum constraints on the withdrawal amounts and also impose constraints on the maximum stock allocation.

\subsection{Our Contributions}
In this chapter, we consider the same scenario posed in \citet{Bengen1994}. Our objective function is to maximize total withdrawals over
a thirty year period (with minimum/maximum annual constraints) and  minimize the risk of running out of savings before year thirty.

An  optimal stochastic control problem used for decumulation is solved at each  rebalancing date, by alternatively  solving a linear partial-integro differential  equation (PIDE) followed by an optimization
step.   The  PIDE is solved by using a $\delta$-monotone Fourier method, which ensures that monotonicity holds to $O(\delta)$, an important property for control problems.  The  
$\delta-$monotone technique generalizes  the method described in \citet{forsythlabahn2017}. 
Our basic dynamic programming technique is similar to the high-level description in \citet{Forsyth_2021_b}, 
but here we delve much deeper in describing the numerical methods.  
In particular, we focus on such issues as wrap-around errors in 
Fourier methods for Partial-Integro Differential Equations (PIDEs), 
and the 2-d $\delta$-monotone Green's function approximation.

Previous studies for decumulation have imposed no-shorting and  no-leverage constraints  on stock investments.  In this chapter, 
we relax the no-leverage constraint  on the stock investment and also investigate 
more restrictive  investments in stocks.  For example, a risk averse retiree may wish to 
restrict the maximum fraction invested in equities  to be significantly less than one.

%
Rather surprisingly, we find that restricting the maximum fraction in stocks to be 50\% does not lower the efficient frontier by a large amount.
The most significant effect of the 50\% restriction, compared to allowing a maximum stock fraction of up to 130\% (leverage) is that the time 
for the median withdrawals to reach the maximum is delayed by one year.  This may be an acceptable trade-off for risk averse investors.

From the numerical point of view, when solving optimal control problems using
Fourier methods,
wrap-around error can be significant, due to the possible instantaneous
movement to  domain boundaries at reallocation 
times \citep{Lippa2013,Ignatieva_2018,alonso-garcca_wood_ziveyi_2018}.  We show that
our technique reduces wrap-around error to almost machine level precision.

\subsection{Organization of this Chapter}
The rest of this chapter is organized as follows.  In the next section we state our problem setting followed by a mathematical formulation of our model in  Section \ref{form_section}. The latter section includes information on the  stochastic processes followed by bonds and stocks along with a description of our control set. Section \ref{risk_section} describes our conflicting risk and reward measures, Expected Shortfall and Expected Withdrawals, while Section \ref{ew_es_section} looks at how one maximizes these conflicting measures, something which  is initially set up as a pre-commitment problem which is not implementable. Section \ref{DP_program} then discusses how one can formulate our problem into an (implementable)  dynamic program, one which is solved by a PIDE between withdrawal time intervals. Section \ref{delta_section} then shows how one can solve these PIDEs using Green's functions and Fourier analysis. The wrap-around problem that comes with using Fourier analysis and periodicity is addressed in the next section. That section also includes the resulting numerical algorithm for solving our optimal control decumulation problem.  Section \ref{example_section} then presents an example along with its numerical solution followed in the next section by tests of robustness of our numerical example. The paper ends with a conclusion along with an Appendix providing extra details regarding the wrap-around error discussed earlier.

\section{Problem Setting}\label{problem_section}

As noted in \citet{Anarkulova_2022_a}, retirees have a revealed preference for spending rules, such as that advocated by \citet{Bengen1994}.
Spending rules (such as the four per cent rule) are  popular with retirees as they are simple to implement, have an intuitive interpretation, and the rolling 30 year backtests
in \citet{Bengen1994} are easy to understand.

In this paper we restrict attention to the scenario outlined in 
\citet{Bengen1994}, that is, we consider a 65-year old retiree who desires fixed minimum annual (real) cash flows over a 30 year
time horizon.  We also impose an upper limit (a cap)  on maximum withdrawals in any year.  From the CPM2014 table from the Canadian Institute of
Actuaries\footnote{\url{www.cia-ica.ca/docs/default-source/2014/214013e.pdf}.}, the probability that a 65-year old Canadian male attains the age of 95 is about 0.13.  In spite of this relatively
low probability of achieving the age of 95, assuming a  30 year time horizon is considered a prudent approach for minimizing the risk of exhausting  savings. In addition, observe that we will not
mortality weight future cash flows, as is done when averaging over a population for pricing annuities.  Mortality weighting may not be meaningful for an individual investor, since it reflects
population averages, rather than the circumstances of an individual retiree.

Since we allow investing in risky assets, with a minimum cash withdrawal each year, it is possible to exhaust savings.\footnote{Ignoring the trivial and unlikely case where investing in riskless 
assets can fund 30 year minimum withdrawals.}    In this case, we continue to withdraw funds from the portfolio, which is equivalent to borrowing cash.   In such circumstances, we assume that 
withdrawals continue, something which can be viewed as borrowing against other resources. 
This debt accumulates at the borrowing rate.  Essentially, we are assuming that the investor has
other assets, for example real estate, which can be used as  a hedge of last resort.  In this case,  this debt could then be funded using a reverse mortgage, with real estate as collateral \citep{Pfeiffer_2013}.

It is important to note that an implicit assumption here is that real estate is not 
fungible with other financial assets, except as a last resort. 
This mental separation of assets is a common aspect of behavioral
finance \citep{Shefrin-Thaler:1988}.  Real estate in this case has a dual role: 
if market returns are exceptionally strong, or the retiree passes away earlier than expected, this property
can be a bequest.  If market returns are poor, or in the case of extreme longevity, then real estate can be used to fund expenses if the retirement account is exhausted.

\section{Mathematical Formulation}\label{form_section}

We assume that the investor has access to two funds: a broad market stock index fund and a constant maturity bond index fund.  Let $S(t)$ and $B(t)$ denote the real
(inflation adjusted) \emph{amounts} invested in the stock index and the bond index,  respectively,  at time $t$ with $T$ being the investment horizon.  These amounts will evolve over time, depending on the investor's  asset allocation, and
changes in the real unit prices of the assets. In the absence of an investor determined control  such as \ cash withdrawals or rebalancing, any changes in $S(t)$ and $B(t)$ result from asset price evolution. We model the stock index as following a jump diffusion, which permits modelling both continuous stochastic processes as well as market jumps.

We model the real returns of a constant maturity bond index as a stochastic process \citep[see for example, \@][]{Lin_2015,mitchell_2014}.  This avoids the need to model bond prices and inflation separately.  As done in \citet{mitchell_2014}, we assume that the constant maturity bond index follows a jump diffusion process, with justification being  found in   \citet[Appendix A]{Forsyth_Arva_2022}.

\subsection{The Diffusion Processes.}
For any time $t$ let
\begin{equation*}
S(t^-) \equiv \displaystyle \lim_{\epsilon \rightarrow 0^+}
          S(t- \epsilon) ~~\mbox{ and } ~~
 S(t^+) \equiv \displaystyle  \lim_{\epsilon \rightarrow 0^+}
          S(t + \epsilon)  ~
\end{equation*}
denote the values of $S$ the instant before $t^-$ and the instant after $t^+$   time $t$.  Such notation will also be  used for any other time dependent function.
We let $\xi^s$ be a random number representing a jump multiplier, that is,  when a jump occurs,  $S(t) = \xi^s S(t^-)$. The use of a jump process accounts for  non-normal asset returns. As in \citep{kou:2002,Kou2004}, we assume that $\log(\xi^s)$ follows a double exponential
distribution with $u^s$ the probability of an upward jump and $1-u^s$  the  probability of a downward jump. The density function for $y = \log \xi^s$ is
{\color{black}
\begin{equation}
f^s(y) = u^s \eta_1^s e^{-\eta_1^s y} {\bf{1}}_{y \geq 0} +
       (1-u^s) \eta_2^s e^{\eta_2^s y} {\bf{1}}_{y < 0} ~~~~;
         ~~~~~~{\bf{1}}_{\mathcal{C}} = \begin{cases}
                               1, & {\mathcal{C}} {\text{ is true }} \\
                               0, & {\text{ otherwise }}
                              \end{cases} ~.
\label{eq:dist_stock}
\end{equation}
}
We also define
{\color{black}
\begin{equation}
\gamma^s_{\xi} = E[ \xi^s -1 ] =
  \frac{u^s \eta_1^s}{\eta_1^s - 1} + 
  \frac{(1 - u^s)\eta_2^s}{\eta_2^s + 1} - 1  ~,
\end{equation}
}
where  $E[\cdot]$ is the expectation.
In the absence of control, $S(t)$ evolves according to
\begin{equation}
\frac{dS(t)}{S(t^-)} = 
  \left(\mu^s -\lambda_\xi^s \gamma_{\xi}^s \right) \, dt + 
  \sigma^s \, d Z^s +  d\left( \ds \sum_{i=1}^{\pi_t^s} (\xi_i^s -1) \right) ,
\label{jump_process_stock}
\end{equation}
where $\mu^s$ is the (uncompensated) drift rate, $\sigma^s$ is the volatility, $d Z^s$ is the increment of a Wiener process, $\pi_t^s$ is a Poisson process with positive intensity parameter $\lambda_\xi^s$,
and the $\xi_i^s$ are  i.i.d.\ positive random variables having distribution (\ref{eq:dist_stock}). Moreover, $\xi_i^s$, $\pi_t^s$, and $Z^s$ are assumed to all be mutually independent.

Similarly, let the amount in the bond index the instant before $t$ be $B(t^-)$. In the absence of investor intervention, $B(t)$ evolves as
\begin{equation}
\frac{dB(t)}{B(t^-)} = \left(\mu^b -\lambda_\xi^b \gamma_{\xi}^b  
   + \mu_c^b {\bf{1}}_{\{B(t^-) < 0\}}  \right) \, dt + 
  \sigma^b \, d Z^b +  d\left( \ds \sum_{i=1}^{\pi_t^b} (\xi_i^b -1) \right) ,
\label{jump_process_bond}
\end{equation}
where the terms in equation (\ref{jump_process_bond}) are defined
analogously to equation (\ref{jump_process_stock}). In particular,
$\pi_t^b$ is a Poisson process with positive intensity parameter
$\lambda_\xi^b$, and the density function  for $y= \log \xi^b$ is
\begin{linenomath*}
\begin{equation*}
f^b( y) = u^b \eta_1^b e^{-\eta_1^b y} {\bf{1}}_{y \geq 0} +
       (1-u^b) \eta_2^b e^{\eta_2^b y} {\bf{1}}_{y < 0}~,
\end{equation*}
\end{linenomath*}
and $\gamma_{\xi}^b = E[ \xi^b -1 ]$.  Again $\xi_i^b$, $\pi_t^b$, and
$Z^b$ are assumed to all be mutually independent. The term $\mu_c^b
{\bf{1}}_{\{B(t^-) < 0\}}$ in equation~(\ref{jump_process_bond})
represents the extra cost of borrowing, that is, the spread.

The diffusion processes are correlated, that is, \ $d Z^s \cdot d Z^b = \rho_{sb}~ dt$.   However, contrary to common belief, an analysis of historical data suggests that the stock and bond jump processes are essentially
uncorrelated (see \citet{forsyth_2020_a} for empirical justification).  We make this assumption in this work.

We define the investor's total wealth at time $t$ as
\begin{equation*}
\text{Total wealth } \equiv W(t) = S(t)+ B(t).
\end{equation*}
In case of insolvency,
the portfolio is liquidated, trading stops and any outstanding debt accrues interest 
at the borrowing rate.

\subsection{The Set of Controls.}
Consider a set of discrete times 
\begin{equation*}
  \mathcal{T} = \{t_0=0 <t_1 <t_2< \ldots <t_M=T\}  \label{T_def}
\end{equation*}
where we assume that $t_i - t_{i-1} = \Delta t =T/M$ is constant. 
Here $t_0=0$ is  the inception time of the investment and  $\mathcal{T} $ is the set of withdrawal/rebalancing times, as
defined in equation (\ref{T_def}). At each rebalancing time $t_i$, $i =
0, 1, \ldots, M-1$, the investor first withdraws an amount of cash $\qq_i$
from the portfolio, and then afterwards rebalances the portfolio. At $t_M =
T$ the portfolio is liquidated and no cash flow occurs.  This is enforced by specifying $\qq_M = 0$.

Let
\begin{equation*}
  W(t_i^-) = 
    S(t^-_i) + B(t_i^-) 
\end{equation*}
denote the instant before withdrawals and rebalancing at $t_i$,  we then have that   $W(t_i^+)$ is given by
\begin{equation*}
   W(t_i^+) =  W(t_i^-) 
              -\qq_i ~;~ i \in \mathcal{T} ~,  
\end{equation*}
with $W(t_M^+) =  W(t_M^-) $ since $\qq_M \equiv 0$.

Typically, DC plan savings are held in a tax-advantaged account, with
no taxes triggered by rebalancing. With infrequent (e.g.\ yearly)
rebalancing, we also expect other transaction costs, 
to be small, and hence can be ignored. It is possible to
include transaction costs, but at the expense of increased computational
cost \citep{Van2018}.

Let $X\left(t\right)= \left( S \left( t \right), B\left( t \right) \right)$,
$t\in\left[0,T\right]$ denote the multi-dimensional controlled underlying process and  $x = (s,
b)$ the realized state of the system.  The rebalancing control $\pp_i(\cdot)$ is the fraction invested
in the stock index at the rebalancing date $t_i$, that is, 
\begin{equation*}
  \pp_i \left(X(t_i^-)\right) = 
  \pp \left(X(t_i^-), t_i \right) = \frac{S(t_i^+)}{S(t_i^+) + B(t_i^+)}~.
\end{equation*}

Let 
$\qq_i(\cdot)$ be the amount withdrawn at
time $t_i$, that is, \ $\qq_i \left( X(t_i^-) \right) = \qq \left( X(t_i^-),
t_i \right)$. Formally, the controls depend on the state of the
investment portfolio, before the rebalancing occurs, that is, 
{\myred 
$\pp_i(\cdot) =  \pp\left(X(t_i^-), t_i \right)$ and  
$\qq_i(\cdot) =  \qq\left(X(t_i^-), t_i \right)$
} 
$t_i \in \mathcal{T}$, where $\mathcal{T}$ is the set of rebalancing
times.
However, it will be convenient to note that in our case, we find the
optimal control $\pp_i(\cdot)$ amongst all strategies with constant
wealth (after withdrawal of cash). Hence, with some abuse of notation,
we will now consider $\pp_i(\cdot)$ to be a function of wealth after
withdrawal of cash, where we use the shorthand notation $W_i^-$ and $W_i^+$ for the variables representing wealth the instant before and instant after $t_i$. Using a similar shorthand notation for $S_i$ and $B_i$ we then have
%
\begin{align*}
  \pp_i(\cdot) &= \pp(W(t_i^+), t_i) = \pp_i(W_i^+)  \nonumber \\
  S_i^+ &=  \pp_i(W_i^+)~ W_i^+ \nonumber \\
  B_i^+ &= (1 -\pp_i(W_i^+)) ~W_i^+ ~.  
\end{align*}
Note that the control for $\pp_i(\cdot)$ depends only on $W_i^+$. Since
$\pp_i(\cdot) = \pp_i(W_i^- - \qq_i)$,  it follows that
\begin{equation*}
  \qq_i(\cdot) = \qq_i(W_i^-)~, 
\end{equation*}
which we discuss further in Section \ref{DP_program}.

A control at time $t_i$ is then given by the pair $(\qq_i(\cdot),
\pp_i(\cdot))$ where the notation $(\cdot)$ denotes that the control
is a function of the state. Let $\mathcal{Z}$ represent the set of
admissible values of the controls $(\qq_i(\cdot), \pp_i(\cdot))$. We
impose no-shorting, bounded leverage constraints (assuming solvency) along with maximum and minimum values for the withdrawals. In addition we apply the
constraint that in the event of insolvency due to withdrawals ($W(t_i^+)
< 0$),  or in the case of leverage, trading ceases and debt (negative wealth) accumulates at the
appropriate borrowing rate of return (that is,  a spread over the bond rate).
We also specify that the stock assets are liquidated at $t=t_M$. 

We  can then  define our controls by 
{\myred 
\begin{align}
  \mathcal{Z}_{\qq}(W_i^-, t_i) &=
    \begin{cases}
      [\qq_{\min},\qq_{\max}] & t_i \in \mathcal{T} ~;~ t_i \neq t_M~;~ 
        W_i^- \geq \qq_{\max} \\
      [\qq_{\min},\max(\qq_{\min},W_i^-)] & t_i \in \mathcal{T} ~;~ 
        t \neq t_M ~;~ W_i^- < \qq_{\max} \\
      \{0\} & t_i = t_M
   \end{cases} ~,  
\label{Z_q_def} \\
  \mathcal{Z}_\pp (W_i^+,t_i) &=
    \begin{cases}
      [0, \pp_{\max}]  & ~ t_i \in \mathcal{T}~;~ t_i \neq t_M ~; ~W_i^+ > 0 \\
      \{0\} &~ t_i \in \mathcal{T}~;~  t_i \neq t_M~;~  W_i^+ \leq 0  \\
      \{0\} & t_i=t_M
    \end{cases} ~.  
\label{Z_p_def} 
\end{align}
} 
The rather complicated expression in equation (\ref{Z_q_def}) imposes
the assumption that as wealth becomes small, the retiree first tries
to avoid insolvency as much as possible and in the event of
insolvency, withdraws only $\qq_{\min}$.

The set of admissible values for $(\qq_i, \pp_i), t_i \in \mathcal{T}$,
can then be written as
\begin{equation}
  (\qq_i, \pp_i) \in \mathcal{Z}(W_i^-, W_i^+,t_i) = 
    \mathcal{Z}_{\qq}(W_i^-, t_i) \times \mathcal{Z}_{\pp} (W_i^+,t_i)~.
\label{admiss_set}
\end{equation}
For implementation purposes, we have written equation (\ref{admiss_set})
in terms of the wealth after withdrawal of cash. However, we remind the
reader that since $W_i^+ = W_i^- -\qq_i$, the controls are formally a
function of the state $X(t_i^-)$ before the control is applied.

The admissible control set $\mathcal{A}$ can then be written as
\begin{equation*}
  \mathcal{A} = \biggl\{
    (\qq_i, \pp_i)_{0 \leq i \leq M} : (\pp_i, \qq_i) \in 
    \mathcal{Z}(W_i^-, W_i^+,t_i) \biggr\}~
\end{equation*}
with an admissible control $\mathcal{P} \in \mathcal{A}$  written as
\begin{equation*}
  \mathcal{P} = \{(\qq_i(\cdot), \pp_i(\cdot)) ~:~ i=0, \ldots, M \} ~.
\end{equation*}
We also define $\mathcal{P}_n \equiv \mathcal{P}_{t_n} \subset
\mathcal{P}$ as the tail of the set of controls in $[t_n, t_{n+1},
\ldots, t_{M}]$, that is, 
\begin{equation*}
  \mathcal{P}_n =\{(\qq_n(\cdot), \pp_n(\cdot)), \ldots, 
                   (\qq_{M}(\cdot), \pp_{M}(\cdot)) \} ~
\end{equation*}
and  $\mathcal{A}_n$, the tail of the admissible
control set,   as
\begin{equation*}
  \mathcal{A}_n = \biggl\{
    (\qq_i, \pp_i)_{n \leq i \leq M} : 
    (\qq_i, \pp_i) \in \mathcal{Z}(W_i^-, W_i^+,t_i) \biggr\}~,
\end{equation*}
so that $\mathcal{P}_n \in \mathcal{A}_n$.

\section{Risk and Reward}\label{risk_section}

In this section we consider our measures of risk and reward for our retiree.  In this case the retiree
is primarily concerned with the risk of depleting savings while at the same time hoping to maximize 
the cash withdrawals from his plan.

\subsection{Risk: Expected Shortfall (ES)}

Two typical measures of financial tail risk of an investment portfolio  are Value at Risk (VAR) and Conditional Value of risk (CVAR), with the latter also known as Expected Shortfall.
Expected Shortfall  is an alternative to value at risk that is more sensitive to the shape of the tail of the loss distribution.

Suppose
\begin{linenomath*}
\begin{equation*}
\int_{-\infty}^{W^*_{\alpha}}  g(W_T) ~dW_T = \alpha,
\end{equation*}
\end{linenomath*}
where $g(W_T)$ be the probability density function of wealth $W_T$ at $t=T$.
{\myred 
Then $W^*_{\alpha}$ can be viewed as the Value at Risk (VAR) at level $\alpha$ since  the probability of wealth being more than $W^*_{\alpha}$ is  $1 - \alpha$. 
For example, if $\alpha = .01$, then 99\% of the outcomes have $W_T >
W^*_{\alpha}$. If $W^*_{\alpha}$ is sufficiently large, 
this suggests very low risk of running out of savings.
}  
The Expected Shortfall (ES)  at level $\alpha$ is then
\begin{equation}
  {\text{ES}}_{\alpha} = 
    \frac{1}{\alpha} \int_{-\infty}^{W^*_{\alpha}} W_T ~ g(W_T)~dW_T~,
\label{ES_def_1}
\end{equation}
which is the mean of the worst $\alpha$ fraction of outcomes. Typically
$\alpha \in \{ .01, .05 \}$. The definition of ES in equation
\eqref{ES_def_1} uses the probability density of the final wealth
distribution, not the density of \emph{loss}. Hence in our case, a
larger value of ES, that is,   a larger value of average worst case terminal
wealth,  is desired.  \footnote{ In
practice, the negative of $W^*_{\alpha}$ is often the reported VAR. and the negative of ES is commonly referred to as
Conditional Value at Risk (CVAR).}

Define 
$X_0^- = X(t_0^-)$. Given an expectation
under control $\mathcal{P}$, $E_{\mathcal{P}} [\cdot]$, then as noted by
\citet{Uryasev_2000}, the expected shortfall  at level $\alpha$ has the alternate formulation as 
\begin{equation}
  {\text{ES}}_{\alpha}( X_0^-, t_0^-) =  
    \sup_{W^*} E_{\mathcal{P}_0}^{X_0^-, t_0^-}
    \biggl[W^* + \frac{1}{\alpha} \min(  W_T - W^* , 0 ) \biggr] ~.
\label{ES_def}
\end{equation}
The admissible set for $W^*$ in equation (\ref{ES_def}) is over the set
of possible values for $W_T$.

The notation ${\text{ES}}_{\alpha}( X_0^-, t_0^-)$ emphasizes that
${\text{ES}}_{\alpha}$ is as seen at $(X_0^-, t_0^-)$, that is, the pre-commitment ${\text{ES}}_{\alpha}$. A strategy based
purely on optimizing the pre-commitment value of ${\text{ES}}_{\alpha}$
at time zero is \emph{time-inconsistent}. As such it  has been termed by many
as \emph{non-implementable}, since the investor has an incentive to
deviate from the time zero pre-commitment strategy at $t >0$. 
However  we  consider the pre-commitment strategy merely
as a device to determine an appropriate level of $W^*$ in equation
(\ref{ES_def}). If we fix $W^*$ for all positive $t$, then this strategy is
the induced time-consistent strategy \citep{Strub_2019_a,forsyth_2019_c,Strub_2022} and hence
is implementable. For further discussion of the relationship 
between time consistent and pre-commitment strategies,
see \citep{vigna:2014,Vigna_2017b,Vigna2017,Strub_2019_a,forsyth_2019_c,bjork_book_2021,Strub_2022}. 
In particular, see \citet{forsyth_2019_c} for discussion of
the induced time consistent policy resulting from the 
use of ES risk.

An alternative measure of risk could be based on variability of
withdrawals \citep{Forsyth_Arva_a}. However, we note that we have
constraints on the minimum and maximum withdrawals, so that variability
is mitigated. We also assume that given these constraints, the retiree
is primarily concerned with the risk of depleting savings, something which is well
measured by ES.

\subsection{Reward: Expected Total Withdrawals (EW)}
We will use expected total withdrawals as a measure of reward in the
following. More precisely, we define EW (expected withdrawals) as
{\myred 
\begin{equation}
  {\text{EW}}( X_0^-, t_0^-) = E_{\mathcal{P}_0}^{X_0^+, t_0^+} 
    \biggl[\sum_{i=0}^{M} \qq_i \biggr]~,  
\label{EW_def}
\end{equation}
where we assume that the investor survives for the entire decumulation
period. This is  consistent with the scenario in \citet{Bengen1994} .
} 

We remark that there is no discounting term in equation (\ref{EW_def}) as  all quantities are real, that is, inflation-adjusted. It is
straightforward to introduce discounting, but we view setting the real
discount rate to zero to be a reasonable and conservative choice. See
\citet{Forsyth_2021_b} for further comments.

\section{ Maximizing Conflicting Measures: Problem EW-ES}\label{ew_es_section}
Expected withdrawals (EW) and expected shortfall (ES) are two  conflicting measures.  We handle this by using a scalarization technique to find the
Pareto points for this bi-objective optimization problem, that is, the set of points where one objective cannot be improved without worsening the other.  
For a given scalarization parameter $\kappa > 0$, the goal is  to then find the
control $\mathcal{P}_0$ that maximizes
\begin{equation*}
  \text{EW}( X_0^-, t_0^-) + \kappa~\text{ES}_{\alpha}( X_0^-, t_0^-)~.
\end{equation*}

More precisely, we define the pre-commitment EW-ES problem
$(PCES_{t_0}(\kappa))$ problem in terms of the value function
$J(s, b , t_0^-)$ derived from (\ref{ES_def}) and (\ref{EW_def}):

\begin{align}
J\left(s, b,  t_{0}^-\right)  
  & = \sup_{\mathcal{P}_{0}\in\mathcal{A}} \sup_{W^*}
     \Biggl\{
       E_{\mathcal{P}_{0}}^{X_0^-, t_{0}^-}
     \Biggl[
           \sum_{i=0}^{M} \qq_i +
            \kappa\biggl(W^* + \frac{1}{\alpha} \min (W_T -W^*, 0) \biggr)
             \Biggr. \Biggr. \nonumber \\
     & \Biggl. \Biggl.
     ~~~~~~~~~~~~~~~~~~~~~~ +~\epsilon W_T 
     \bigg\vert X(t_0^-) = (s, b) 
    ~\Biggr] 
   \Biggr\} \label{PCES_a}\\
\text{ subject to } &
\begin{cases}
   (S(t), B(t)) \text{ follows processes \eqref{jump_process_stock} and 
                    \eqref{jump_process_bond}};  
                     ~~t \notin \mathcal{T} \\
  W_{\ell}^+ = W_{\ell}^{-} - \qq_\ell \,; 
              ~X_\ell^+ = (S_\ell^+ , B_\ell^+) \\
  W_{\ell}^- = \biggl(S_\ell^- + B_\ell^- \biggr)  \\
  S_\ell^+ = \pp_\ell(\cdot) W_\ell^+ \,; 
             ~B_\ell^+ = (1 - \pp_\ell(\cdot) ) W_\ell^+ \, \\
  (\qq_\ell(\cdot) , \pp_\ell(\cdot)) \in 
    \mathcal{Z}(W_\ell^-, W_\ell^+,t_\ell) \\
  \ell = 0, \ldots, M ~;~ t_\ell \in \mathcal{T} \\
\end{cases}~.
\label{PCES_b}
\end{align}

Note that we have added an extra term
$E_{\mathcal{P}_{0}}^{X_0^-,t_{0}^-}[ \epsilon W_T ]$ to equation
(\ref{PCES_a}). If we have a maximum withdrawal constraint, and if $W_t
\gg W^*$ as $t \rightarrow T$, then the controls become ill-posed. In this
fortunate state for the investor, we can break investment policy ties
either by setting $\epsilon < 0$, which will force investments in bonds,
or by setting $\epsilon > 0$, which will force investments into stocks.
Choosing $ |\epsilon| \ll 1$ ensures that this term only has an effect
if $W_t \gg W^*$ and $t \rightarrow T$. See \citet{Forsyth_2021_b} for
more discussion of this.

We can interchange the $\sup \sup (\cdot) $ 
\footnote{
  Let $F = \sup_{ (a,b) \in A \times B} f(a, b)$, then $\forall \epsilon >0$,
$\exists (a^*,b^*) \in  A \times B$, s.t.
$ f(a^*,b^*) > F - \epsilon$.
Then $F \ge \sup_{a \in A} \sup_{b \in B} f(a, b)
        \ge \sup_{b \in b} f(a^*, b) 
        \ge   f(a^*,b^*) > F - \epsilon$. Hence, $\epsilon \rightarrow 0$
implies $\sup_{a \in A} \sup_{b \in B} f(a, b) = F$.
Similarly $\sup_{b \in B} \sup_{a \in A} f(a, b) = F$.
} 
in equation (\ref{PCES_a}) to represent the 
value function as 

\begin{align}
\qquad J\left(s, b,  t_{0}^-\right) &=  
  \sup_{W^*} \sup_{\mathcal{P}_{0}\in\mathcal{A}} \Biggl\{
  E_{\mathcal{P}_{0}}^{X_0^+,t_{0}^+}
  \Biggl[\sum_{i=0}^{M} \mathfrak{q}_i +
  \kappa \biggl( W^* + \frac{1}{\alpha} \min (W_T -W^*, 0) \biggr)
  \Biggr. \Biggr.  \nonumber \\
   &  \Biggl. \Biggl. ~~~~~~~~~~~~~~~~~~~~~~ +~\epsilon W_T 
     \bigg\vert X(t_0^-) = (s, b) \Biggr] \Biggr\} .  
\label{pcee_inter}
\end{align}
Since the inner supremum in equation (\ref{pcee_inter}) is a
continuous function of $W^*$ and  the optimal value of $W^*$ in
equation (\ref{pcee_inter}) is bounded, \footnote{This is the same as
noting that a finite value at risk exists.} we can define
\begin{align}
\mathcal{W}^*(s,b) &=\displaystyle \argmax_{W^*} \biggl\{
                      \sup_{\mathcal{P}_{0}\in\mathcal{A}} 
  \Biggl\{ E_{\mathcal{P}_{0}}^{X_0^+,t_{0}^+}
  \Biggl[ ~\sum_{i=0}^{M} \mathfrak{q}_i  + 
  \kappa \biggl( W^* + \frac{1}{\alpha} \min (W_T -W^*, 0) \biggr) 
  \biggr. \biggr.  \nonumber \\
    & \biggl. \biggl. ~~~~~~~~~~~~~~~~~~~~~~~~~~~~~ +~ \epsilon W_T 
     \bigg\vert X(t_0^-) = (s, b)\Biggr]\Biggr\} ~.
\label{pcee_argmax}
\end{align}
Given $\mathcal{W}^*(s,b)$ from equation (\ref{pcee_argmax}), then the optimal control $\mathcal{P}^*$ which solves
Problem (\ref{pcee_inter}) is the same control which 
solves 
\begin{align}
\qquad \hat{J}\left(s, b,  t_{0}^-\right) &=
   \sup_{\mathcal{P}_{0}\in\mathcal{A}} \Biggl\{
  E_{\mathcal{P}_{0}}^{X_0^+,t_{0}^+}
  \Biggl[\sum_{i=0}^{M} \mathfrak{q}_i +
    \frac{\kappa}{\alpha} \min (W_T - \mathcal{W}^*(s,b), 0) 
  \Biggr. \Biggr.  \nonumber \\
   &~~~~~~~~~~~ ~~~~~~~ ~~~~~~~~\Biggl. \Biggl. +~\epsilon W_T
     \bigg\vert X(t_0^-) = (s, b) \Biggr] \Biggr\} .
\label{pcee_inter_consistent}
\end{align}
Hence Problem (\ref{pcee_inter_consistent}) is the induced time consistent
policy for Problem (\ref{pcee_inter}).
We refer the reader to \citet{forsyth_2019_c} for an extensive discussion
concerning pre-commitment and time consistent expected shortfall strategies.

\section{Formulation as a Dynamic Program}\label{DP_program}

We can use the method in \citet{forsyth_2019_c} to  determine our value function. 
Write  (\ref{pcee_inter}) as
\begin{equation*}
   J(s, b, t_0^-) = \sup_{W^*} \tildev(s, b, W^*, 0^-)~,
\end{equation*}
where $\tildev(s, b, W^*, t)$ is defined as
\begin{align}
  \tildev(s, b, W^*, t_n^-) &= \sup_{\mathcal{P}_{n}\in\mathcal{A}_n}
    \Biggl\{ E_{\mathcal{P}_{n}}^{\hat{X}_n^+,t_{n}^+}
    \Biggl[\sum_{i=n}^{M} \qq_i + \kappa \biggl(
    W^* + \frac{1}{\alpha} \min((W_T -W^*),0) \biggr) 
       \Biggr. \Biggr.  \nonumber \\ 
     &~~~~~~~~~~~~~~~~~  ~~~~~~~~~~~~
              +~\epsilon W_T   ~\bigg\vert ~ \hat{X}(t_n^-) = (s, b, W^*) \Biggr] \Biggr\}
\label{expanded_1} \\
\text{ subject to } &
\begin{cases}
 (S(t), B(t)) \text{ follows processes \eqref{jump_process_stock} and 
                  \eqref{jump_process_bond}};  ~~t \notin \mathcal{T} \\
  W_{\ell}^+ = W_{\ell}^- - \qq_\ell \,; 
             ~ X_\ell^+ = (S_\ell^+ , B_\ell^+, W^*) \\
  W_{\ell}^- = \biggl(S_\ell^- + B_\ell^- \biggr) 
               \\
  S_\ell^+ = \pp_\ell(\cdot) W_\ell^+ \,; 
  ~B_\ell^+ = (1 - \pp_\ell(\cdot) ) W_\ell^+ \, \\
  (\qq_\ell(\cdot), \pp_\ell(\cdot)) \in 
     \mathcal{Z}(W_\ell^-, W_\ell^+ , t_\ell) \\
  \ell = n, \ldots, M ~;~ t_\ell \in \mathcal{T}
\end{cases}  ~.
\end{align}

The original problem (\ref{pcee_inter}) has therefore been decomposed  into two
steps:
\begin{enumerate}
\item For given initial cash $W_0$, and a fixed value of $W^*$, solve
problem (\ref{expanded_1}) using dynamic programming in order to determine
$\tildev(0,W_0, W^*, 0^-)$.
\item   Solve problem (\ref{pcee_inter}) by maximizing over $W^*$
\begin{equation}
  J(0,W_0, 0^-) = \sup_{W^*} v(0,W_0, W^*, 0^-) ~.
\label{final_step_EWES}
\end{equation}
\end{enumerate}

\subsection{Dynamic Programming Solution of Problem (\ref{expanded_1})}\label{dp_section}

We give a brief overview of the method described in detail in
\citep{Forsyth_2021_b}.  We apply the dynamic programming principle to $t_n
\in \mathcal{T}$
\begin{align}
  \tildev(s, b, W^*, t_n^-) &= 
    \sup_{\qq \in \mathcal{Z}_\qq(w^-,t_n)} \biggl\{
    \sup_{\pp \in \mathcal{Z}_\pp(w^+ , t_n)} 
    \biggl[\qq+ \tildev( w^+ \pp,  w^+
    (1-\pp) , W^*, t_n^+) \biggr] ~\biggr\} 
   \nonumber \\
 & \quad \mbox{ where } w^- = (s+b) ~; ~ w^+ = w^- - \qq ~.   
\label{dynamic_a}
\end{align}
If we  set
\begin{equation}
h(w, t_n, W^*) = 
  \biggl[\sup_{\pp \in \mathcal{Z}_\pp( w, t_n)}
         \tildev(w \pp , w(1-\pp) , W^*, t_n^+) \biggr] ~
\label{dynamic_b}
\end{equation}
then  equation (\ref{dynamic_a})  becomes
\begin{align}
\tildev(s, b, W^*, t_n^-) &= 
  \sup_{\qq \in \mathcal{Z}_\qq(w^-,t_n)}
  \biggl\{ \qq + \biggl[h( (w^- -\qq) , W^*, t_n^+)
                                     \biggr] \biggr\} \nonumber \\
 & \quad \mbox{ with } w^- = (s+b) .
\label{dynamic_c}
\end{align}
This approach effectively replaces a two dimensional optimization for
$(\qq_n, \pp_n)$, by two sequential one dimensional optimizations.
From equations (\ref{dynamic_b}) and (\ref{dynamic_c}), the
optimal pair $(\qq_n, \pp_n)$ satisfies
\begin{align*}
\qq_n &= \qq_n(w^- , W^*)  \quad ~  \mbox{ where } 
   \quad w^- = (s+b)   
     \nonumber \\
\pp_n &= \pp_n( w, W^*)  \quad \quad \mbox{ where } 
   \quad w = w^- - \qq_n ~.
\end{align*}
In other words, the optimal withdrawal control $ \qq_n$ is only
a function of total wealth before
withdrawals while the optimal control $\pp_n$ is a function only of total
wealth after withdrawals.  If a withdrawal results in $W^+ >0$, then
the optimal control for $\pp$ is determined along lines
of constant wealth in the $(s, b)$ plane while if a withdrawal results
in $W^+ < 0$, then stocks are liquidated ($\pp = 0$).
This is illustrated 
in Figure \ref{insolvent_fig}.

\begin{figure}[htb!]
\centerline{%
\includegraphics[width=4.0in]{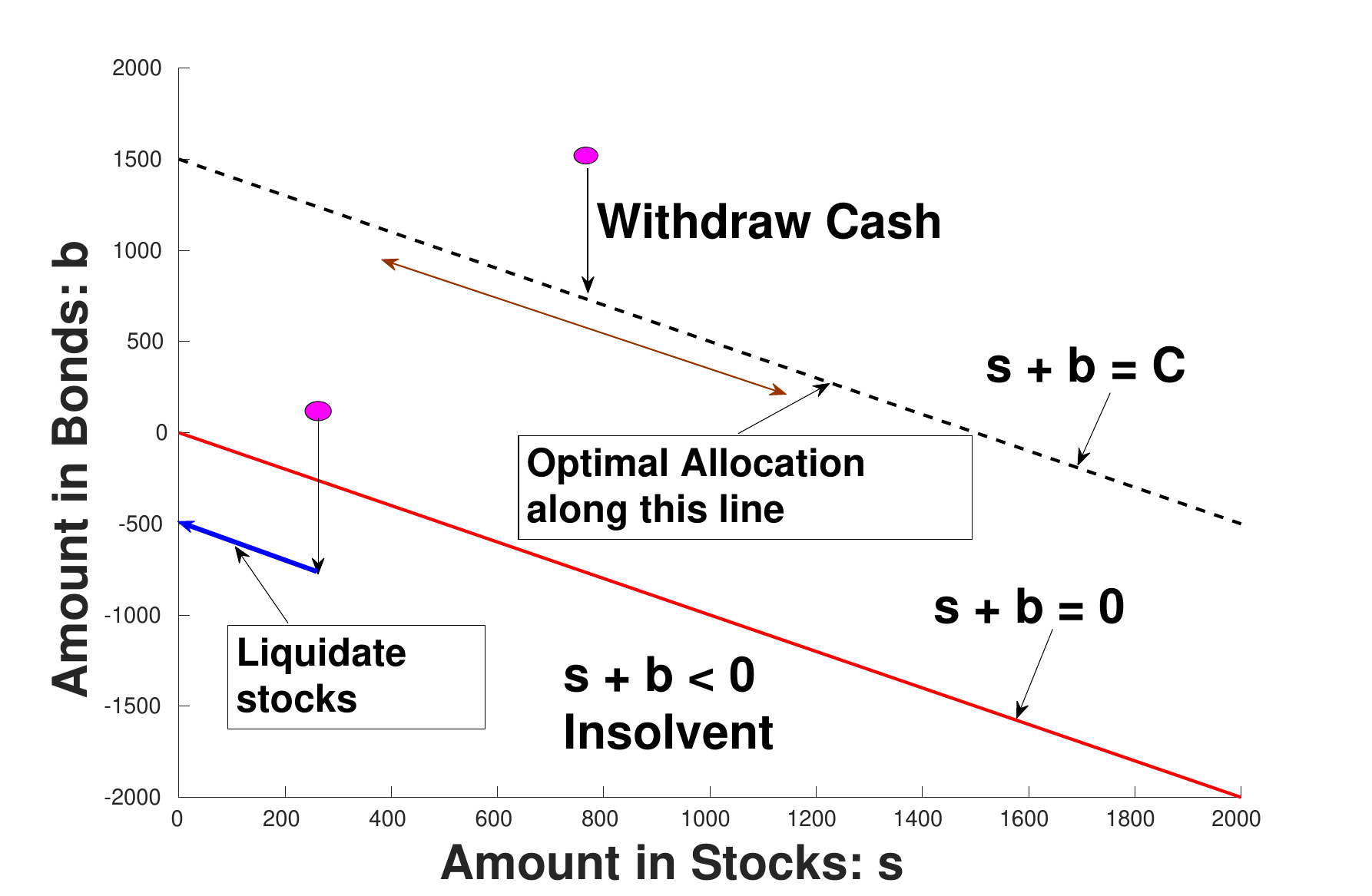}
}
\caption{
Schematic of withdrawal controls.
}
\label{insolvent_fig}
\end{figure}

At $t=T$, we have
\begin{equation*}
\tildev(s, b, W^*, T^+) =  \kappa \biggl( 
  W^* + \frac{\min( (s+b -W^*), 0 )}{\alpha} 
  \biggr) + \epsilon (s+b)~,
\end{equation*}
while at points in between rebalancing times, that is when \ $t
\notin \mathcal{T}$,  standard arguments from SDEs
(\ref{jump_process_stock}-\ref{jump_process_bond}), and
\citet{Forsyth_2021_b} give
\begin{align}
\tildev_t &+ \frac{(\sigma^s)^2 s^2}{2} \tildev_{ss} 
     +(\mu^s - \lambda_{\xi}^s \gamma_{\xi}^s) s \tildev_s
     + \lambda_{\xi}^s \int_{-\infty}^{+\infty} \tildev( e^ys, b, t) f^s(y)~dy 
     + \frac{(\sigma^b)^2 b^2}{2} \tildev_{bb} \nonumber \\
    &+ (\mu^b + \mu_c^b {\bf{1}}_{\{ b < 0 \}} 
     - \lambda_{\xi}^b \gamma_{\xi}^b) b \tildev_b       
     + \lambda_{\xi}^b \int_{-\infty}^{+\infty} \tildev( s, e^yb, t) f^b(y)~dy 
     -(\lambda_{\xi}^s + \lambda_{\xi}^b ) \tildev \nonumber \\
    & + \rho_{sb} \sigma^s \sigma^b s b \tildev_{sb} = 0~.
\label{expanded_7}
\end{align}

It is convenient to consider the two cases $b\geq 0$ and $b \leq 0$ separately.
For $b\geq 0$, we solve  
\begin{align}
\tildev_t &+ \frac{(\sigma^s)^2 s^2}{2} \tildev_{ss}
     +(\mu^s - \lambda_{\xi}^s \gamma_{\xi}^s) s \tildev_s
     + \lambda_{\xi}^s \int_{-\infty}^{+\infty} \tildev( e^ys, b, t) f^s(y)~dy
     + \frac{(\sigma^b)^2 b^2}{2} \tildev_{bb} \nonumber \\
    &+ (\mu^b  
     - \lambda_{\xi}^b \gamma_{\xi}^b) b \tildev_b
     + \lambda_{\xi}^b \int_{-\infty}^{+\infty} \tildev( s, e^yb, t) f^b(y)~dy
     -(\lambda_{\xi}^s + \lambda_{\xi}^b ) \tildev \nonumber \\
     & + \rho_{sb} \sigma^s \sigma^b s b \tildev_{sb} = 0~, \nonumber \\
       &  ~~~~~~~~~~~~~~~~~~~~~~~~~~~~~~~~~~~~~~~~~~~~~~~~~ b \geq 0, s \geq0~.
\label{expanded_7a}
\end{align}
When $b <0$, it is convenient  to re-write equation
(\ref{expanded_7}) in terms of debt $\hat{b} = -b$.
Letting
$\hat{v} (s, \hat{b}, t) = \tildev(s, b, t), b < 0, \hat{b} = -b$ in equation
(\ref{expanded_7}) we obtain
\begin{align}
\hat{v}_t &+ \frac{(\sigma^s)^2 s^2}{2} \hat{v}_{ss}
     +(\mu^s - \lambda_{\xi}^s \gamma_{\xi}^s) s \hat{v}_s
     + \lambda_{\xi}^s \int_{-\infty}^{+\infty} \hat{v}( e^y s, \hat{b}, t) f^s(y)~dy
     + \frac{(\sigma^b)^2 \hat{b}^2}{2} \hat{v}_{\hat{b} \hat{b} } \nonumber \\
    &+ (\mu^b + \mu_c^b
     - \lambda_{\xi}^b \gamma_{\xi}^b) \hat{b} \hat{v}_{\hat{b}}
     + \lambda_{\xi}^b \int_{-\infty}^{+\infty} \hat{v}( s, e^y \hat{b}, t) f^b(y)~dy
     -(\lambda_{\xi}^s + \lambda_{\xi}^b ) \hat{v} \nonumber \\
     &+ \rho_{sb} \sigma^s \sigma^b s \hat{b} \hat{v}_{s \hat{b} } = 0~,
               \nonumber \\
      & ~~~~~~~~~~~~~~~~~~~~~~~~~~~~~~~~~~~~~~~~~~~~~~~~~ \hat{b} >   0, s \geq 0~.
\label{minus_b_2}
\end{align}

Note that equation (\ref{minus_b_2}) is now amenable to a
transformation of the form $\hat{x} = \log \hat{b}$ since
$\hat{b} > 0$, something required when using a Fourier method
\citep{forsythlabahn2017,Forsyth_2021_b} to solve equation
(\ref{minus_b_2}).

After rebalancing, if $b \ge 0$, then $b$ cannot become negative, since
$b=0$ is a barrier in equation (\ref{expanded_7}). However, $b$ can
become negative after withdrawals, in which case $b$ remains in the
state $b<0$. There  equation (\ref{minus_b_2}) applies, unless there
is an injection of cash to move to a state with $b>0$. The terminal
condition for equation (\ref{minus_b_2}) is
\begin{equation*}
\hat{v}(s, \hat{b}, W^*, T^+) = \kappa \biggl( 
  W^* + \frac{\min( (s  -\hat{b} -W^*), 0 )}{\alpha} 
  \biggr) + \epsilon (s -\hat{b})  ~;~ \hat{b} > 0~. 
\end{equation*}

\section{The $\delta$-Monotone Fourier Method}\label{delta_section}

The $\delta$-monotone Fourier method originated with the work of  \citet{forsythlabahn2017}, where it was applied to 
one dimensional problems in  stochastic control. In the following,
we give a sketch showing how to extend those methods to work for  two dimensional
problems. 

We illustrate this method by focusing on equation (\ref{expanded_7a}).
Since $b > 0, s > 0$ we can let $x_1  =  \log s,   x_2 = \log b, \tau  =  T-t$ with $
   \plainv( x_1, x_2, \tau)  =  \tildev( e^{x_1}, e^{x_2}, T - \tau)$. 
   Then (\ref{expanded_7a}) converts to
\begin{align}
\plainv_{\tau}  &=  \frac{(\sigma^s)^2 }{2} \plainv_{x_1 x_1 }
     +(\mu^s - \lambda_{\xi}^s \gamma_{\xi}^s)  \plainv_{x_1}
     + \lambda_{\xi}^s \int_{-\infty}^{\infty} v( x_1 + y , x_2, \tau) f^s(y)~dy
     + \frac{(\sigma^b)^2 }{2} \plainv_{x_2 x_2} \nonumber \\
    &+ (\mu^b 
     - \lambda_{\xi}^b \gamma_{\xi}^b)  \plainv_{x_2}
     + \lambda_{\xi}^b \int_{-\infty}^{\infty} \plainv( x_1, x_2 + y, t) f^b(y)~dy
     -(\lambda_{\xi}^s + \lambda_{\xi}^b )\plainv 
     + \rho_{s b} \sigma^s \sigma^b  \plainv_{x_1 x_2} . 
\label{expanded_log}
\end{align}

The exact solution of equation (\ref{expanded_log}) can be written as
\begin{eqnarray}
     \plainv(x_1, x_2, \tau + \Delta \tau ) & = & \int_{-\infty}^{\infty} \int_{-\infty}^{\infty} 
                      g( x_1 - \zz, x_2 - \zzz, \Delta \tau) ~\plainv( \zz, \zzz, \tau)
                         ~d \zz~d\zzz~,
     \label{green_1}
\end{eqnarray}
where $g(x_1, x_2, \Delta \tau)$ is the Green's function of equation (\ref{expanded_7a}) \citep{menaldi_1992}.

\begin{remark}[Form of Green's function]\label{diff_remark}
Note that the Green's function for equation (\ref{expanded_7a}) is of the form
$g(x_1, x_2, \zz, \zzz) = g( x_1 - \zz, x_2 - \zzz)$, that is a function of the differences.
Intuitively, we can view the Green's function as a scaled
probability density $f$, and this means that  $f(x_1, x_2) | (\zz, \zzz)$
depends only on the difference, $f(x_1, x_2) | (\zz, \zzz)
= f(  x_1 - \zz, x_2 - \zzz)$.  See \citet{forsythlabahn2017} for more discussion of this observation.
\end{remark}

The Fourier transform pair for the Green's  function is given as 
\begin{eqnarray}
    G( \omega_1, \omega_2, \Delta \tau) & = & \int_{-\infty}^{\infty} \int_{-\infty}^{\infty} 
                     g( x_1 , x_2 , \Delta \tau) e^{ -2 \pi i w_1 x_1} e^{ -2 \pi i w_2 x_2} 
                         ~d x_1 ~dx_2~  \nonumber \\
    g( x_1, x_2, \Delta \tau) & = &  \int_{-\infty}^{\infty} \int_{-\infty}^{\infty} 
                     G( \omega_1, \omega_2, \Delta \tau) e^{ 2 \pi i w_1 x_1} e^{ 2 \pi i w_2 x_2}
                         ~d \omega_1 ~d \omega_2~.
              \label{FT_b} 
\end{eqnarray}
where $i = \sqrt{ -1}$.
Standard techniques then give the Fourier Transform of the Green's function $G(\omega_1, \omega_2, \Delta \tau)$
for equation (\ref{expanded_log}) as
\begin{eqnarray}
    G( \omega_1, \omega_2, \Delta \tau) & = & e^{ \Psi( \omega_1, \omega_2) \Delta \tau}~, \label{Green_ft_1}
\end{eqnarray}
where
\begin{eqnarray}
   \Psi( \omega_1, \omega_2) & = &  -\frac{(\sigma^s)^2 }{2} (2 \pi \omega_1)^2 
                                     +\left(\mu^s - \lambda_{\xi}^s \gamma_{\xi}^s - \frac{(\sigma^s)^2 }{2} \right) (2 \pi i \omega_1)
                                     + \lambda_{\xi}^s \overline{F^s}(\omega_1) \nonumber \\
             & & -\frac{(\sigma^b)^2 }{2} (2 \pi \omega_2)^2 
                                     +\left(\mu^b - \lambda_{\xi}^b \gamma_{\xi}^b - \frac{(\sigma^b)^2 }{2} \right) (2 \pi i \omega_2)
                                     + \lambda_{\xi}^b \overline{F^b}(\omega_1) \nonumber \\
             & &                - ( \lambda_{\xi}^s +  \lambda_{\xi}^b) 
                                - \rho_{s b} \sigma^s \sigma^b  (2 \pi \omega_1 ) (2 \pi \omega_2). \label{Green_ft_2}
\end{eqnarray}
Here $\overline{F^b}(\omega_1), \overline{F^s}(\omega_1)$ are the complex conjugates of the Fourier transforms of the density functions $f^s, f^b$:
\begin{eqnarray*}
    \overline{F^s}(\omega_1) & = & \frac{ u^s}{ 1 - \frac{ 2 \pi i \omega_1}{\eta_1^s}  }
               + \frac{ 1 - u^s}{ 1 + \frac{2 \pi i \omega_1}{ \eta_2^s} } \nonumber \\
    \overline{F^b}(\omega_2) & = & \frac{ u^b}{ 1 - \frac{ 2 \pi i \omega_2}{\eta_1^b}  }
               + \frac{ 1 - u^b}{ 1 + \frac{2 \pi i \omega_2}{ \eta_2^b} } ~.
\end{eqnarray*}

\subsection{Localization}
Define
\begin{eqnarray*}
   \Omega & = & [(x_1)_{\min}, (x_1)_{\max}] \times [(x_2)_{\min}, (x_2)_{\max}] ~.
\end{eqnarray*}
As is typically the case, we assume that the Green's function $g(x_1, x_2,  \Delta \tau)$
decays to zero as  $|x_1|, |x_2| \rightarrow \infty$.  More
precisely, we assume that 

\begin{assumption}[Decay of Green's Function]\label{green_decay_assumption}
For some $A >0$, $g(x_1, x_2, \Delta \tau)$ is negligible if $\min( |x_1|, |x_2|) > A$.
\end{assumption}
We choose our region $\Omega$ so that we can assume that the Green's function
is negligible for $(x_1,x_2) \notin \Omega$. As such we  replace the Fourier transform pair (\ref{FT_b}) by their Fourier series equivalent
\begin{eqnarray}
   G( \omega_1^k, \omega_2^j, \Delta \tau) & \simeq & \int_{(x_1)_{\min}}^{(x_1)_{\max}}  \int_{(x_2)_{\min}}^{(x_2)_{\max}}
                    g( x_1, x_2,  \Delta \tau) ~e^{ - 2 \pi i \omega_1^k x_1 }  ~e^{ - 2 \pi i \omega_2^j x_2 }~dx_1  ~dx_2
                                     \nonumber \\
   g(x_1, x_2, \Delta \tau)  & \simeq &  
      \frac{1}{P_1 P_2} \sum_{k=-\infty}^{\infty}  \sum_{j=-\infty}^{\infty}
                                        G(\omega_1^k, \omega_2^j, \Delta \tau) ~e^{ 2 \pi i \omega_1^k x_1} 
                                                           ~e^{ 2 \pi i \omega_2^j x_2 }  ~ \label{series_1}
                                                           \end{eqnarray}
                                                           with
                                                           \begin{eqnarray*}
              P_1 = (x_1)_{\max} - (x_1)_{\min} ~;~P_2 & = &  (x_2)_{\max} - (x_2)_{\min} ;  \quad 
               \omega_1^k = \frac{k}{P_1} ~;~ \omega_2^j = \frac{j}{P_2}
\end{eqnarray*}
along with a localized version of equation (\ref{green_1})
\begin{eqnarray}
     \plainv(x_1,x_2, \tau + \Delta \tau ) & \simeq & \int_{(x_1)_{\min}}^{(x_1)_{\max}}  \int_{(x_2)_{\min}}^{(x_2)_{\max}}
                      g( x_1 - \zz, x_2 - \zzz, \Delta \tau) ~\plainv( \zz, \zzz, \tau)
                         ~d \zz~d\zzz~.
               \nonumber \\
     \label{green_local}
\end{eqnarray}
Scaling factors in equation (\ref{series_1}) are selected to correspond to a
finite domain form of equation (\ref{FT_b}).

\subsection{Periodic Extension}
We have informally derived the expression for the Green's function by localizing
the infinite domain Green's function.
Localizing the problem on $\Omega$, and using a Fourier series representation of the Green's function
makes the implicit assumption of periodic extension.
{\nolinenumbers 
\begin{assumption}{Periodicity} \label{periodic_assumption}
 \begin{enumerate}
     \item The control problem is defined on the finite domain $\Omega$.
     \item The solution $\plainv(x_1, x_2, \tau)$ is extended periodically
                  \begin{eqnarray}
                    \plainv(x_1 \pm P_1, x_2 , \tau) = \plainv(x_1 , x_2 \pm P_2, \tau)
                     & = &\plainv(x_1, x_2, \tau)
                          ~. \label{periodic_1}
                  \end{eqnarray}
    \item The jump size density functions $f^s(y_1), f^b(y_2)$ are defined on
                 $y_1 \in [(x_1)_{\min}, (x_1)_{\max}]$ and $ y_2 \in [(x_2)_{\min}, (x_2)_{\max}]$,
                 with periodic extension
                    \begin{eqnarray}
                    f^s(y_1 \pm P_1) & = & f^s( y_1) ~~\mbox{ and } 
                   ~~f^b( y_2 \pm P_2)  =  f^b( y_2) ~.   \label{periodic_2}
                  \end{eqnarray}
   \item Periodic boundary conditions (\ref{periodic_1})-\ref{periodic_2}) are specified
               for the PIDE problem (\ref{expanded_log}).
   \item From equation (\ref{series_1}) we can also see that
         \begin{eqnarray*}
                    g(x_1 \pm P_1, x_2 , \tau) = g(x_1 , x_2 \pm P_2, \tau)
                     & = & g(x_1, x_2, \tau).
                  \end{eqnarray*}
  \end{enumerate}
\end{assumption}
}  
It is more rigorous to make  Assumption \ref{periodic_assumption}, 
and then derive the Green's function.  It is straightforward to verify
that this yields the Fourier series representation (\ref{series_1}), with
$G(\omega_1^k, \omega_2^j)$ given by equations (\ref{Green_ft_1})-(\ref{Green_ft_2}).
In particular, we have that \citep{menaldi_1992}
\begin{eqnarray*}
    g(x_1, x_2, \Delta \tau) & \geq & 0 \nonumber \\
    \int_{\Omega} g(x_1, x_2, \Delta \tau) dx_1 ~ dx_2 & = & 1
         ~.
\end{eqnarray*}

\subsection{Discretization}
We discretize equation (\ref{green_local}) on a grid $\{ (x_1, x_2)_{j,k} \}, \{ (\zz, \zzz)_{j,k} \}$, by setting
\begin{eqnarray*}
   x_1^j = \hat{x}_1^0 + j \Delta x_1 & ; &  x_2^k = \hat{x}_2^0 + k \Delta x_2~;~
              j= -\frac{N_1}{2}, \ldots \frac{N_1}{2} -1 ~;~  k= -\frac{N_2}{2}, \ldots \frac{N_2}{2} -1 \nonumber \\
     & & \Delta x_1 = \frac{P_1}{N_1} ~;~ \Delta x_2 = \frac{P_2}{N_2} \nonumber \\
     & & P_1 = (x_1)_{\max} - (x_1)_{\min} ~;~P_2 = (x_2)_{\max} - (x_2)_{\min} \nonumber \\
     & &  \plainv_{j,k}(\tau)  \equiv \plainv(x_1^j, x_2^k, \tau) .
\end{eqnarray*}
Define linear basis functions as
\begin{eqnarray*}
      \phi_1^j(x_1) & = & \begin{cases}
                         \frac{x_1 -x_1^{j-1} }{ \Delta x_1} & x_1^{j-1}  \leq x_1 \leq {x_1^j} \\
                         \frac{x_1^{j+1}  - x_1}{\Delta x_1} & x_1^j \leq x_1 \leq x_1^{j+1}  \\
                         0 & {\mbox{ otherwise}}~~~~~.
                      \end{cases} 
\end{eqnarray*}
with a similar definition for $\phi_2^k(x_2)$.
We can then represent the solution $\plainv(x_1, x_2, \tau)$ as a linear combination of the basis functions
\begin{eqnarray}
   \plainv(x_1, x_2, \tau) & \simeq & \sum_{j=-N_1/2~}^{N_1/2-1} \sum_{k=-N_2/2}^{N_2/2 -1}
                           \phi_1^j(x_1) \phi_2^k(x_2) \plainv_{j,k}(\tau)
              ~.
       \label{v_basis}
\end{eqnarray}
Substituting equation (\ref{v_basis}) into equation (\ref{green_local}) then gives
{\small 
\begin{eqnarray}
  & & \plainv_{\ell,m}(\tau + \Delta \tau) \nonumber \\
  &= &     \sum_{j=-N_1/2 ~}^{N_1/2-1} \sum_{k=-N_2/2}^{N_2/2 -1} \plainv_{j,k}(\tau)
                     \int_{(x_1)_{\min}}^{(x_1)_{\max}}  \int_{(x_2)_{\min}}^{(x_2)_{\max}} \phi_1^j(x_1)
                                ~\phi_2^k(x_2) ~g( x_1^{\ell} - x_1, x_2^{m} -x_2, \Delta \tau) ~dx_1~dx_2~.
   \nonumber \\
\label{convolve_a}
\end{eqnarray}
} 
For a given pair $(j, k)$ the double integrals then simplify to 
\begin{eqnarray*}
    \int_{(x_1)_{\min}}^{(x_1)_{\max}} &&  \int_{(x_2)_{\min}}^{(x_2)_{\max}} \phi_1^j(x_1)
                             ~\phi_2^k(x_2) ~g( x_1^{\ell} - x_1, x_2^{m} -x_2, \Delta \tau) ~dx_1~dx_2 \nonumber \\
   & & = \int_{x_1^{j} -\Delta x_1}^{x_1^{j} +\Delta x_1}
                                            \int_{x_2^{k} -\Delta x_2}^{x_2^{k} +\Delta x_2}
                           \phi_1^j(x_1) ~\phi_2^k(x_2) 
                         ~g( x_1^{\ell} - x_1, x_2^{m} -x_2, \Delta \tau) ~dx_1~dx_2 \nonumber \\
    & & = \int_{x_1^{\ell} - x_1^j -\Delta x_1}^{x_1^{\ell} - x_1^{j} +\Delta x_1}
                                            \int_{x_2^m - x_2^{k} -\Delta x_2}^{x_2^m -x_2^{k} +\Delta x_2}
                                     \phi_1^j(x_1^{\ell} -x_1) ~\phi_2^k(x_2^m -x_2) ~g( x_1, x_2, \Delta \tau) ~dx_1~dx_2
                                 \nonumber \\
   & & = \int_{x_1^{\ell} - x_1^j -\Delta x_1}^{x_1^{\ell} - x_1^{j} +\Delta x_1}
                              \int_{x_2^m - x_2^{k} -\Delta x_2}^{x_2^m -x_2^{k} +\Delta x_2}
                                 \phi_1^{ \ell -j}(x_1) \phi_2^{m-k}(x_2) ~g( x_1, x_2, \Delta \tau) ~dx_1~dx_2 ~.
\end{eqnarray*}
Here the last line follows from the property of linear basis functions, $$\phi_1^j(x_1^{\ell} -x_1) = \phi_1^{ \ell -j}(x_1), \quad
\phi_2^k(x_2^k -x_2) = \phi_2^{m-k}(x_2).$$
Defining
\begin{eqnarray*}
  & & \tilde{g}_{\ell - j, m-k}( \Delta \tau)  \nonumber \\
    & & \equiv \frac{1}{\Delta x_1 \Delta x_2} 
        \int_{x_1^{\ell} - x_1^j -\Delta x_1}^{x_1^{\ell} - x_1^{j} +\Delta x_1}
                              \int_{x_2^m - x_2^{k} -\Delta x_2}^{x_2^m -x_2^{k} +\Delta x_2}
                                 \phi_1^{ \ell -j}(x_1) \phi_2^{m-k}(x_2) ~g( x_1, x_2, \Delta \tau) ~dx_1~dx_2, 
\end{eqnarray*}
then implies  that for each $~\ell = -N_1/2,\ldots, N_1/2 -1 ,~
                           m = -N_2/2,\ldots, N_2/2-1,$  the convolution (\ref{convolve_a}) can be written as
\begin{eqnarray}
    \plainv_{\ell,m}(\tau + \Delta \tau) 
   & = & \sum_{j=-N_1/2~}^{N_1/2-1} \sum_{k=-N_2/2}^{N_2/2 -1} \plainv_{j,k}(\tau)
                 ~ \tilde{g}_{\ell -j, m-k}( \tau)  
                  ~ \Delta x_1 ~ \Delta x_2.   \label{green+project_1}
\end{eqnarray}
Furthermore using the Fourier series form for $g(\cdot)$ from equation (\ref{series_1}) we have
{\small 
\begin{eqnarray}
    & &\tilde{g}_{\ell -j ,m -k}(\Delta \tau) \nonumber \\
    & = & \frac{1}{ \Delta x_1 ~\Delta x_2} \int_{x_1^{^\ell} -x_1^j  -\Delta x_1}^{x_1^{\ell} -x_1^j +\Delta x_1}
                                            \int_{x_2^{m} -x_2^k -\Delta x_2}^{x_2^{m} -x_2^k +\Delta x_2}
                                          \phi_1^{\ell-j}(x_1) ~\phi_2^{m-k}(x_2) ~  g( x_1, x_2, \tau) ~dx_1~dx_2 
      \nonumber \\
    & = & \frac{1}{P_1 P_2} 
       \sum_{u=-\infty}^{\infty}  \sum_{p=-\infty}^{\infty}
                          \frac{ \sin^2 \pi \omega_1^u \Delta x_1}{(\pi \omega_1^u \Delta x_1)^2} ~
                          \frac{ \sin^2 \pi \omega_2^p \Delta x_2}{(\pi \omega_2^p \Delta x_2)^2} ~
                                        G(\omega_1^u, \omega_2^p, \Delta \tau) 
                                         ~e^{ 2 \pi i \omega_1^u (x_1^{\ell} - x_1^j)} 
                                                           ~e^{  2 \pi i \omega_2^p (x_2^m -x_2^k) } 
                   \nonumber \\
    & = &  \frac{1}{P_1 P_2}  \sum_{u=-\infty}^{\infty}  \sum_{p=-\infty}^{\infty}
                          \frac{ \sin^2 \pi \omega_1^u \Delta x_1}{(\pi \omega_1^u \Delta x_1)^2} ~
                          \frac{ \sin^2 \pi \omega_2^p \Delta x_2}{(\pi \omega_2^p \Delta x_2)^2} ~
                                        G(\omega_1^u, \omega_2^p, \Delta \tau) 
                                         ~e^{ 2 \pi i u (\ell - j)/N_1} 
                                                           ~e^{  2 \pi i p (m -k)/N_2 }.
          \nonumber \\
     \label{smooth_g_1}
\end{eqnarray}
}
\noindent
Here for the inside integrals in (\ref{smooth_g_1})  we have used the following formula for linear basis functions
$$
\frac{1}{\Delta x_1}  \int_{x_1^j-\Delta x_1}^{x_1^j+\Delta x_1} e^{2 \pi  i \omega x_1}  \phi_1^j(x) dx_1 
     = e^{2 \pi i \omega x_1^j} ~\frac{ \sin^2 \pi \omega \Delta x_1}{(\pi \omega \Delta x_1)^2}
$$
with a similar result for the integral involving $\phi_2$.
We note that on a uniform grid $x_1^j = \hat{x}_1^0 + j \Delta x_1 $ and
$x_2^k = \hat{x}_2^0 + k \Delta x_2$, with the result that
$x_1^{\ell-j} = x_1^\ell - x_1^j = (\ell -j) \Delta x_1$,
and $x_2^{m-k} = x_2^m - x_2^k = (m-k) \Delta x_2$.
Consequently,  equation (\ref{smooth_g_1}) is independent of $\hat{x}_1^0, ~ \hat{x}_2^0$.\footnote{For the
case of $u=0, p = 0$ in equation (\ref{smooth_g_1}) we take the limiting value
as $ \omega_1^0 \rightarrow 0, \omega_2^0 \rightarrow 0$.}

For future reference, we note the DFT pair for any grid function\footnote{Note that the DFT here
is independent of any physical coordinates.}
$h(x_1^j, x_2^k) = h_{jk}$ is given by
\begin{eqnarray}
    H(\omega_1^\ell, \omega_2^m, \Delta \tau) 
       & = & \frac{P_1}{N_1} \frac{P_2}{N_2}
                  \sum_{j,k}  e^{- 2 \pi i  \ell j/N_1 } 
                              e^{- 2 \pi i m k/N_2} h_{j k}
                     \label{DFT_1} \\
    h_{p n}(\Delta \tau) & = &
               \frac{1}{P_1 P_2} 
                  \sum_{\ell,m}
                         e^{2 \pi i  \ell p/N_1  } e^{ 2 \pi i m n/N_2}  
                            H(\omega_1^\ell, \omega_2^m, \Delta \tau)
\label{IDFT_1}
\end{eqnarray}
which can be verified by substituting equation (\ref{DFT_1}) into
equation (\ref{IDFT_1})
\begin{eqnarray*}
    h_{p n}(\Delta \tau) & = & 
                  \frac{1}{N_1 N_2}
                      \sum_{j,k}  
                          h_{j k} 
                           \sum_{\ell,m}
                            e^{ 2 \pi i  \ell (p- j)/N_1  } 
                            e^{ 2 \pi i m (n-k)/N_2  }
               =  h_{p n}~,
\end{eqnarray*}
since
\begin{eqnarray}
    \sum_{\ell,m} e^{ 2 \pi i  \ell (p- j)/N_1  } 
                            e^{ 2 \pi i m (n-k)/N_2  }
              & = & \begin{cases}  
                           N_1 N_2 & p = j {\text{ and }} n=k \\
                           0 & {\text{ otherwise}} \\
                    \end{cases} ~.
             \label{orthog_1}
\end{eqnarray}
We denote this pair by $H = \text{DFT}(h), h = \text{IDFT}(H)$.

\subsubsection{Discrete Index Domain}\label{index_section}

Although we have defined the solution on the physical domain $\Omega$,
note that from equation (\ref{smooth_g_1}) $\tilde{g}_{\ell -j, m-n}$
is independent of $(\hat{x}_1^0, \hat{x}_2^0)$.  In addition, the DFT pair
is defined independently of the original grid.
Consequently, in order to avoid confusion concerning the various physical and grid domains,
we will refer only to discrete index domains.  Define
\begin{eqnarray*}
   \mathcal{D} = \left\{ (\ell,j) \big\vert  -\frac{N_1}{2} \leq \ell \leq  \frac{N_1}{2}-1 ~,~ -\frac{N_2}{2} \leq j \leq \frac{N_2}{2}-1 
                 \right\}.
\end{eqnarray*}
Then $\tilde{g}_{\ell,j}, \plainv_{\ell,j}$ are both in the  index domain $\mathcal{D}$.

\subsection{Efficient Computation of the Discrete Convolution}
The discrete convolution (\ref{green+project_1}) is performed at each rebalancing
date, hence an efficient evaluation is desirable.
The convolution can be written as
\begin{eqnarray}
   \plainv_{\ell,m}(\tau + \Delta \tau) 
    &= & \sum_{j=-N_1/2}^{N_1/2-1} \sum_{~~k=-N_2/2}^{N_2/2 -1} \plainv_{j,k}(\tau)
                 \tilde{g}_{\ell - j, m - k} (\Delta \tau)
                  ~ \Delta x_1 ~ \Delta x_2 \nonumber \\
   & = & \frac{P_1}{N_1} \frac{P_2}{N_2}\sum_{j,k} \tilde{g}_{\ell - j, m - k} (\Delta \tau) \plainv_{j,k}(\tau)
                  ~~. \label{conv_final_1}
\end{eqnarray}
Using equations (\ref{DFT_1}) and (\ref{IDFT_1}) to write $\tilde{g}(\Delta \tau) 
= \text{IDFT}( \tilde{G} (\Delta \tau)), {v} = \text{IDFT}( {V} )$ 
in equation (\ref{conv_final_1}) gives
\begin{eqnarray*}
& & \plainv_{\ell,m}(\tau + \Delta \tau)  \nonumber \\
  & & =  \frac{1}{N_1 P_1} \frac{1}{ N_2 P_2} 
    \sum_{j,k} \sum_{p,u} \tilde{G}_{pu}(\Delta \tau) e^{2 \pi i p (\ell -j)/N_1}
                                                      e^{2 \pi i u (m-k) /N_2}
                         \sum_{q,r} V_{q,r} e^{2 \pi i q j/N_1} e^{2 \pi i r k /N_2}
                    \nonumber \\
 &  & = \frac{1}{N_1 P_1} \frac{1}{N_2 P_2}
               \sum_{p,u} \sum_{q,r} \tilde{G}_{pu}(\Delta \tau) V_{qr}
                        e^{2\pi i p \ell/N_1} e^{2 \pi i  u m /N_2}
                      \sum_{j,k} e^{2 \pi i j (q -p)/N_1}
                                 e^{ 2 \pi i k (r-u)/N_2} \nonumber \\
 &  & =   \frac{1}{ P_1} \frac{1}{ P_2}
                      \sum_{p,u}  \tilde{G}_{pu}(\Delta \tau) {V}_{pu} e^{2 \pi i p \ell/N_1} 
                                                                e^{ 2 \pi i u m/N_2}
\end{eqnarray*}
where we have used equation (\ref{orthog_1}) in the last step. More compactly we can write 
\begin{eqnarray}
\tilde{v}(\tau + \Delta \tau) & = & \text{IDFT} \bigl(~ \tilde{G}(\Delta \tau) \circ \text{DFT} 
                \left({v}(\tau) \right) ~\bigr) ~,
   \label{time_advance_compact}
\end{eqnarray}
where $ y \circ z$ is the Hadamard product of vectors $y, z$.

\subsection{$\delta$-Monotonicity}
Using the definition for $\tilde{g}$ from equation (\ref{smooth_g_1}), then
this represents an exact integration over linear basis functions of
the exact Green's function (with periodic boundary conditions)
on $\Omega$.  Hence 
\begin{eqnarray}
    \tilde{g}_{j,k} &\geq & 0 ~;~ \forall (j,k) \in \mathcal{D}
  ~.
              \label{mono_2}
\end{eqnarray}
This gives rise to an important {\em monotonicity} property \citep{forsythlabahn2017}.
Suppose we have two discrete solutions $\plainv_{\ell,m}(\tau) , u_{\ell,m}(\tau)$
such that
\begin{eqnarray}
   \plainv_{\ell,m}(\tau)  > u_{\ell,m}(\tau)~;~ \forall (\ell,m) \in \mathcal{D}~.
     \label{order_1}
\end{eqnarray}
Then, using the time advance algorithm (\ref{green+project_1}), it follows
that the {\em discrete comparison principle} holds:
\begin{eqnarray*}
   \plainv_{\ell,m}(\tau + \Delta \tau)  > u_{\ell,m}(\tau + \Delta \tau)~;~ 
          \forall (\ell,m) \in \mathcal{D}.
\end{eqnarray*}
The positivity condition (\ref{mono_2}) ensures that the order property (\ref{order_1}) is
preserved by the discretized algorithm.  This is, of course, a property of the
exact solution.  This property is important
since we compare the values of the discrete solution 
in order to determine the optimal control.  See \citet{forsythlabahn2017} for more discussion of this.

However, note that the dimension of $\tilde{g}$ is $(N_1, N_2)$, but we need to sum
an infinite series to determine the discrete values of $\tilde{g}_{m,u}$ from
equation (\ref{green+project_1}).  In practice, we truncate the series in  
equation (\ref{smooth_g_1}), that is, 
{\small
\begin{eqnarray}
& &\tilde{g}_{\ell -j ,m -k}(\Delta \tau) \nonumber \\
& \simeq &  \frac{1}{P_1 P_2}  \sum_{u=-N_1^g/2~~}^{N_1^g/2 -1 }  \sum_{p=-N_2^g/2}^{N_2^g/2 -1}
                          \left( \frac{ \sin^2 \pi \omega_1^u \Delta x_1}{(\pi \omega_1^u \Delta x_1)^2} \right)
                          \left( \frac{ \sin^2 \pi \omega_2^p \Delta x_2}{(\pi \omega_2^p \Delta x_2)^2} \right)
                                        G(\omega_1^u, \omega_2^p, \Delta \tau) 
                                         ~e^{ 2 \pi i \frac{u (\ell - j)}{N_1}} 
                                                           ~e^{  2 \pi i  \frac{p(m -k)}{N_2 }} 
          \nonumber \\
    \label{smooth_g_trunc}
\end{eqnarray}
} 
\noindent
where the grid sizes $(N_1^g, N_2^g)$ can be chosen independently from $(N_1, N_2)$.  However, truncating
the Fourier series means that the  positivity condition (\ref{mono_2}) may not hold any longer.
In \citet{forsythlabahn2017}, it is suggested that $(N_1^g, N_2^g)$ be selected so that
\begin{eqnarray}
    \sum_{m = -N_1^g/2}^{N_1^g/2+1} \sum_{n=-N_2^g/2}^{N_2^g/2-1} 
             \Delta x_1 \Delta x_2 ~| \min( \tilde{g}_{m,n},0)| & < & \delta \frac{\Delta \tau}{T}
    ~. \label{mono_cond_1}
\end{eqnarray}
This ensures that the cumulative effect (after  $(T/(\Delta \tau)$ steps) 
of non-positivity means that the discrete comparison
principle holds to $O(\delta)$.  Note that if the timestep $(\Delta \tau)$ is
constant (which is usually the case), then equation (\ref{smooth_g_trunc}) needs to be
computed only once. Hence it is inexpensive to select sizes $(N_1^g, N_2^g)$ which
satisfy condition (\ref{mono_cond_1}).  Equation (\ref{smooth_g_trunc})  can also
be computed efficiently using FFTs, assuming $N_1^g/N_1, N_2^g/N_2$ are powers of two
(see \citet{forsythlabahn2017}.

\section{Periodicity and the problem of Wrap-around error}\label{section_periodicity}
Localization of the Green's function effectively implies a periodic extension of 
the Green's function and the solution (see Assumption \ref{periodic_assumption}).
In option pricing applications, this {\em wrap-around} error does not
cause difficulty.  However, in control applications, impulse controls (such
as rebalancing) may require extensive use of solution values near the edges
of the grid.  For example, derisking by investing in an all
bond portfolio results in an impulse control which drives $s \rightarrow 0$ instantaneously.
This can cause significant wrap-around pollution
(see \citep{Lippa2013,Ruijter_2013,Ignatieva_2018}).

Given the localized problem defined on $\Omega = [(x_1)_{\min}, (x_1)_{\max}] \times [(x_2)_{\min}, (x_2)_{\max}]$, with widths
$$
P_1 = (x_1)_{\max} - (x_1)_{\min}  ~~ \mbox{ and } ~~ P_2 = (x_2)_{\max} - (x_2)_{\min} 
$$
we construct an auxiliary grid with $N_1^{\dagger} = 2N_1$ nodes in the $x_1$ direction,
and $N_2^{\dagger} = 2 N_2$ nodes in the $x_2$ direction, on the domain $\Omega^{\dagger}
= [(x_1)_{\min}^{\dagger}, (x_1)_{\max}^{\dagger}] \times [(x_2)_{\min}^{\dagger}, (x_2)_{\max}^{\dagger}]$,
with
\begin{eqnarray*}
  (x_1)_{\min}^{\dagger} = (x_1)_{\min} - \frac{ P_1} {2} 
  & ;~~~ & (x_1)_{\max}^{\dagger} = (x_1)_{\max} + \frac{ P_1 }{2}   \nonumber \\   
(x_2)_{\min}^{\dagger} = (x_2)_{\min} - \frac{ P_2}{2}
  & ; ~~~& (x_2)_{\max}^{\dagger} = (x_2)_{\max} + \frac{ P_2 }{2} \nonumber \\ \nonumber \\
 %
  P_1^\dagger  =  (x_1)_{\max}^{\dagger} - (x_1)_{\min}^{\dagger} =  2  P_1 &;~~~ & 
   P_2^\dagger  = (x_2)_{\max}^{\dagger} - (x_2)_{\min}^{\dagger} = 2 P_2 
\end{eqnarray*}

Although we have increased the size of the physical domain $\Omega$, we remind the reader
of the discussion in Section \ref{index_section}.  It will be more convenient in
the following to refer to the extended index domain
\begin{eqnarray*}
  \mathcal{D}^{\dagger} = \left\{ (j, k)~ \big\vert  ~~
              -N_1^{\dagger}/2 \leq j \leq  N_1^{\dagger}/2-1 ~,~ -N^{\dagger}_2/2 
               \leq k \leq N_2^{\dagger}/2-1 
                 \right\}~.  
\end{eqnarray*}

Denote the projected Green's function for grid indexes in $\mathcal{D}^{\dagger}$
by $\tilde{g}^{\dagger}$.  Note from equation (\ref{smooth_g_1}) that $\tilde{g}^{\dagger}$
does not depend on the actual physical domain, but only on $(\Delta x_1, \Delta x_2)$.
We construct and store the DFT of the projection of the Green's function $\tilde{G}^{\dagger}(\Delta \tau)
= \text{DFT}( \tilde{g}^{\dagger} (\Delta \tau) )$ 
on this auxiliary index grid $\mathcal{D}^{\dagger}$.   

Before applying the time advance algorithm (\ref{time_advance_compact})
we form the padded array ${v}^{\dagger}(\tau)$:
{\small 
\begin{eqnarray}
   \plainv^{\dagger}(x_1^j, x_2^k, \tau) =
      \begin{cases}
            \plainv(x_1^j, x_2^k, \tau)~ , &  (j, k) \in \mathcal{D} \\
       \plainv(x_1^{-N_1/2}, x_2^k, \tau) ~, &  j \in [-N_1^{\dagger}/2, -N_1/2-1] ,
                                         ~~ k \in [-N_2/2, N_2/2-1] \\
       \plainv(x_1^{-N_1/2}, x_2^{-N_2/2} , \tau) ~, & j \in [-N_1^{\dagger}/2, -N_1/2-1] ,
                                                  ~~k \in [-N_2^{\dagger}/2,  -N_2/2 -1] \\
       \plainv(x_1^j, x_2^{-N_2/2}, \tau) ~, &  j \in [-N_1/2, N_1/2-1] ,
                                                   ~~ k \in [-N_2^{\dagger}/2, -N_2/2-1]    \\
       A(x_1^j, x_2^k, \tau) ~,   & j > N_1/2-1~~ {\text{ or }}~~ k > N_2/2-1 \\
    \end{cases} \nonumber \\
           \label{padded_v} 
\end{eqnarray}}
where $ A(x_1^j, x_2^k, \tau)$ is an asymptotic form of the solution which
we assume to be available from financial reasoning. 
On the auxiliary grid,
the points where $x_1 < (x_1)_{min}$ or $x_2 < (x_2)_{\min}$ correspond
to the points where $s \rightarrow 0$ or $b \rightarrow 0$, with very small
grid spacing.  Hence extending the solution by  constant values to the left
and bottom is expected to generate a very small error. 

We then modify the time advance algorithm (\ref{time_advance_compact})
as shown in Algorithm \ref{advance_time_padded}.
Note that the step (\ref{discard_step}) discards the computed solution in the padded areas 
$ \{ \plainv^{\dagger} ( \tau + \Delta \tau)_{j, k} \big\vert (j, k)  \in 
\mathcal{D}^{\dagger} - \mathcal{D} \}$,
since these points may be contaminated by wrap-around errors.

\begin{algorithm}[!h]
\caption{
\label{advance_time_padded}
Advance time $\plainv(\tau) \rightarrow \plainv(\tau + \Delta \tau)$.
}
 \begin{algorithmic}[1]
    \REQUIRE $ \plainv(\tau)~;~ \widetilde{G}^{\dagger}$
    \STATE Form $\plainv^{\dagger} (\tau)$ using (\ref{padded_v})
    \STATE $\plainv^{\dagger} ( \tau + \Delta \tau)
              = \text{IDFT} \bigl(~ \tilde{G}^{\dagger} \circ \text{DFT} \left( \plainv^{\dagger} ( \tau ) \right) ~\bigr) $
          ~~~~~~~~~~~~~~~~~~~~~~\COMMENT {IDFT (Hadamard product)}
    \STATE $\plainv_{j, k}( \tau + \Delta \tau) = \plainv^{\dagger}_{j, k}( \tau + \Delta \tau) 
               ~,~ (j, k) \in \mathcal{D}$  ~~~~~~~~~~~~~~~~~\COMMENT {Discard values in $\mathcal{D}^{\dagger} - \mathcal{D}$.}
           \label{discard_step}
 \end{algorithmic}
\end{algorithm}

\subsection{Error Due to Wrap-around}
Let $\tau^n = n \Delta \tau$ and  $\mathcal{N} \Delta \tau = T$, with $\Delta \tau$ fixed.  We can then write the 
discrete convolution (\ref{conv_final_1}) in the physical domain,  compactly as
\begin{eqnarray}
  \plainv_{\ell,m}^{n+1} & = & \Delta x_1 \Delta x_2
      \sum_{ (j,k) \in \mathcal{D}^{\dagger}}
           \tilde{g}_{\ell -j, m - k}^{\dagger} \plainv_{j,k}^{\dagger}(\tau^n) ~,~
                    (\ell,m ) \in \mathcal{D} ~.
    \label{conv_compact}
\end{eqnarray}
Note that the righthand side of equation (\ref{conv_compact})  uses the padded values for $\plainv_{j,k}^{\dagger}$,
but generates the unpadded $\plainv_{\ell,m}^{n+1} $ on the lefthand side.  In addition,
$\tilde{g}^{\dagger}$ is periodic with period $P_1^{\dagger}$ in the $x_1$ direction,
and period $P_2^{\dagger}$ in the $x_2$ direction.

The use of Algorithm \ref{advance_time_padded} means that the periodic extension 
of $\tilde{g}$ is used in equation (\ref{conv_compact}) for any
terms of the sum such that 
\begin{eqnarray*}
     ( \ell-j, m-k) \notin \mathcal{D}^{\dagger} ~.
\end{eqnarray*}
This periodic extension {\em wraps around} the solution domain.  For example, points with
indexes $\ell-j  < - N_1^{\dagger}/2$ reference points with index $\ell-j + N_1^{\dagger}$,
which is at the opposite side of the grid, and potentially generates
wrap-around error. 
See also Appendix \ref{wrap_appendix} for further discussion of this.

\begin{definition}[Wraparound error.] \label{def_wrap}
Assume $\tilde{g}_{\ell , m}^{\dagger}$ for $(\ell, m) \in \mathcal{D}^{\dagger}$
is periodic
\begin{eqnarray}
    \tilde{g}^\dagger (\ell  \pm N_1^{\dagger}, m)
            =
    \tilde{g}^\dagger( \ell , m \pm N_2^{\dagger} )
            &=& \tilde{g}^\dagger(\ell, m )  
     ~. \label{index_periodic}
\end{eqnarray}
Suppose $\plainv_{j,k}^{\dagger}(\tau^n)$ is determined by boundary
data for $(x_1^j , x_2^m) \in ( \Omega^{\dagger} - \Omega)$, which
we assume to be exact,\footnote{In practice, we use the method
in equation (\ref{padded_v}), which we expect will generate a small error.}
with $(N_1^{\dagger}, N_2^{\dagger}) =  (2 N_1, 2 N_2)$.
Then, the wrap-around error for convolution (\ref{conv_compact}) evaluated using
DFTs, at timestep $n \Delta \tau$, denoted by
$\epsilon^n_{\text{wrap}}$ is
\begin{eqnarray}
    \varepsilon^n_{\text{wrap}}
   & = & \Delta x_1 \Delta x_2 \max_{ (\ell, m) \in \mathcal{D} }
           \sum_{ ( j ,k) \in \mathcal{D}^{\dagger} }
            \big\vert \tilde{g}_{\ell -j, m - k}^{\dagger} 
                      \plainv_{j,k}^{\dagger}(\tau^n) \big\vert
               ~{\bf{1}}_{ (\ell-j, m-k) \notin \mathcal{D}^{\dagger} }
               ~ .\label{wrap_def_a}
 \end{eqnarray}
\end{definition}

We now state a theorem on the effectiveness of our padding method.
See Appendix \ref{wrap_appendix}  for a proof.

\begin{theorem}[Wraparound error]\label{wrap_theorem}
Suppose $\tilde{g}_{\ell , m}^{\dagger}$ for $(\ell, m) \in \mathcal{D}^{\dagger}$
is periodic as defined in equation (\ref{index_periodic}),
and that $\plainv_{\ell,m}^{\dagger}(\tau^n)$ be determined by
boundary data for $(\ell, m) \in (\mathcal{D}^{\dagger} - \mathcal{D} )$,
with $(N_1^{\dagger}, N_2^{\dagger}) = 2 (N_1, N_2)$.
Assume further that there exists a constant $C$ such that
\begin{eqnarray*}
       \big\vert \plainv_{\ell, m}^{\dagger}(\tau^n) \big \vert \leq C ~,~~~~ (\ell, m) \in \mathcal{D}^{\dagger}
                   ~,~~~ \forall ~ n ~\mbox{ such that } ~ (n \Delta \tau)  \leq T
         ~,
\end{eqnarray*}
and that we choose $(N_1, N_2)$ sufficiently large so  that
\begin{eqnarray}
      \Delta x_1 \Delta x_2 \sum_{( j ,k) \in \mathcal{D}^{\dagger} - \mathcal{D} } 
                  | \tilde{g}_{j,k}^{\dagger} | 
   &~~ \leq ~~ &
      \varepsilon_b ~ \Delta \tau~.
     \label{wrap_error_local_1a}
\end{eqnarray}
Then
\begin{eqnarray}
    \varepsilon^n_{\text{wrap}} & \leq& C \varepsilon_b ~ \Delta \tau ~,
                \label{wrap_error_local_1}
\end{eqnarray}
and the wrap-around error after $\mathcal{N}$ steps is bounded by
\begin{eqnarray*}
  \varepsilon^{\mathcal{N}} & \leq & T \epsilon_b C  ~~ \mbox{ where } 
                 T = \mathcal{N} \Delta \tau ~. 
\end{eqnarray*}
\end{theorem}

\begin{remark}[Asymptotic Form of Green's Function]
For a pure diffusion in one dimension, the Green's function for the Black-Scholes
equation behaves like
\begin{eqnarray}
   g(x - x^{\prime}, \Delta \tau) & = & O \left( e^{ - (x-x^{\prime})^2/( \sigma^2 \Delta \tau) } \right)
     \mbox{ as } 
     | x - x^{\prime} | \rightarrow \infty ~. \nonumber
\end{eqnarray}
Similar asymptotic exponential decay as $\max( |x_1|,|x_2|) \rightarrow \infty$ 
is also seen for the full 2-d Green's function \citep{menaldi_1992}.
Hence  ensuring that condition (\ref{wrap_error_local_1a}) holds is not onerous.
\end{remark}

\subsection{Numerical Method for Optimal Control}

As mentioned  in Section \ref{expanded_1}, we need  the optimal pair $(\qq, \pp)$ at
each rebalancing date. The optimal controls are  obtained by first discretizing the controls
and thereafter determining the controls by exhaustive search. Recall that the PIDE grid
spacing is $\Delta x_1, \Delta x_2$. Let  $N_{\qq}$ and $N_{\pp}$
denote the  number of points in the discretizations of $\qq$ and $\pp$, respectively.  We will also need to 
discretize the possible values of total wealth $w = s + b$, with $N_w$  denoting
the number of discrete wealth values. This approach (discretization
and  search), is then guaranteed to converge to the viscosity
solution of the optimal control problem in the limit as
\begin{eqnarray}
   \Delta x_1, \Delta x_2 \rightarrow 0 & ; & N_{\qq} ,  N_{\pp}, N_w \rightarrow \infty
    ~;~ \delta \rightarrow 0~, \label{converge_limit}
\end{eqnarray}
which is discussed in Section \ref{converge_section}.

Given an array of discrete values, that is, $\plainv_{j,k} (\tau_n^-)$, we 
 denote the linear interpolant of the grid values at an arbitrary point $(\zz, \zzz)$
by
\begin{eqnarray*}
    \plainv( \mathcal{I}( \zz, \zzz) , \tau_n^-) ~.
\end{eqnarray*}

The numerical analogue of equation (\ref{dynamic_b}) requires a temporary array, $h(w_k), k=1, \ldots, N_w$,
where $w_k = s + b$, the total wealth.
This array is determined by
\begin{eqnarray}
                \underbrace{\pp^*(w_k)}_{ k=1,\ldots,N_w} &=& \argmax_{   
                                \substack{
                                     j =1, \ldots, N_{\pp}\\
                                      \pp_j \in \mathcal{Z}_{\pp}(w_k, \tau_n^-)
                                    }
                                     }
                                     \plainv \biggl( ~
                                        \mathcal{I}\left(~ \log( w_k \pp_j)~,~ \log( w_k (1 - \pp_j) ) ~\right) 
                                           ,\tau_n^-
                                      \biggr)    
                      \nonumber \\
        \underbrace{h(w_k)}_{k=1,\ldots,N_w}
                                  & = & 
                                     \plainv
                                        \biggl( 
                                          \mathcal{I}\biggl(~ 
                                                  \log( ~w_k \pp^*(w_k)~)~,~ \log(~ w_k (1 - \pp^*(w_k)) ~)~
                                                ~\biggr), \tau_n^-  
                                         \biggr)
           ~.
           \label{optimal_h_def}
\end{eqnarray}
In order to compute (\ref{dynamic_c}), we break this down into two steps, first using
\begin{eqnarray}
   \underbrace{\qq^*( w_k)}_{k=1,\ldots,N_w}
            & = & \argmax_{
                    \substack{
                     j =1, \ldots, N_{\qq}\\
                                      \qq_j \in \mathcal{Z}_{\qq}(w_k, \tau_n)
                     }
                  }
                             \biggl\{~
                               \qq_j + h\left(~ \mathcal{I}(w_k - \qq_j) ~\right)
                             ~\biggr\}
             \label{optimal_q}
\end{eqnarray}
where $ h(\mathcal{I}(\cdot) )$ refers to one dimensional interpolation for the one dimensional array $h(w_k)$.
Then
\begin{eqnarray}
    \underbrace{\plainv(x_1^j, x_2^k, \tau_n^+)}_{\substack{
                                             j=-N_1/2,\ldots, N_1/2-1\\
                                              k=-N_2/2,\ldots,N_2/2-1\\
                                               w_{j,k} = e^{ x_1^j} + e^{x_2^k}
                                            }
                                  }
                         & = &   \qq^*\left( ~\mathcal{I}( w_{j,k}) ~\right)
                        + h\biggl( ~\mathcal{I} \left(~ 
                                         w_{j,k}  - \qq^*\left(~ \mathcal{I}( w_{j,k}) ~\right) 
                                        ~\right) ~
                            \biggr) 
            ~.
                   \label{v_plus}
\end{eqnarray}

\subsection{Convergence to the Viscosity Solution}\label{converge_section}
Convergence of a numerical method for HJB equations is ensured
provided that the complete numerical scheme is $\delta-$monotone, consistent and $\ell_{\infty}$
stable, (see \citet{barles-souganidis:1991,forsythlabahn2017}).
We should clarify that each of these properties must be verified for the complete numerical
scheme, including the time advance and the impulse control update.

Stability is easily established using the techniques in \citet{forsythlabahn2017}. The time
advance step is $\delta$-monotone by construction.  The optimization and projection steps 
in equations (\ref{optimal_h_def}-\ref{v_plus}) are monotone, since linear
interpolation is used at all off grid points.  Note that extrapolation is
never used (which ensures stability).

Proving consistency is a bit more involved.   The key is to note that consistency,
in the viscosity sense, is defined in terms of smooth, infinitely differentiable
test functions.  This means that proof of consistency involves replacing
all the functions in equations (\ref{optimal_h_def}-\ref{v_plus}) by smooth test
functions.  Discretizing the control and using linear interpolation
means that the maximum of the smooth test functions in
equations (\ref{optimal_h_def}-\ref{v_plus}) will converge  to
the global maximum in the limit. 

While the discretization of the control and use of exhaustive search
may appear to be crude, this technique (coupled with linear interpolation)
is the only method which is guaranteed to converge to the global maximum
of a (possibly) non-convex smooth test function.  This is not merely
a point of academic interest.  Initial attempts to 
use standard optimization methods sometimes resulted in convergence to
a different solution (compared to the solution obtained by
exhaustive search).  Since the viscosity solution is unique,
we can deduce that the solution obtained without using
exhaustive search was spurious.

However, numerical experiments show that the discretization of the control
can be quite coarse, and this does not appear to impede convergence.
A non-rigorous explanation of this effect is that a smooth function
has a vanishing first derivative at extreme points.


\subsection{Summary: Algorithm for Problem (\ref{expanded_1}).}
Algorithm \ref{alg_monotone_final} summarizes the final algorithm.  The initial
set-up cost is computing the Fourier weights on the extended domain $\widetilde{G}^{\dagger}$.
Each timestep consists of first determining the optimal control, then followed by
advancing to the next rebalancing time using the precomputed weights $\widetilde{G}^{\dagger}$.

\begin{algorithm}[!h]
\caption{
\label{alg_monotone_final}
Monotone Fourier method.
}
 \begin{algorithmic}[1]
    \REQUIRE Weights $\widetilde{G}^{\dagger}$  in padded Fourier domain satisfying tolerance $\delta$
               (see equations (\ref{smooth_g_trunc}-\ref{mono_cond_1}))
    \STATE Input: number of timesteps $\mathcal{N}$, initial solution $(\plainv^0)^-$
    \FOR[Timestep loop]{$n=1,\ldots, \mathcal{N}$}
    \STATE Optimal control $\plainv(\tau_{n-1}^-) \rightarrow \plainv(\tau_{n-1}^+)$ equations (\ref{optimal_h_def}-\ref{v_plus})
    \STATE Advance time $\plainv(\tau_{n-1}^+) \rightarrow \plainv(\tau_{n}^-)$ using Algorithm \ref{advance_time_padded}
   \ENDFOR \COMMENT{End timestep loop}
 \end{algorithmic}
\end{algorithm}

\begin{remark}[Use of linear interpolation]
Equations (\ref{optimal_h_def}-\ref{v_plus}) make heavy use of linear interpolation.
As discussed in \cite{forsythlabahn2017}) this is the only interpolation method (in general)
which will preserve $\delta$-monotonicity.  This means that the convergence rate cannot exceed
second order in the grid spacing.  This is in contrast to the methods in \citep{Fang2008,Fang2009,Ruijter_2013},
where, in some circumstances, high order rates of convergence can be achieved,
but at the cost of possible non-monotonicity.
\end{remark}

\begin{remark}[Insolvency Between Rebalancing Times]
For $b >0$, insolvency cannot occur between rebalancing dates.
However, for $b<0$, it is possible that a large drop in the value
of stocks could result in $s+b<0$ between rebalancing times.
Numerical tests using subtimesteps between rebalancing times
(and checking for insolvency) did not show any significant
changes to the final values at $\tau =T$.  Hence, in all our 
reported numerical
tests, we only check for insolvency at rebalancing times.
\end{remark}

\subsection{Final Algorithm: Problem EW-ES (\ref{final_step_EWES}).}
The final step for the complete solution of the EW-ES problem
is the optimization step (\ref{final_step_EWES}).  We solve Problem
EW-ES on a sequence of grids, with increasing number of grid nodes.
On the coarsest grid, we discretize $W^*$ and find the maximum
of equation (\ref{final_step_EWES}) by exhaustive search.
Note that each evaluation of the objective function in (\ref{final_step_EWES})
requires a solution of  Problem \ref{expanded_1}.
Using the coarse grid value of $W^*$ as a starting value, we use
a one-dimensional optimization algorithm to maximize the objective
function (\ref{final_step_EWES}) on each finer grid.

\section{Numerical Example and Results}\label{example_section}
We consider a 65 year old retiree who has \$1,000K (one million) in
pension savings.  The retiree needs to withdraw a minimum of
30K per year, but has no use for more than 60K per year.  All amounts
are inflation adjusted.  We assume
that the retiree has other sources of income (work pension, government
benefits) which, when added to the 30K per year of withdrawals, accumulate
to a satisfactory income level.

We consider a 30 year time horizon, since this is consistent with
the \citet{Bengen1994} scenario.  Recall that the probability of
a 65 year old Canadian male attaining the age of 95 is about 0.13, so
this is a fairly conservative assumption. Similar probabilities are also true for other western countries.

The pension savings are assumed to be held in a tax advantaged account,
so that there are no tax consequences for rebalancing.  
We consider the
withdrawals to be the desired amount before any income taxes.

Note that we also consider that the retiree has mortgage free
real estate worth 400K.  We can regard this as a hedge of last
resort.
\begin{remark}[Reverse Mortgage Hedge]\label{reverse_mortgage_hedge}
If we assume that a reverse mortgage can be used to obtain
a non-recourse loan of half the value of the real estate ($200K$),
then from a risk management point of view, any
strategy which results in an expected shortfall $> -200K$ is
probably acceptable.  All quantities are assumed inflation
adjusted.
\end{remark}

Other scenario details are listed in Table \ref{base_case_1}.
As a point of comparison, the \citet{Bengen1994} strategy 
would withdraw a fixed amount of 40K (real) per year (4\% of the
initial 1,000K), and rebalance to a constant weight of
50\% in stocks at each rebalancing date.

\begin{table}[hbt!]
\begin{center}
\begin{tabular}{lc} \toprule
Investment horizon $T$ (years) & 30.0  \\
Equity market index & CRSP Cap-weighted index (real) \\
Bond index & 30-day T-bill (US) (real) \\
Initial portfolio value $W_0$  & 1000 \\
Mortgage free real estate & 400 \\
Cash withdrawal$/$rebalancing times & $t=0,1.0, 2.0,\ldots, 29.0$\\
Maximum withdrawal (per year)   & $ q_{\max} = 60$\\
Minimum withdrawal (per year)   & $ q_{\min} = 30$\\
Equity fraction range & $[0,p_{\max}]$\\
                      & $p_{\max} = 0.5, 0.8, 1.0, 1.3$\\
Borrowing spread $\mu_c^{b}$ & 0.03 \\
Rebalancing interval (years) & 1.0  \\
$\alpha$ (Expected shortfall parameter)       & .05 \\
Stabilization $\epsilon$ (see equation (\ref{PCES_a})) & $ -10^{-4} $ \\
Market parameters & See Table~\ref{fit_params} \\ \bottomrule
\end{tabular}
\caption{Input data for examples.  Monetary units: thousands of dollars.
All amounts inflation adjusted.
\label{base_case_1}}
\end{center}
\end{table}

For the computational study in this paper, we use data from the Center for Research in Security Prices (CRSP)
on a monthly basis from 1926:1 to 2024:12.\footnote{More specifically, results presented here
were calculated based on data from Historical Indexes, \copyright
2024 Center for Research in Security Prices (CRSP), The University of
Chicago Booth School of Business. Wharton Research Data Services was
used in preparing this article. This service and the data available
thereon constitute valuable intellectual property and trade secrets
of WRDS and/or its third-party suppliers.} The specific indices used are the CRSP 30 day U.S. T-bill
index for the bond asset,
and the CRSP cap-weighted total return index for the stock asset\footnote{The stock index
includes all distributions for all domestic stocks trading on major U.S. exchanges.}.
Retirees are, naturally, concerned with preserving real (not nominal) spending power.
Hence, we use the US CPI index (from CSRP) to adjust these indexes for inflation.
We use the above market data in two different ways in subsequent investigations.

We use the threshold technique
\citep{mancini2009,contmancini2011,Dang2015a} to estimate the parameters
for the parametric stochastic process models. Since the index data is in
real terms, all parameters reflect real returns. Table \ref{fit_params}
shows the results of calibrating the models to the historical data.
The correlation $\rho_{sb}$ is computed by removing any returns which
occur at times corresponding to jumps in either series, and then
using the sample covariance. Further discussion of the validity of
assuming that the stock and bond jumps are independent is given in
\citet{forsyth_2020_a}.

{\small
\begin{table}[hbt!]
\begin{center}
\begin{tabular}{cccccccc} \toprule[1pt]
 CRSP & $\mu^s$ & $\sigma^s$ & $\lambda^s$ & $u^s$ &
  $\eta_1^s$ & $\eta_2^s$ & $\rho_{sb}$ \\ \midrule
       &~ 0.088241  &  ~0.147361&   ~0.31313  & ~ 0.22581 & ~4.3608 & ~5.5309 & ~0.096279\\
 \midrule[1pt]
30-day T-bill & $\mu^b$ & $\sigma^b$ & $\lambda^b$ & $u^b$ &
  $\eta_1^b$ & $\eta_2^b$ & $\rho_{sb}$ \\ \midrule
        & ~0.0034 & ~0.0139 & ~0.3838  &  ~0.3947 &  ~61.510 & ~53.356 &~ 0.096279
\\
\bottomrule[1pt]
\end{tabular}
\end{center}
\caption{Parameters for parametric market models (\ref{jump_process_stock})
and (\ref{jump_process_bond}), fit to CRSP data (inflation
adjusted) for 1926:1~to~2024:12.  \label{fit_params} }
\end{table}
}

We first compute and store the optimal controls in the {\em synthetic market}, that is,  the stock market processes follow the parametric models (\ref{jump_process_stock})
and (\ref{jump_process_bond}).  We then use Monte Carlo simulation,
coupled with the stored optimal controls, to analyze the properties of
the optimal allocation/withdrawal.  

Finally, as a check on robustness, we  test the controls (computed in
the synthetic market), in the {\em historical market}, using block bootstrap
resampled historical data.


\subsection{Results: Convergence}
We first check on the convergence of our method.  The numerical parameters
are given in Table \ref{base_case_numerical}.

\begin{table}[hbt!]
\begin{center}
\begin{tabular}{lc} \toprule
Grid sizes ($b >0$) & $512 \times 512$ \\
                    & $1024 \times 1024$ \\
                    & $2048 \times 2048$\\
Extended grid $(N_1^{\dagger}, N_2^{\dagger})$ & ($2N_1, 2N_2$) \\
$\hat{x}_1^0, \hat{x}_2^0$ & $\log( 100 )$\\
$(x_1)_{\min}, (x_2)_{\min}$ & -7.5 + $\hat{x}_1^0$ \\
$(x_1)_{\max}, (x_2)_{\max}$ & +10 + $\hat{x}_2^0$\\
Monotonicity condition $\delta$ (equation(\ref{mono_cond_1})) & $10^{-6}$\\
Number of points in $w$ grid $N_w$ (equation (\ref{optimal_h_def})) & $4N_1$\\
Number of points in $\pp$ grid $N_p$ (equation (\ref{optimal_h_def})) & $N_1/10$\\
Number of points in $\qq$ grid $N_{\qq}$ (equation (\ref{optimal_q})) & $N_1/10$\\
\bottomrule[1pt]
\end{tabular}
\caption{Numerical parameters
\label{base_case_numerical}}
\end{center}
\end{table}


\begin{remark}[Choice of $(x)_{\min}, (x)_{\max}$]
Consider the grid with $b >0$.  In this case, we choose $x_{\min}$ (for both $x_1, x_2$) so
that
\begin{eqnarray}
   e^{x_{\min}} W_0 & < & 10^{-4} W_0
\end{eqnarray}
which is essentially zero  allocation to  bonds or stocks for practical purposes (compared to the original investment).
For example, if the original investment was $\$1000$, then we assume that an allocation of
less than $\$0.10$ is negligible.

In the case of $x_{\max}$ (for both  $x_1, x_2$) we assume that the largest possible wealth would
occur with an all stock portfolio having zero withdrawals.  A rough estimate of  $x_{\max}$ is then given by assuming
geometric Brownian motion with drift $\mu$ and volatility $\sigma$.  We choose $x_{\max}$ such that
\begin{eqnarray}
   e^{x_{\max}} W_0 & > & e^{\mu T} e^{4 \sigma \sqrt{T} } W_0~. \label{x_max_eqn}
\end{eqnarray}
Equation (\ref{x_max_eqn}) postulates a four sigma event, which is extremely unlikely.  Note that
we use $(x_1)_{\max} = (x_2)_{\max} = x_{\max}$, since the impulse control can (in theory) move
the entire stock allocation into bonds instantaneously.

We have tested our results by increasing the ranges of $x_{\min}, x_{\max}$ and verified that
the numerical results were unaffected to six digits.
\end{remark}

Figure \ref{convergence_May_2025_fig}
shows the EW-ES efficient frontiers, computed using various grid sizes.  
The grid here refers to the grid for $b>0$.  There is an additional grid
(of the same size) for $b<0$.  The optimal control
is computed and stored, using the $\delta$-monotone PIDE
method.  Statistics are then generated using Monte Carlo (MC)
simulations, using the stored optimal controls, with $2.56\times 10^6$ MC simulations
being used.  
More precisely, Monte Carlo methods are
used only as a forward policy evaluation tool, with
a known policy.  The optimal
policies are obtained as part of the backward dynamic
programming algorithm.

The curves for all grid
sizes essentially overlap, indicating convergence for practical
purposes.  All subsequent results use the $2048 \times 2048$ grid.
For comparison, we also show the results for the \citet{Bengen1994} policy,
which is clearly much less efficient than the optimal policy.

\begin{figure}[htb!]
\centerline{%
\includegraphics[width=3.0in]{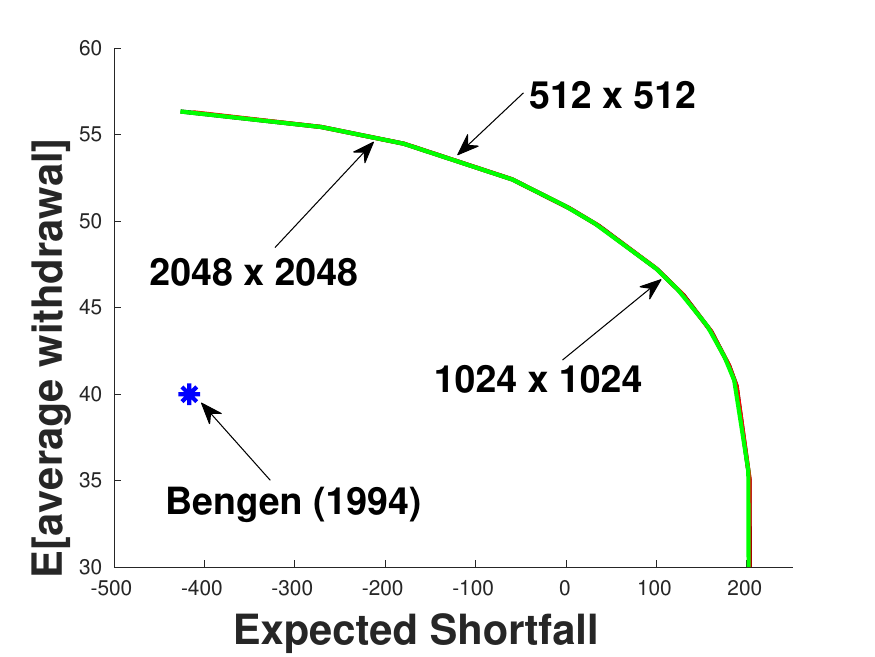}
}
\caption{
EW-ES convergence test.
Real stock index: deflated real capitalization
weighted CRSP, real bond index: deflated 30 day T-bills. Scenario in
Table~\ref{base_case_1}. Parameters in Table~\ref{fit_params}.
The optimal
control is determined
by  using the method in Algorithm \ref{alg_monotone_final}.
Grid refers to the grid used in the PIDE solve,
$n_s \times n_b$, where $n_s$ is the number of
nodes in the $\log s$ direction, and $n_b$ is the number of nodes in the
$\log b$ direction. Units: thousands of dollars (real).
The controls are stored, and then the final results
are obtained using a Monte Carlo method,
with $2.56 \times 10^6$ simulations.
Maximum fraction in stocks $p_{\max} = 1.0$.
Bengen (1994) refers to a constant weight in stocks $p = 0.5$,
rebalanced annually, and constant yearly withdrawals of $40$ per year.
All amounts are inflation adjusted.
}
\label{convergence_May_2025_fig}
\end{figure}

Table \ref{conservative_accuracy} shows detailed convergence results for a
single point on the EW-ES frontier ($\pp_{\max} =1.0, \kappa = 0.866$).
The optimal controls computed using Algorithm \ref{alg_monotone_final} are stored,
and then used in Monte Carlo (MC) simulations.  The value function
appears to converge smoothly (for the PIDE method).  

\begin{table}[hbt!]
\begin{center}
\begin{tabular}{lccc|cc} \toprule
 & \multicolumn{3}{|c|}{Algorithm in \ref{alg_monotone_final}}
 & \multicolumn{2}{c}{Monte Carlo} \\ \midrule
Grid & ES & $~~E[\sum_i \qq_i]/M$ & ~~Value Function (\ref{pcee_inter}) & ES
 & $E[\sum_i \qq_i]/M$ \\ \midrule
$~512 ~\times ~512$   & ~~~-14.6176 &~~ 51.1162 & 1520.941 &  ~-10.302 & 51.1507 \\
$1024 \times 1024$ & ~~~-9.16685 & ~~51.0706  & 1524.251 &~ -7.5173 & 51.0781 \\
$2048 \times 2048$ & ~~~-4.84764 & ~~50.9780  & 1525.179 & ~-3.8866 & 50.9762 \\
\bottomrule
\end{tabular}
\caption{
EW-ES convergence test.
Real stock index: deflated real capitalization
weighted CRSP, real bond index: deflated 30 day T-bills. Scenario in
Table~\ref{base_case_1}. Parameters in Table~\ref{fit_params}.
Maximum fraction in stocks $p_{\max} = 1.0$.
$\kappa = 0.8660$.
Other information given in caption for Figure \ref{convergence_May_2025_fig}.
}
\label{conservative_accuracy}
\end{center}
\end{table}

It is also interesting to note that, in all our tests
using the parameters in Table \ref{base_case_numerical},
we find that the wraparound error condition (\ref{wrap_error_local_1a}) is bounded by
\begin{eqnarray}
         \Delta x_1 \Delta x_2 \sum_{( j ,k) \in \mathcal{D}^{\dagger} - \mathcal{D} } 
                  | \tilde{g}_{j,k}^{\dagger} | 
         & < &
    10^{-14}  ~.
              \label{wrap_numerical}
\end{eqnarray}
Note that
\begin{eqnarray*}
         \Delta x_1 \Delta x_2 \sum_{ (j,k) \in \mathcal{D}^{\dagger}} | \tilde{g}^{\dagger}_{j,k} |
       \leq 1 + \delta \simeq 1 ~ 
\end{eqnarray*}
with $\delta$ given in Table  \ref{base_case_numerical}.
Consequently, the  numerical result  (\ref{wrap_numerical})  
indicates that the wrap-around error is only slightly larger than double precision
machine epsilon.

\subsection{Efficient Frontiers: Synthetic market}
Figure \ref{frontier_May_2025_fig} shows the efficient frontiers computed
using various value of $\pp_{\max}$ (labelled $pmax$ on the
Figure), the maximum values  of the fraction
in stocks. As must
be true mathematically, the curves with higher values of $\pp_{\max}$
plot above curves with smaller $\pp_{\max}$.  However, a striking 
result is that the curves are all fairly tightly clustered,
in comparison to the \citet{Bengen1994}  strategy.
In particular, it seems reasonable to target an $ES \simeq 0$.
The expected annual withdrawal for all values of $\pp_{\max}$ are
very close for $ES = 0$.\footnote{Although from Remark \ref{reverse_mortgage_hedge},
even $ES \simeq -200K$ is also acceptable. Hence a target of $ES = 0$
is very conservative.}  However, this is not quite
a free lunch, as we will see when we examine the results more
closely.

\begin{figure}[htb!]
\centerline{%
\includegraphics[width=3.0in]{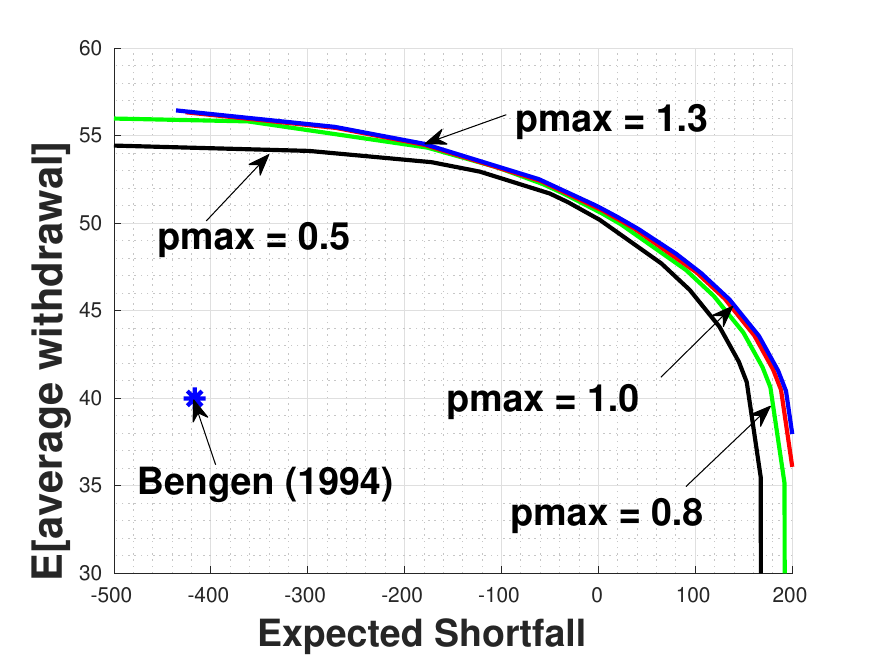}
}
\caption{
Comparison of maximum leverage constraint $p_{\max}$,
synthetic market.
Real stock index: deflated real capitalization
weighted CRSP, real bond index: deflated 30 day T-bills. Scenario in
Table~\ref{base_case_1}. Parameters in Table~\ref{fit_params}.
The optimal
control is determined
using Algorithm \ref{alg_monotone_final}.
The controls are stored  (using $2048 \times 2048$ grid), and then the final results
are obtained using a Monte Carlo method,
with $2.56 \times 10^6$ simulations.
Bengen (1994) refers to a constant weight in stocks $p = 0.5$,
rebalanced annually, and constant  yearly withdrawals of $40K$ per year.
All amounts are inflation adjusted.
}
\label{frontier_May_2025_fig}
\end{figure}


\subsection{Heat Maps: Synthetic Market} 
In order to gain some insight into the optimal controls,
we plot the heat maps of the optimal asset allocation and
the optimal withdrawals, for the cases $\pp_{\max} = 1.3, 1.0, 0.5$
in Figures \ref{plot_heat_one_point_three}, \ref{plot_heat_one_point_zero},
\ref{plot_heat_point_five}.
For each case, we choose the point on the efficient frontier
so that  $ES  \simeq 0$, since this is
an {\em interesting} point on the efficient frontier.

First, note that Figures \ref{plot_heat_map_qplus_one_point_three_fig},
\ref{plot_heat_map_qplus_one_point_zero_fig} and
\ref{plot_heat_map_qplus_zero_point_five_fig} show that
the optimal withdrawal strategies  for a wide range of $\pp_{\max}$ 
are very similar.  In all cases, the withdrawal control
is approximately {\em bang-bang}, that is,  the optimal
policy is only to withdraw the maximum or minimum amounts.
In the continuous withdrawal limit, the withdrawal control
can be proven to be bang-bang, \citet{Forsyth_2021_b}.
We have been unable to prove that this property holds
for discrete withdrawals.  However, numerical studies
appear to show that the transition zone between $q_{\min}$
and $q_{\max}$ becomes smaller as the grid is refined.
We conjecture that the true converged withdrawal
control is bang-bang even for discrete rebalancing/withdrawals.
The lefthand plots of Figures  \ref{plot_heat_one_point_three}, \ref{plot_heat_one_point_zero},
\ref{plot_heat_point_five} show the optimal allocation strategies.
We use the same colour scale for all plots, with the maximum allocation
to stocks being 130\%.  It is interesting to see that the allocation
strategies are quite similar as long as we restrict attention to the areas
in the heat maps above the 5th wealth percentile.  Large differences in
the allocation appear below the 5th wealth percentile.

\begin{figure}[htb!]
\centerline{
\begin{subfigure}[t]{.4\linewidth}
\centering
\includegraphics[width=\linewidth]{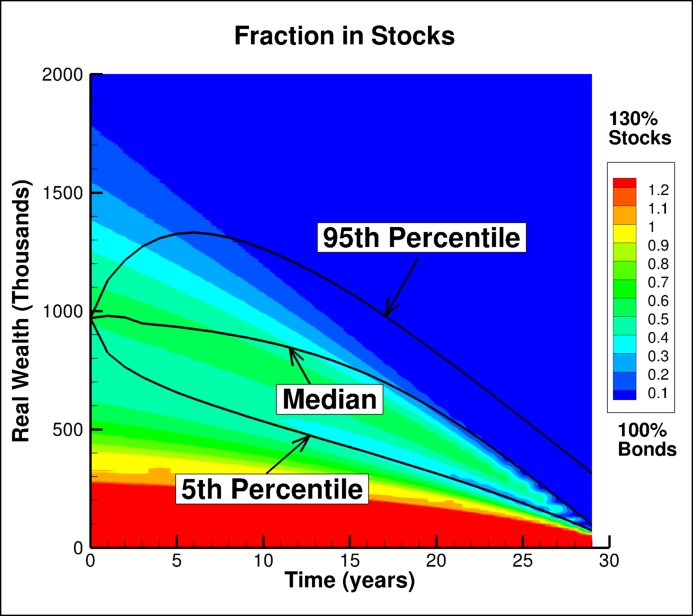}
\caption{Fraction in stocks}
\label{plot_heat_map_one_point_three_fig}
\end{subfigure}
\hspace{.05\linewidth}
\begin{subfigure}[t]{.4\linewidth}
\centering
\includegraphics[width=\linewidth]{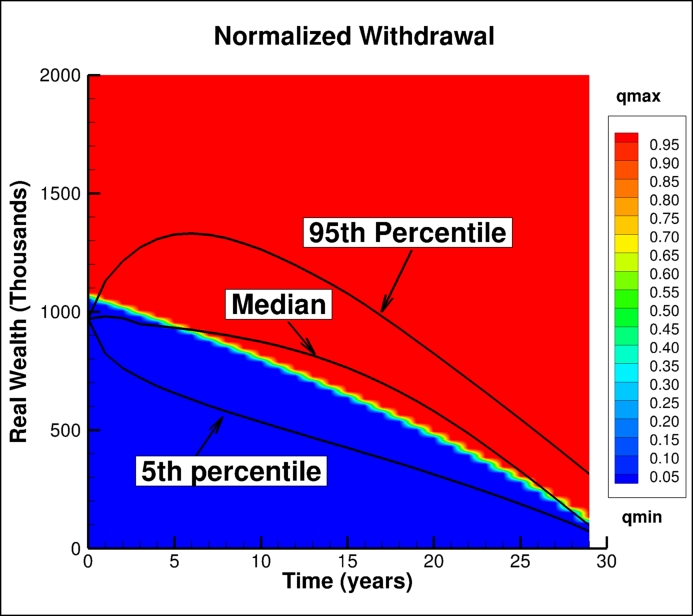}
\caption{Withdrawals}
\label{plot_heat_map_qplus_one_point_three_fig}
\end{subfigure}
}
\caption{
$p_{\max} = 1.3$. Heat map of the controls: fraction in stocks and withdrawals,
computed using Algorithm \ref{alg_monotone_final}. Real
capitalization weighted CRSP index, and real 30-day T-bills. Scenario
given in Table~\ref{base_case_1}. Control computed and stored from the
Algorithm \ref{alg_monotone_final} in the synthetic market. $q_{min} = 30, q_{\max} =
60$ (per year). $\kappa = 0.8583, EW \simeq 50.9$, $ES \simeq 0.96$.  Percentiles from bootstrapped
historical market.
Normalized withdrawal $(q - q_{\min})/(q_{\max} - q_{\min})$. Units:
thousands of dollars.
}
\label{plot_heat_one_point_three}
\end{figure}

\begin{figure}[htb!]
\centerline{
\begin{subfigure}[t]{.4\linewidth}
\centering
\includegraphics[width=\linewidth]{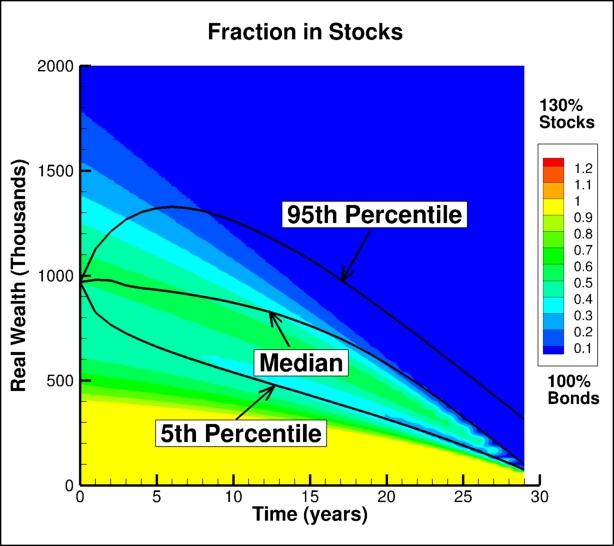}
\caption{Fraction in stocks}
\label{plot_heat_map_one_point_zero_fig}
\end{subfigure}
\hspace{.05\linewidth}
\begin{subfigure}[t]{.4\linewidth}
\centering
\includegraphics[width=\linewidth]{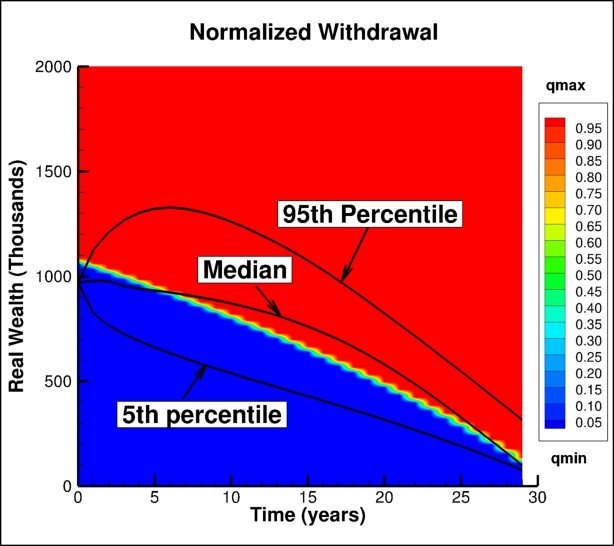}
\caption{Withdrawals}
\label{plot_heat_map_qplus_one_point_zero_fig}
\end{subfigure}
}
\caption{
$p_{\max} = 1.0$. Heat map of the controls.  $\kappa = 0.8860, EW = 50.7$, $ES = 4.6$.
Other information as in caption for Figure \ref{plot_heat_one_point_three}.
}
\label{plot_heat_one_point_zero}
\end{figure}

\begin{figure}[htb!]
\centerline{
\begin{subfigure}[t]{.4\linewidth}
\centering
\includegraphics[width=\linewidth]{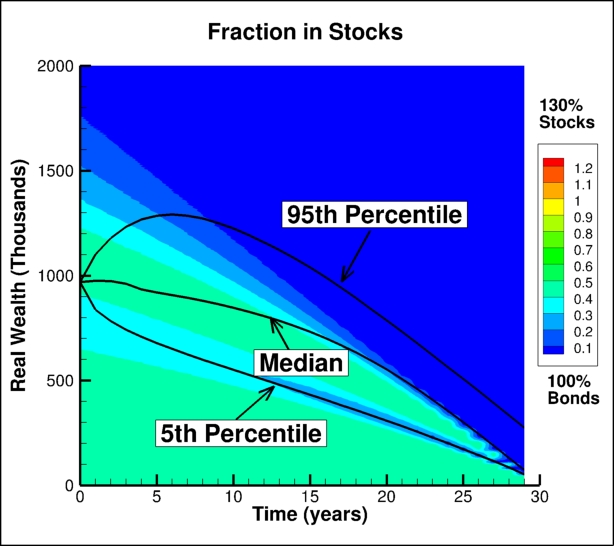}
\caption{Fraction in stocks}
\label{plot_heat_map_zero_point_five_fig}
\end{subfigure}
\hspace{.05\linewidth}
\begin{subfigure}[t]{.4\linewidth}
\centering
\includegraphics[width=\linewidth]{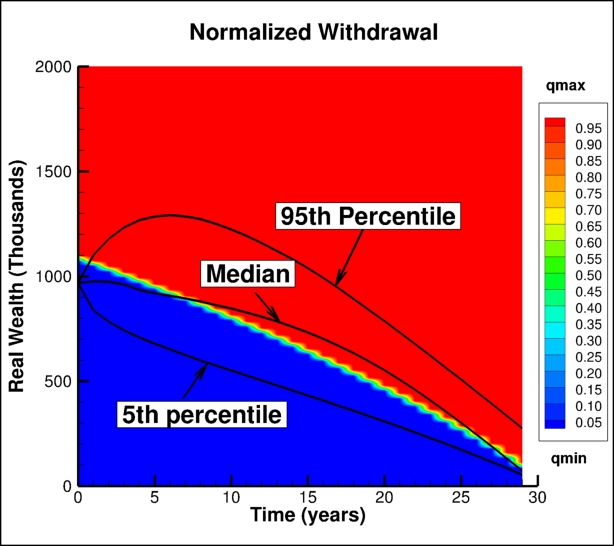}
\caption{Withdrawals}
\label{plot_heat_map_qplus_zero_point_five_fig}
\end{subfigure}
}
\caption{
$p_{\max} = 0.5$. Heat map of the controls. $\kappa = 1.0, EW \simeq 50.2$, $ES = 1.25$.
Other information as in caption for Figure \ref{plot_heat_one_point_three}.
}
\label{plot_heat_point_five}
\end{figure}


\subsection{Percentiles Versus Time: Synthetic Market}
Figures \ref{percentiles_one_point_three_synthetic}.
\ref{percentiles_point_one_point_zero_synthetic} and
\ref{percentiles_point_five_synthetic} show the percentiles
of fraction in stocks, wealth, and withdrawals versus time.
We select the point on the efficient frontier for
each case so that $ES \simeq 0$.

Rather surprisingly, the percentiles fraction
in stocks and percentiles wealth are very similar,
for all values of $\pp_{\max}$.
However, we do see some differences in 
the withdrawal percentiles (the rightmost panel in each plot).
The median withdrawal is slower to increase to the maximum
for $\pp_{\max} = 0.5$ compared to the case
with $\pp_{\max} = 1.0$.

In addition, we can see that even for $\pp_{\max} = 1.3$, the
95th percentile for stock allocation is less than $0.6$ (Figure
\ref{percentile_control_pmax_one_point_three_fig}),
suggesting that with this level of ES risk, the optimal
policy rarely allocates a large fraction
to stocks. 

\begin{figure}[tb]
\centerline{%
\begin{subfigure}[t]{.33\linewidth}
\centering
\includegraphics[width=\linewidth]{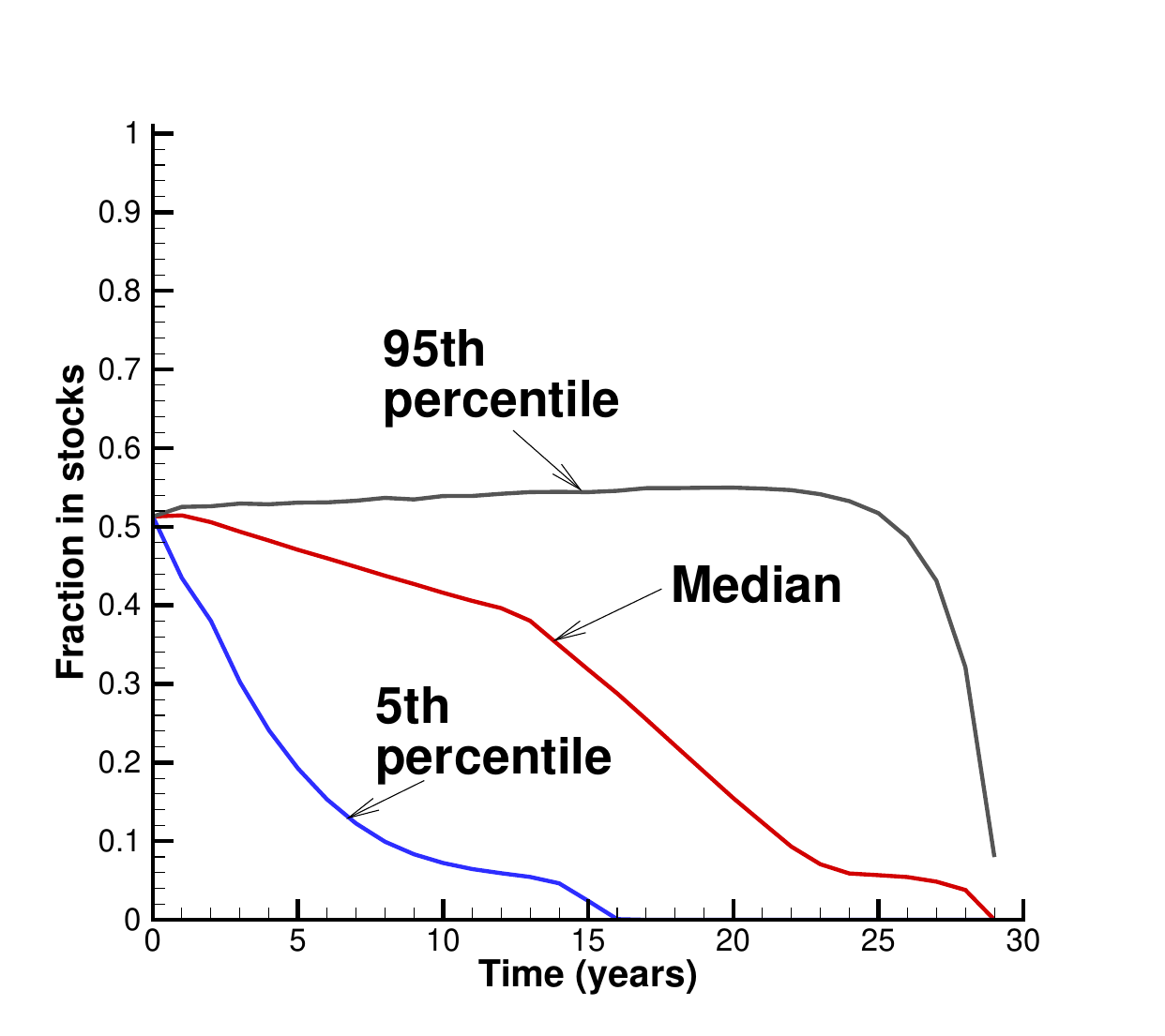}
\caption{Percentiles fraction in stocks}
\label{percentile_control_pmax_one_point_three_fig}
\end{subfigure}
\begin{subfigure}[t]{.33\linewidth}
\centering
\includegraphics[width=\linewidth]{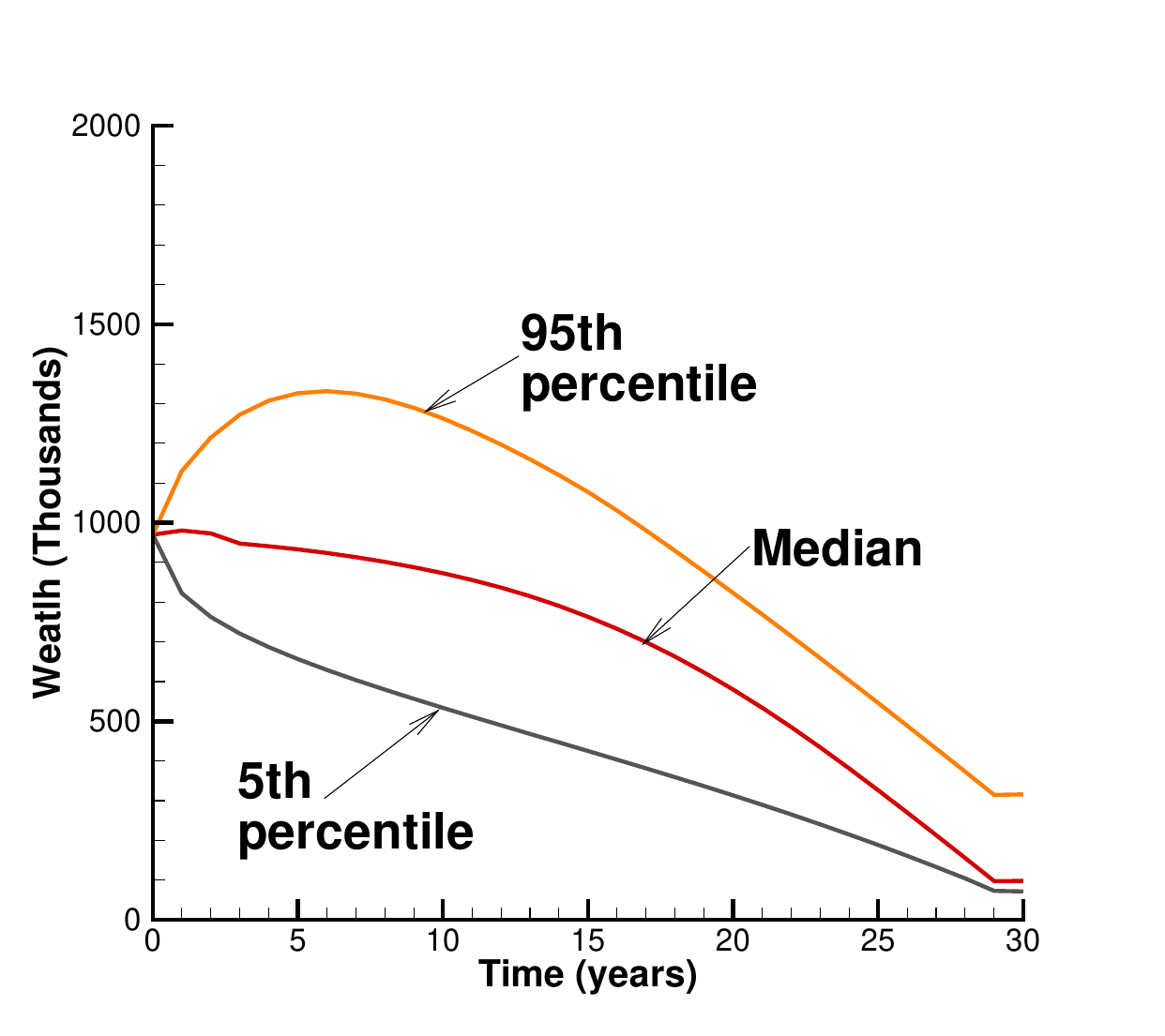}
\caption{Percentiles  wealth}
\label{percentile_wealth_one_point_three_fig}
\end{subfigure}
\begin{subfigure}[t]{.33\linewidth}
\centering
\includegraphics[width=\linewidth]{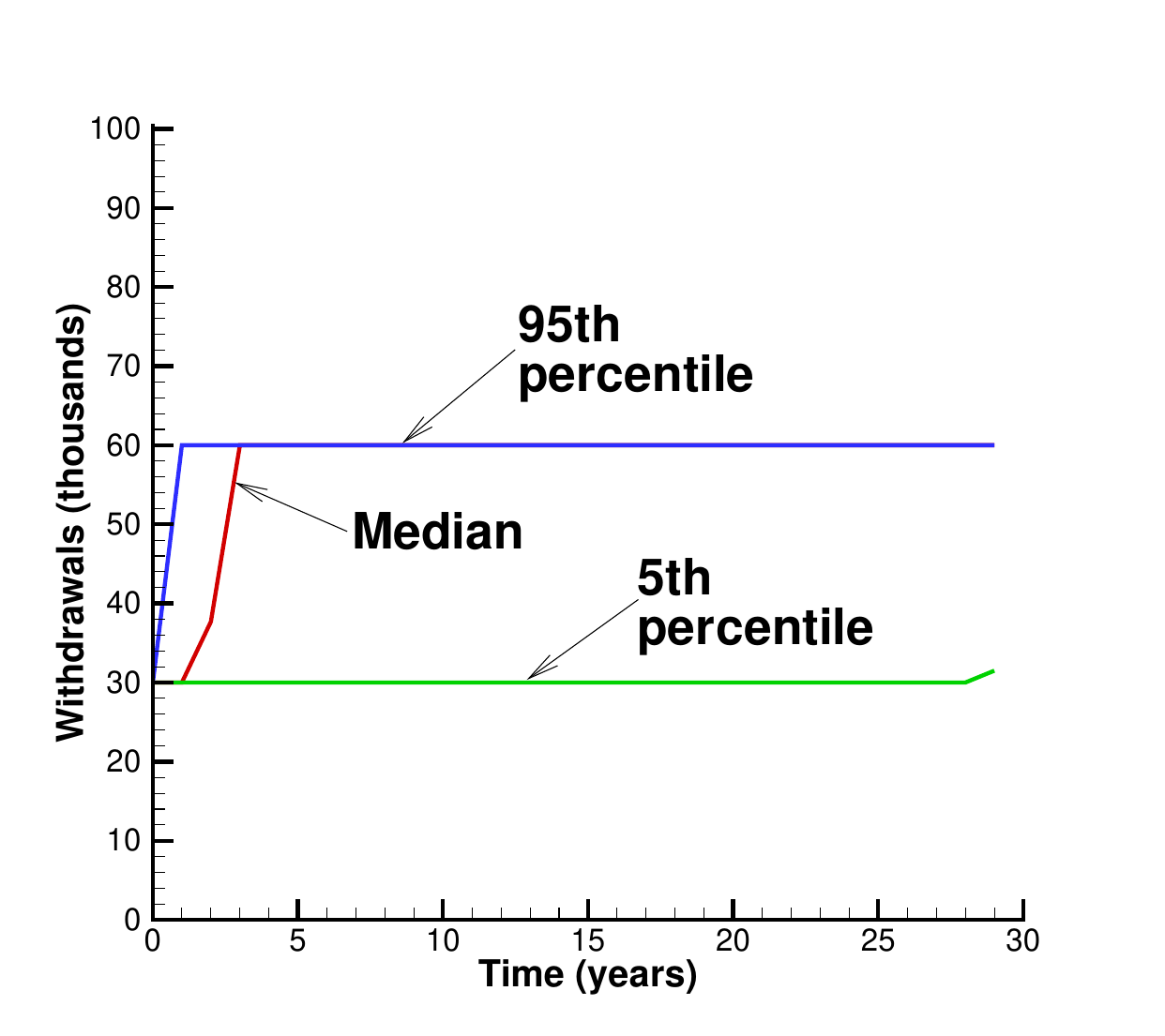}
\caption{Percentiles withdrawals}
\label{percentile_qplus_one_point_three_fig}
\end{subfigure}
}
\caption{
$p_{\max} = 1.3$.
Percentiles of the controls.
Scenario in Table \ref{base_case_1}.
Parameters based on the real CRSP index,
and real 30-day T-bills (see Table~\ref{fit_params}). Control computed
and stored from the Algorithm \ref{alg_monotone_final} in the synthetic market.
$q_{min} = 30, q_{\max} = 60$ (per year), $EW \simeq 50.9$, $ES = 0.96$ ($\kappa = 0.8583 $).
Units: thousands of dollars.
}
\label{percentiles_one_point_three_synthetic}
\end{figure}

\begin{figure}[tb]
\centerline{%
\begin{subfigure}[t]{.33\linewidth}
\centering
\includegraphics[width=\linewidth]{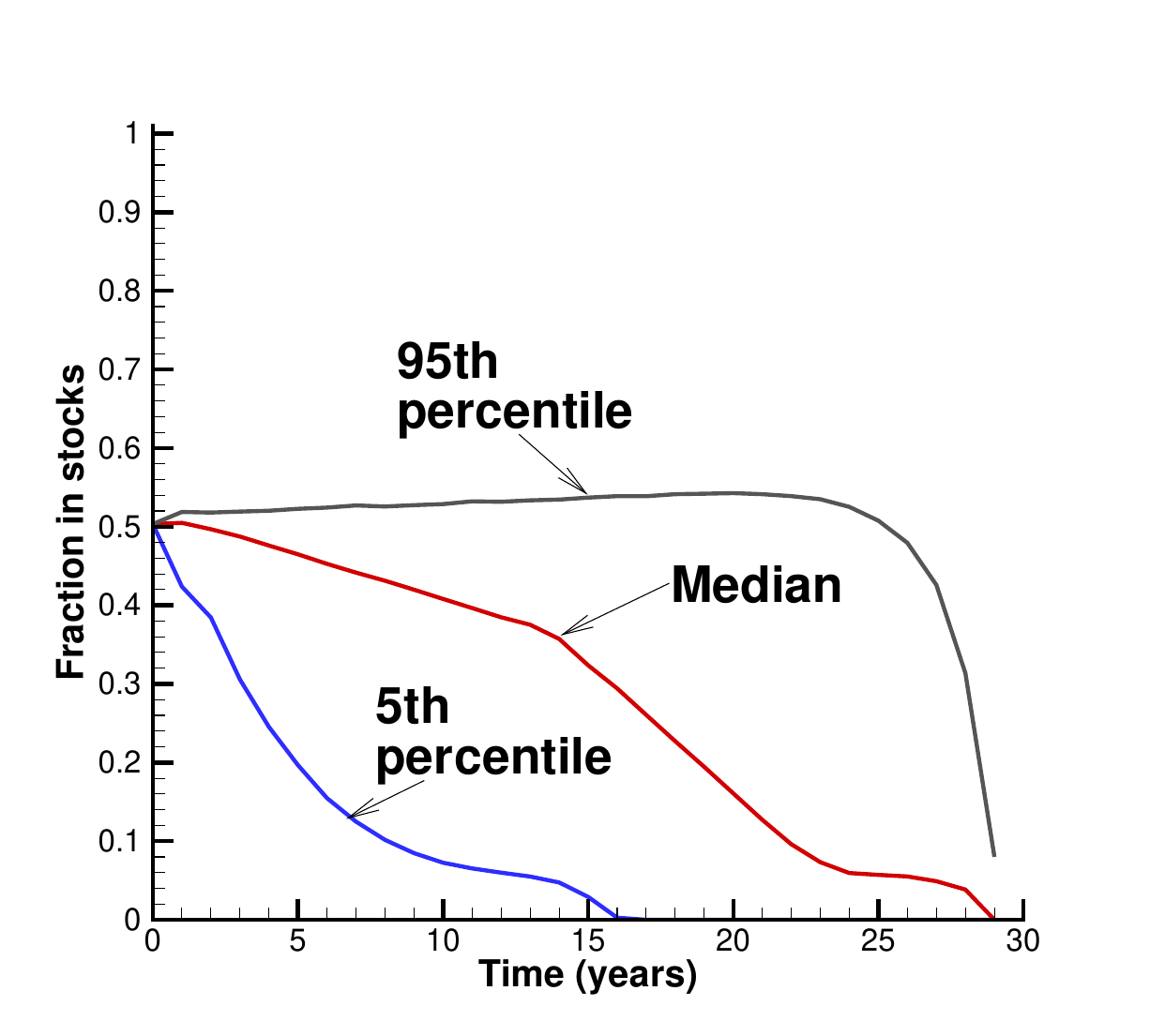}
\caption{Percentiles fraction in stocks}
\label{percentile_control_pmax_one_point_zero_fig}
\end{subfigure}
\begin{subfigure}[t]{.33\linewidth}
\centering
\includegraphics[width=\linewidth]{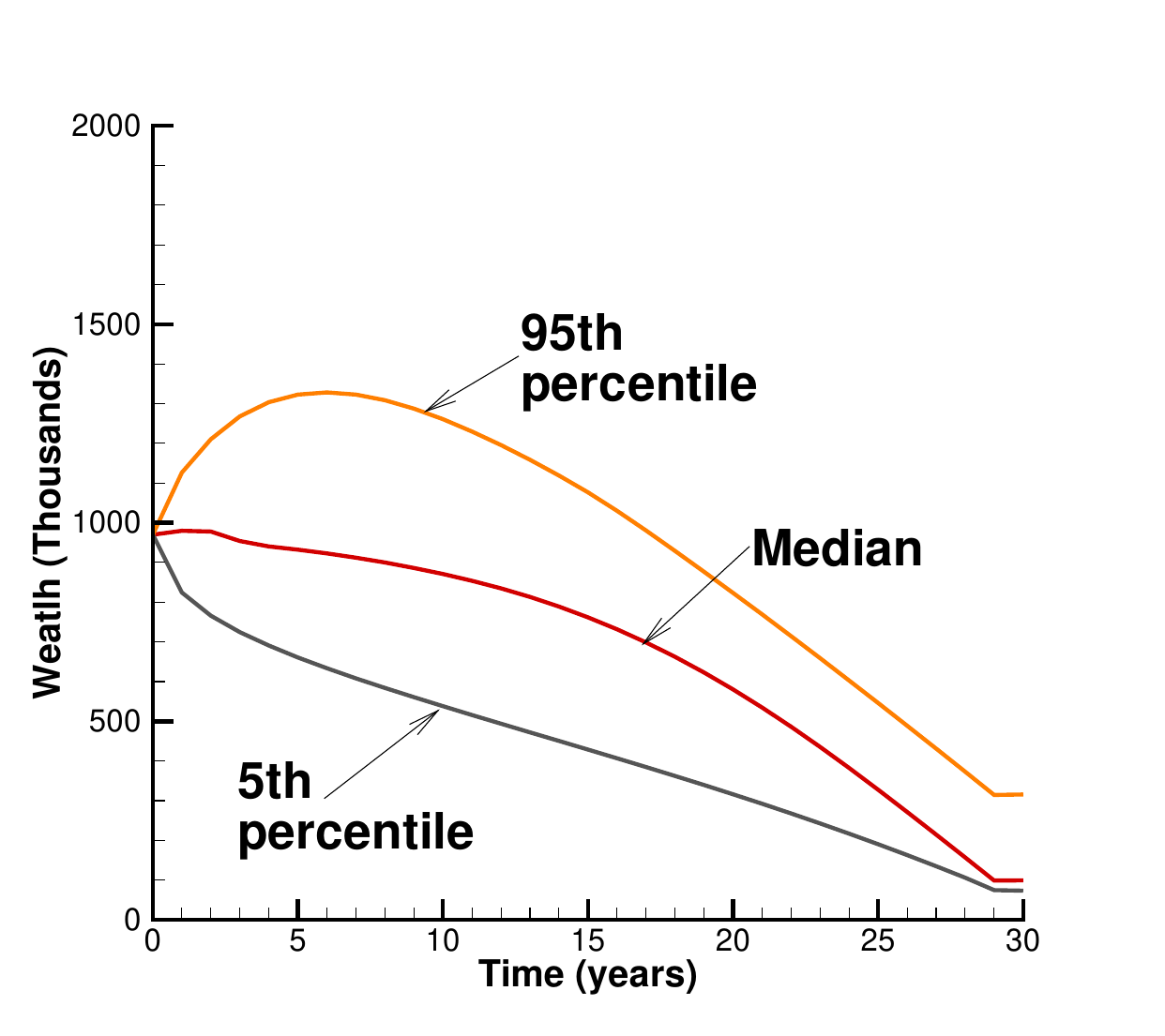}
\caption{Percentiles  wealth}
\label{percentile_wealth_one_point_zero_fig}
\end{subfigure}
\begin{subfigure}[t]{.33\linewidth}
\centering
\includegraphics[width=\linewidth]{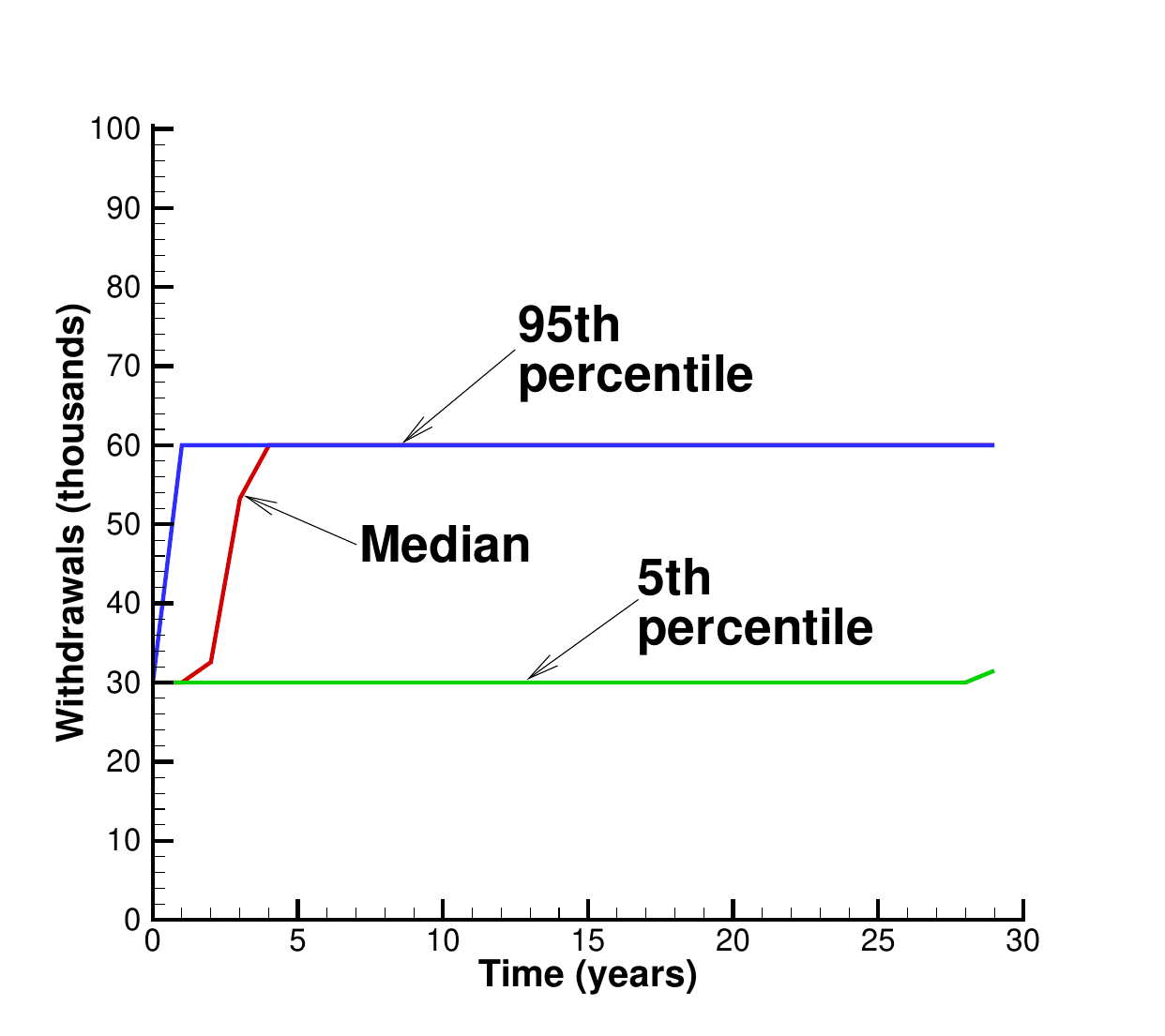}
\caption{Percentiles withdrawals}
\label{percentile_qplus_one_point_zero_fig}
\end{subfigure}
}
\caption{
$p_{\max} = 1.0$. Percentiles of the controls.
$EW \simeq 50.7$, $ES = 4.6$ ($\kappa = 0.8860$).
Other information as in caption for Figure \ref{percentiles_one_point_three_synthetic}.
}
\label{percentiles_point_one_point_zero_synthetic}
\end{figure}

\begin{figure}[tb]
\centerline{%
\begin{subfigure}[t]{.33\linewidth}
\centering
\includegraphics[width=\linewidth]{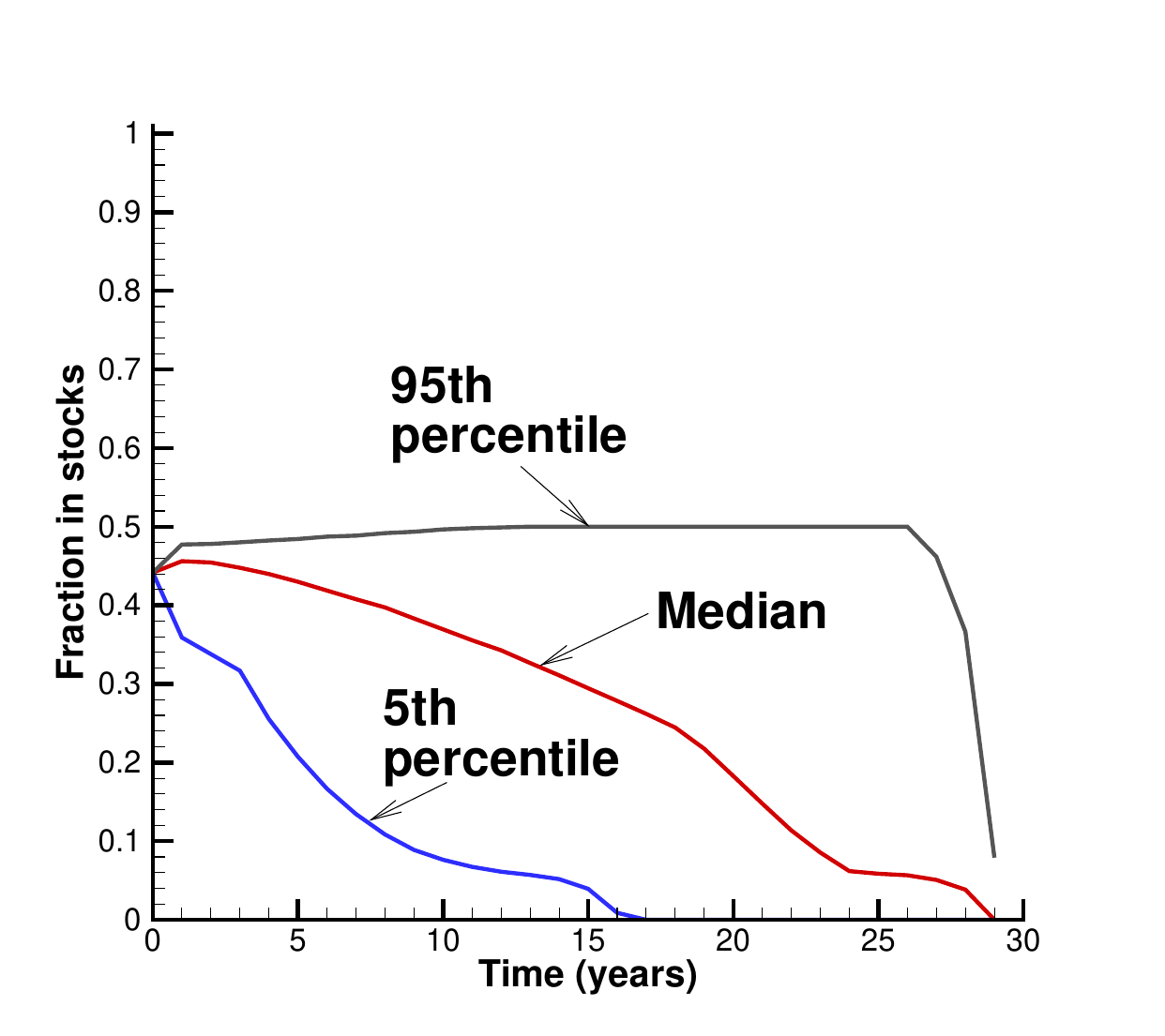}
\caption{Percentiles fraction in stocks}
\label{percentile_control_pmax_zero_point_five_fig}
\end{subfigure}
\begin{subfigure}[t]{.33\linewidth}
\centering
\includegraphics[width=\linewidth]{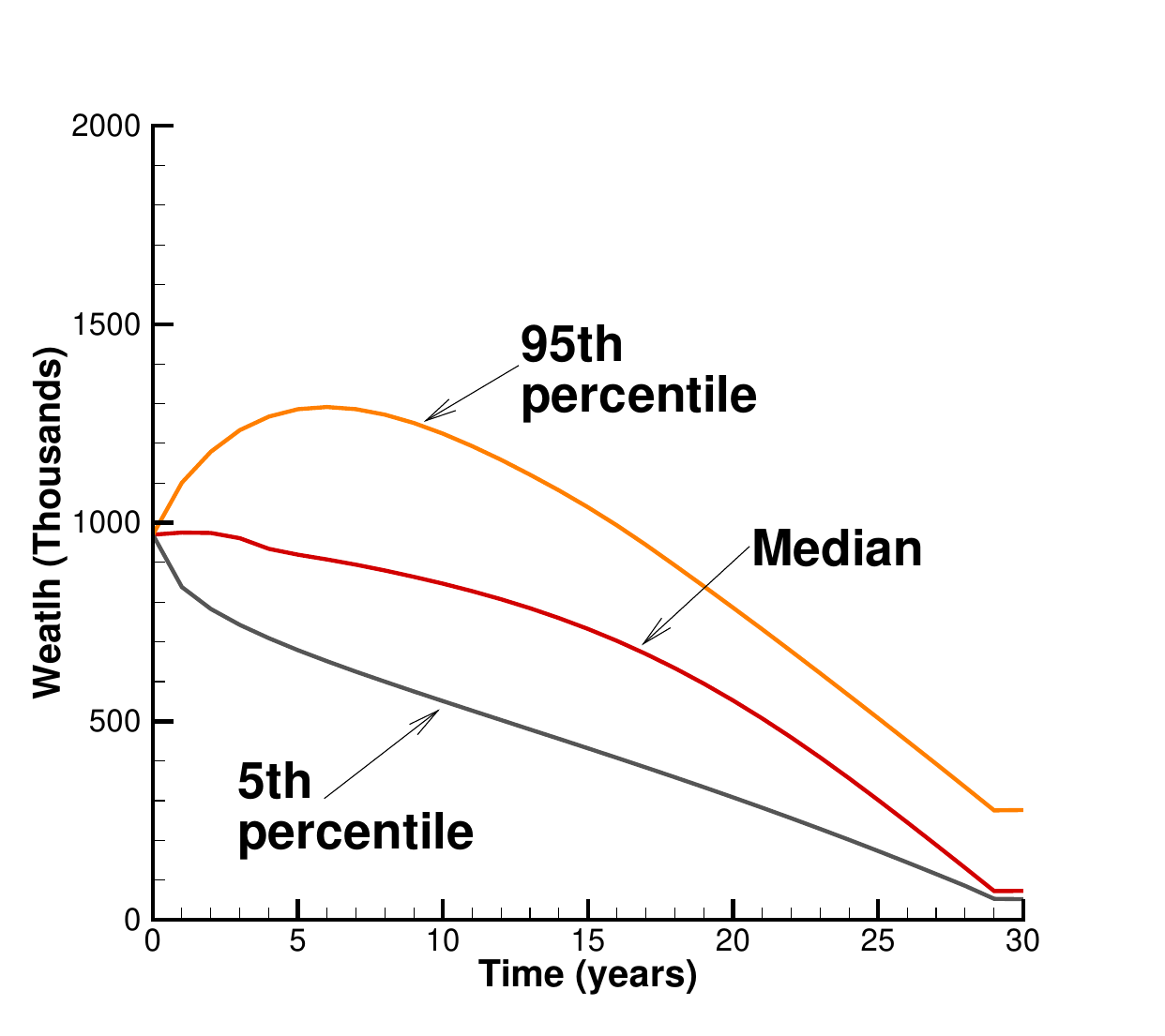}
\caption{Percentiles  wealth}
\label{percentile_wealth_pmax_zero_point_five_fig}
\end{subfigure}
\begin{subfigure}[t]{.33\linewidth}
\centering
\includegraphics[width=\linewidth]{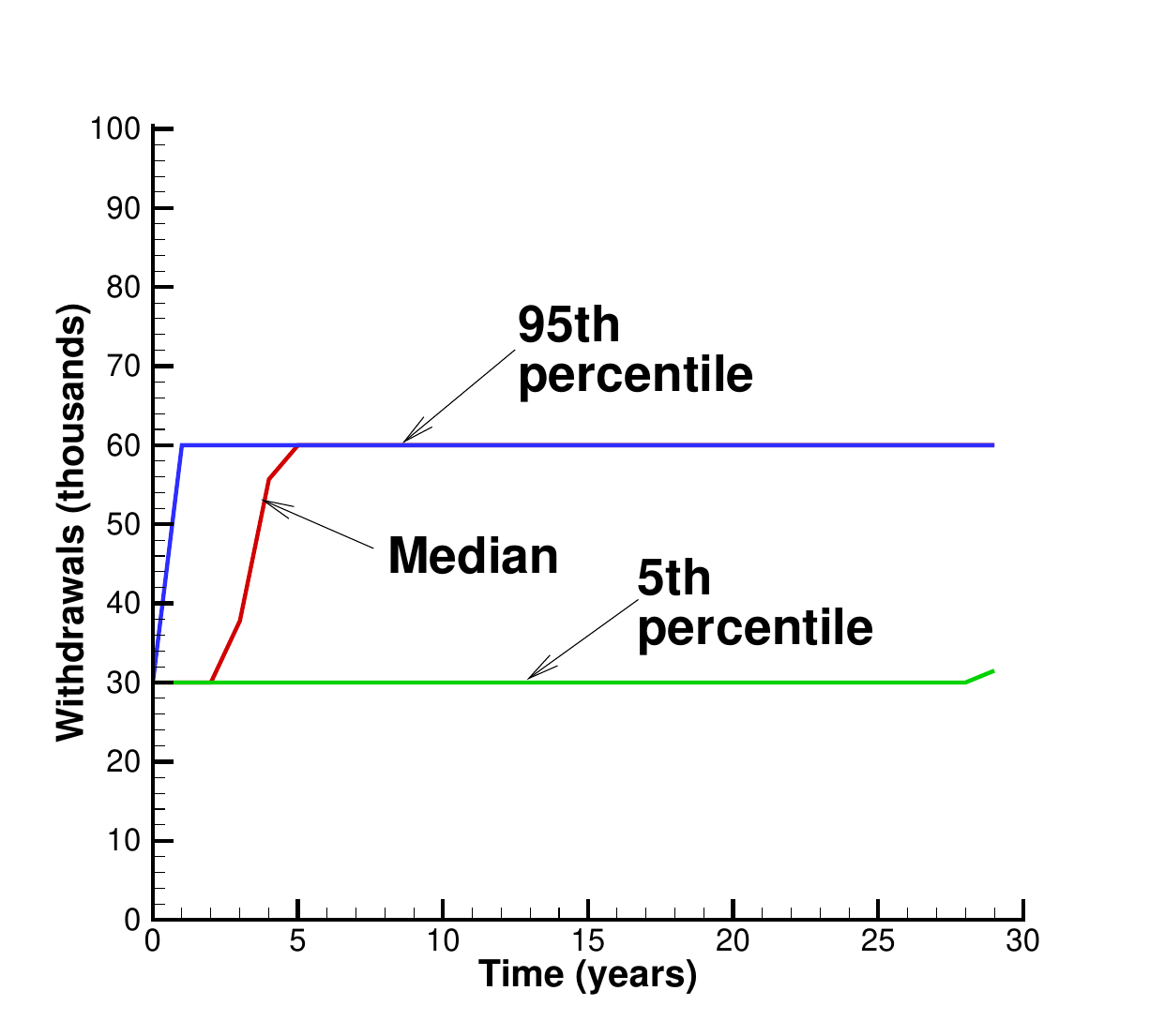}
\caption{Percentiles withdrawals}
\label{percentile_qplus_pmax_zero_point_five_fig}
\end{subfigure}
}
\caption{
$p_{\max} = 0.5$. Percentiles of the controls.
$EW = 50.2$, $ES = 1.25 $ ($\kappa = 1.0$).
Other information as in caption of Figure \ref{percentiles_one_point_three_synthetic}.
}
\label{percentiles_point_five_synthetic}
\end{figure}



\section{Bootstrap results}

\subsection{Historical Market}\label{boot_section}

In order to check on the robustness of the above
results, we proceed as follows.
We compute and store the optimal
controls based on the parametric model (\ref{jump_process_stock}-\ref{jump_process_bond}) as for the
synthetic market case. However, we compute statistical quantities using
the stored controls, but using bootstrapped historical return data
directly. In this case, we make no assumptions concerning the stochastic
processes followed by the stock and bond indices. We remind the
reader that all returns are inflation-adjusted. We use the stationary
block bootstrap method \citep{politis1994,politis2004,politis2009,
Cogneau2010,dichtl2016,Scott_2022,Simonian_2022,Cederburg_2022}.

A key parameter is the expected blocksize. Sampling the data in blocks
accounts for serial correlation in the data series. We use the algorithm
in \citet{politis2009} to determine the optimal blocksize for the bond
and stock returns separately, see Table \ref{auto_blocksize}. However,
in our simulations, we use a
paired sampling approach to simultaneously draw returns from both time
series. In this case, a reasonable estimate for the blocksize for the
paired resampling algorithm would be about $2.0$ years. We will give
results for a range of blocksizes as a check on the robustness of the
bootstrap results. Detailed pseudo-code for block bootstrap resampling
is given in \citet{Forsyth_Vetzal_2019a}.

\begin{table}[h!]
\begin{center}

\text{Optimal expected block size for bootstrap resampling historical data} \\
\medskip
{\small
\begin{tabular}{lc} \toprule
Data          & Optimal expected \\
                     & block size $\hat{b}$ (months) \\ \midrule
  30-day T-bill   &  51 \\
  CRSP cap weighted index &  4.0 \\
\bottomrule
\end{tabular}
}
\end{center}
\caption{Optimal expected blocksize $\hat{b}=1/v$, from \cite{politis2009}. Range of historical
data is between 1926:1 and 2024:12.
The blocksize is a draw from a geometric distribution with $Pr(b = k) = (1-v)^{k-1} v$.
Expected blocksize estimate algorithm from \cite{politis2004,politis2009}
using market data from CRSP.
\label{auto_blocksize}
}
\end{table}

Figure \ref{compare_bootstrap_frontiers} plots (i) the synthetic efficient frontier
and (ii) bootstrap efficient frontier, computed using
the synthetic market controls, tested in the historical market, for
various blocksizes.
Figure \ref{Bootstrap_ifrontier_one_point_three_June_2025_fig} shows
the results for $\pp_{\max} = 1.3$.  In this case, all the frontiers
are quite close, especially compared to the \citet{Bengen1994} strategy,
tested in the historical market.  The effect of using
different expected blocksizes in the block bootstrap
algorithm is quite small.  Overall, this plot suggests that
the optimal controls in this case are robust to model parameter
misspecification.  

The historical test for $\pp_{\max} = 0.5$ is shown in
Figure \ref{Bootstrap_ifrontier_half_June_2025_fig}.
Again, the historical frontiers using different blocksizes are 
quite close.  However, the historical efficient frontiers
do deviate somewhat from the synthetic market frontier,
for values of $ES > 100$.  This indicates that there is
some effect of model parameter misspecification, 
for large values of $ES$ in the case of $\pp_{\max} = 0.5$.
However, for values of $ES < 100$, the synthetic
and historical market frontiers are very close.  So, unless
the retiree is extremely risk averse, this may
not be a problem of practical concern.

\begin{figure}[tb]
\centerline{%
\begin{subfigure}[t]{.45\linewidth}
\centering
\includegraphics[width=\linewidth]{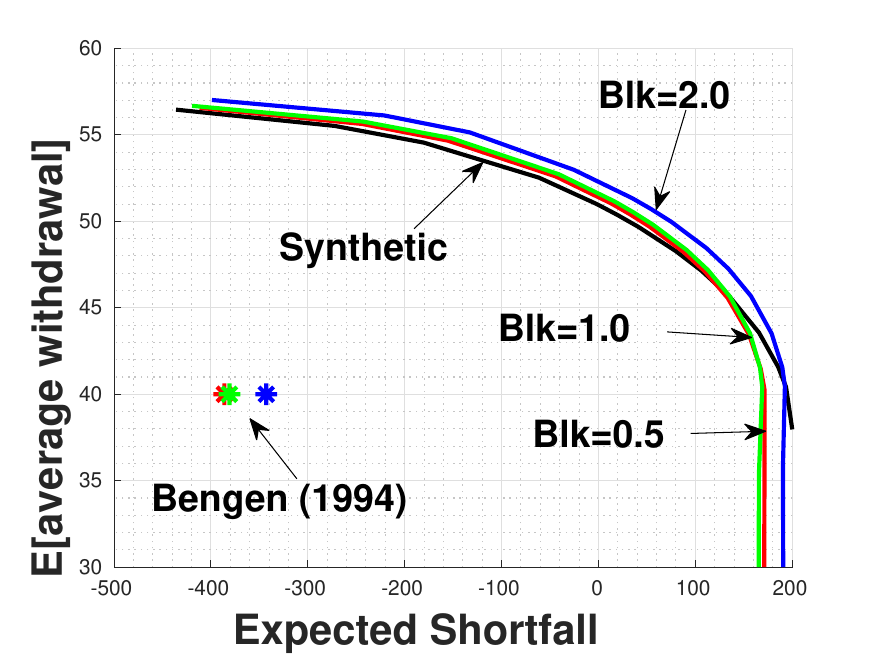}
\caption{
$p_{\max} = 1.3$
}
\label{Bootstrap_ifrontier_one_point_three_June_2025_fig}
\end{subfigure}
\begin{subfigure}[t]{.45\linewidth}
\centering
\includegraphics[width=\linewidth]{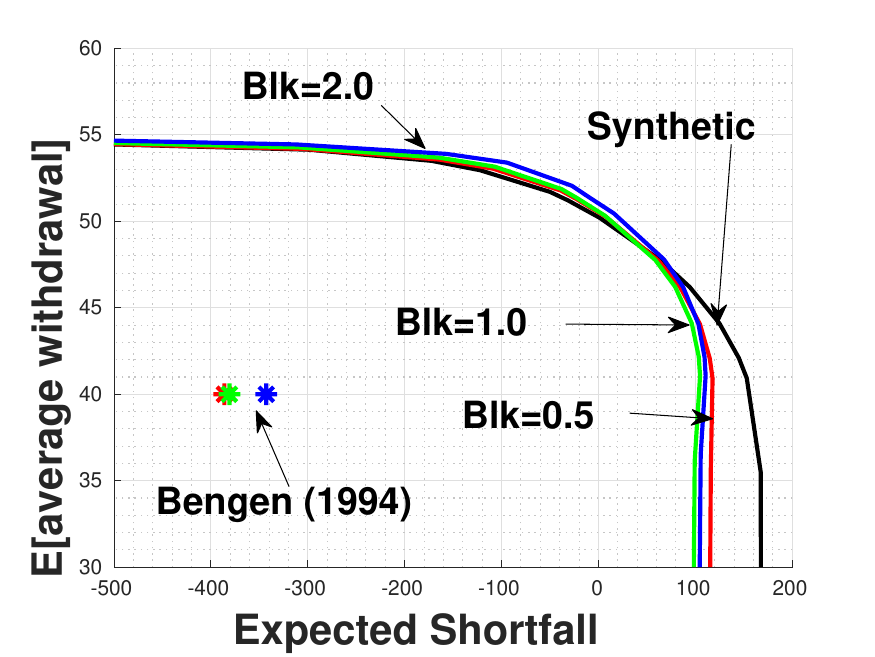}
\caption{$p_{\max} = 0.5$}
\label{Bootstrap_ifrontier_half_June_2025_fig}
\end{subfigure}
}
\caption{
Optimal controls were computed using the synthetic market model.  These
controls tested using bootstrapped historical data.
Expected blocksizes (years) shown.  $10^6$ bootstrap resamples.
Real stock index: deflated real capitalization
weighted CRSP, real bond index: deflated 30 day T-bills. Scenario in
Table~\ref{base_case_1}. Parameters in Table~\ref{fit_params}.
The Bengen control withdraws 40 per year (units: thousands of dollars), and rebalances
annually to 50\% bonds and 50\% stocks.  The Bengen results
are also shown for expected blocksizes of $0.5, 1.0, 2.0$ years.
}
\label{compare_bootstrap_frontiers}
\end{figure}

Figures \ref{percentiles_one_point_three_boot_all} and \ref{percentiles_point_five_boot}
show the percentiles of fraction in stocks, wealth, and
optimal withdrawals, tested in the historical market.
The two extreme cases: $\pp_{max} = 1.3$ and $\pp_{\max} = 0.5$ are shown.

The percentiles in stocks and wealth are fairly close,
for both cases.  However, a comparison of Figures \ref{percentile_qplus_one_point_three_boot_fig}
and Figure \ref{percentile_qplus_pmax_zero_point_five_fig_boot} indicates that
the median withdrawal increases from the minimum to the maximum over about
three years for the case with $\pp_{max} = 1.3$.  The comparable time
frame for $\pp_{\max} = 0.5$ is about four years.

Although most retirees would prefer  larger withdrawals during
the early years of retirement, slightly lower initial
withdrawals in exchange for never having more than 50\% stock allocation
might be seen as a reasonable tradeoff.

\begin{figure}[tb]
\centerline{%
\begin{subfigure}[t]{.33\linewidth}
\centering
\includegraphics[width=\linewidth]{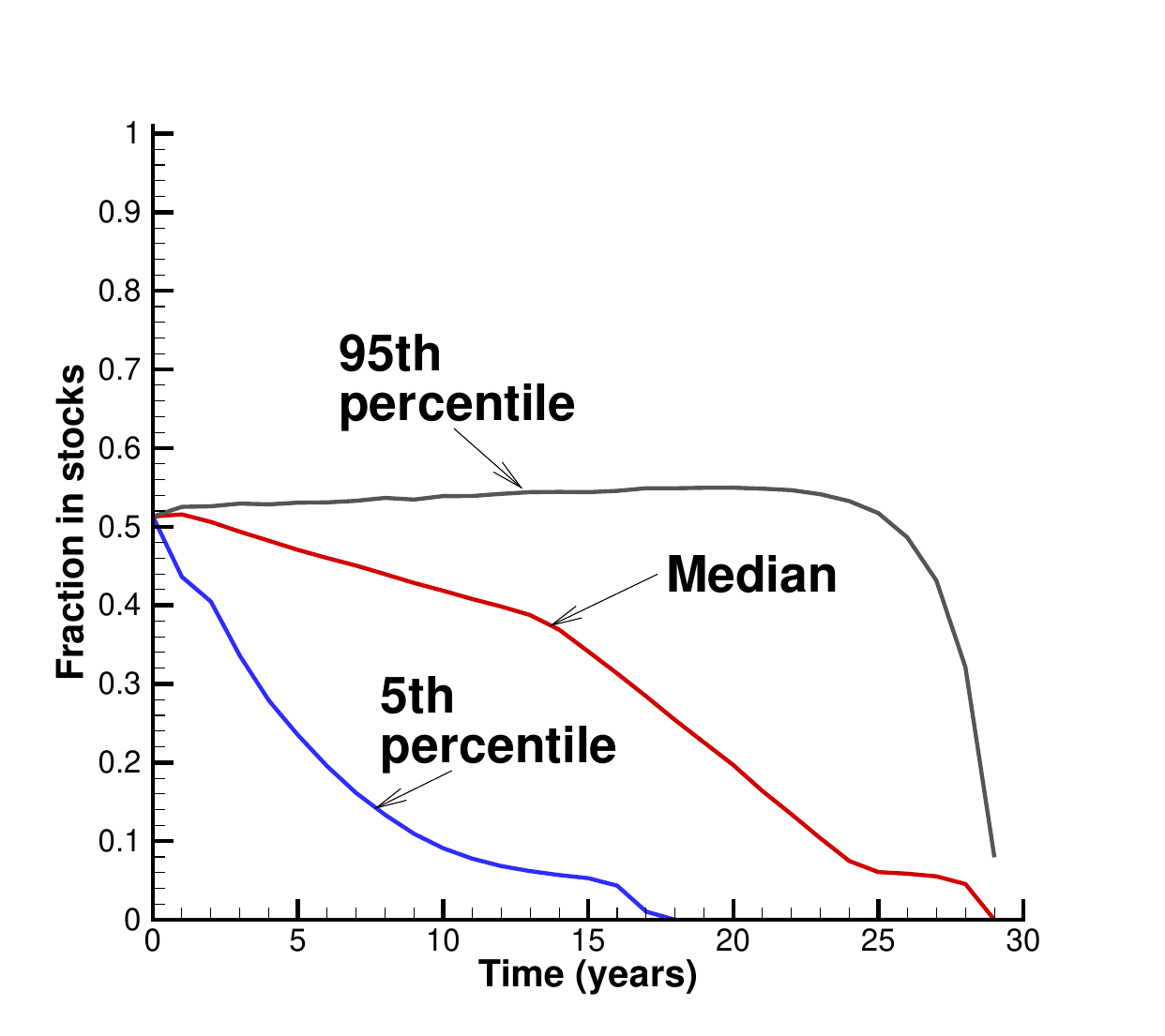}
\caption{Percentiles fraction in stocks}
\label{percentile_control_pmax_one_point_three_boot_fig}
\end{subfigure}
\begin{subfigure}[t]{.33\linewidth}
\centering
\includegraphics[width=\linewidth]{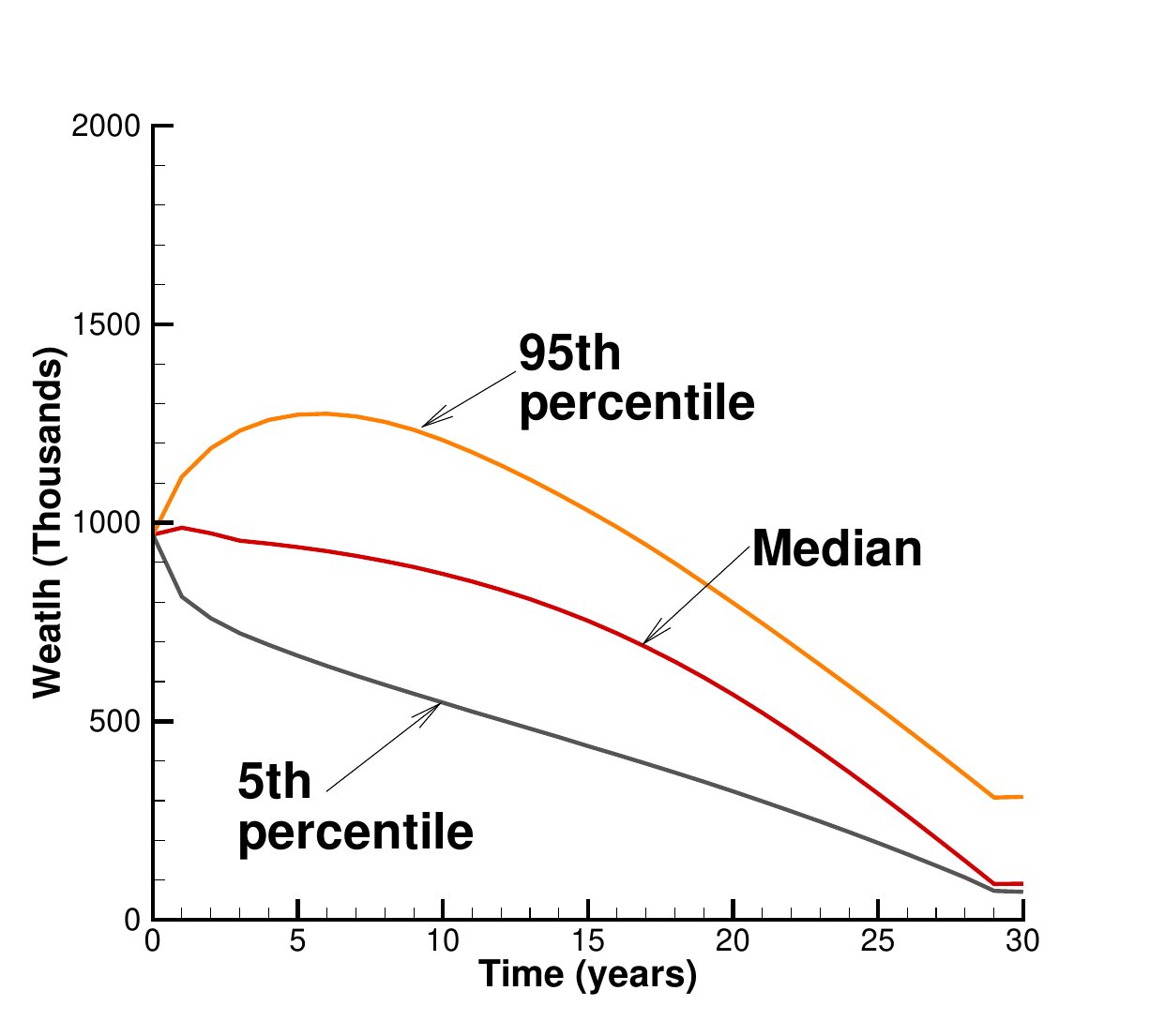}
\caption{Percentiles  wealth}
\label{percentile_wealth_one_point_three_boot_fig}
\end{subfigure}
\begin{subfigure}[t]{.33\linewidth}
\centering
\includegraphics[width=\linewidth]{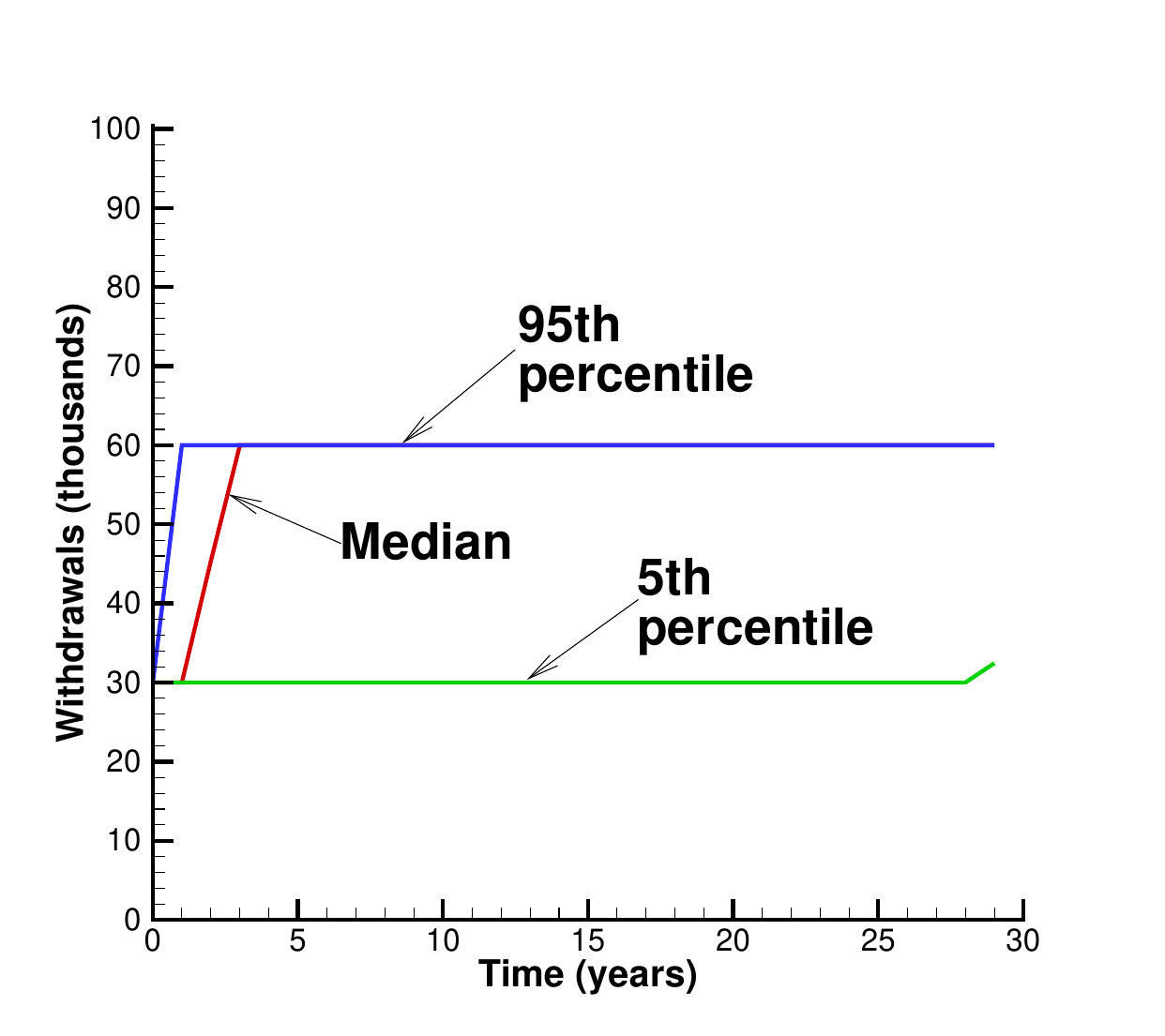}
\caption{Percentiles withdrawals}
\label{percentile_qplus_one_point_three_boot_fig}
\end{subfigure}
}
\caption{
$p_{\max} = 1.3$. Percentiles of the controls.
$EW = 51.7$, $ES = 19$ ($\kappa = 0.8583$).
Median withdrawal at $q_{\max}$ at year 3.0. Bootstrap results.
Scenario in Table \ref{base_case_1}
Strategy computed in synthetic market. 
Tested in the bootstrapped historical market, expected blocksize 1.0 years.
$10^6$ bootstrap simulations.
Bootstrap date  based on the real CRSP index,
and real 30-day T-bills 1926:1-2024:12. Control computed
and stored using Algorithm \ref{alg_monotone_final} in the synthetic market.
$q_{min} = 30, q_{\max} = 60$ (per year).
Units: thousands of dollars.
}
\label{percentiles_one_point_three_boot_all}
\end{figure}

\begin{figure}[tb]
\centerline{%
\begin{subfigure}[t]{.33\linewidth}
\centering
\includegraphics[width=\linewidth]{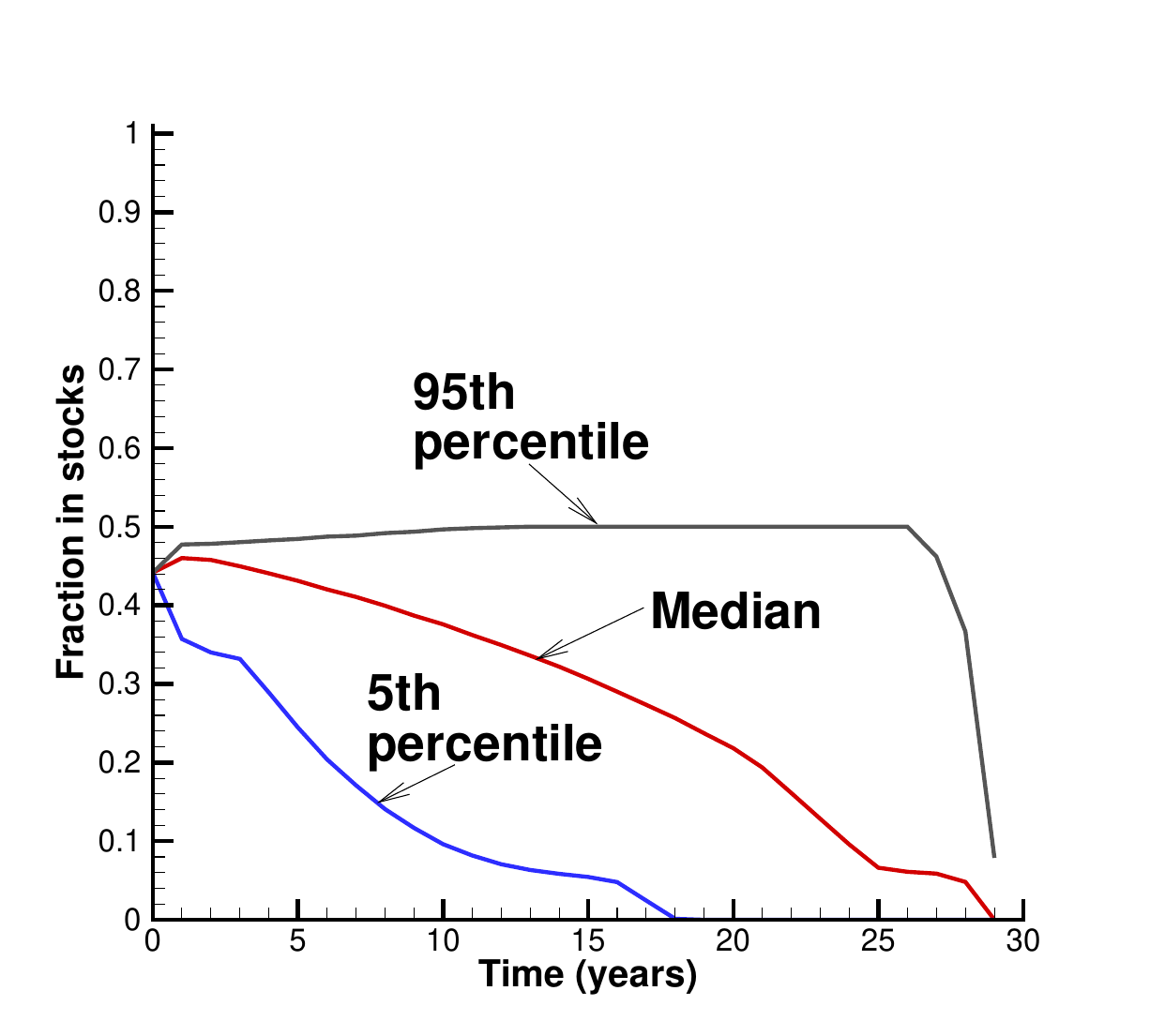}
\caption{Percentiles fraction in stocks}
\label{percentile_control_pmax_zero_point_five_fig_boot}
\end{subfigure}
\begin{subfigure}[t]{.33\linewidth}
\centering
\includegraphics[width=\linewidth]{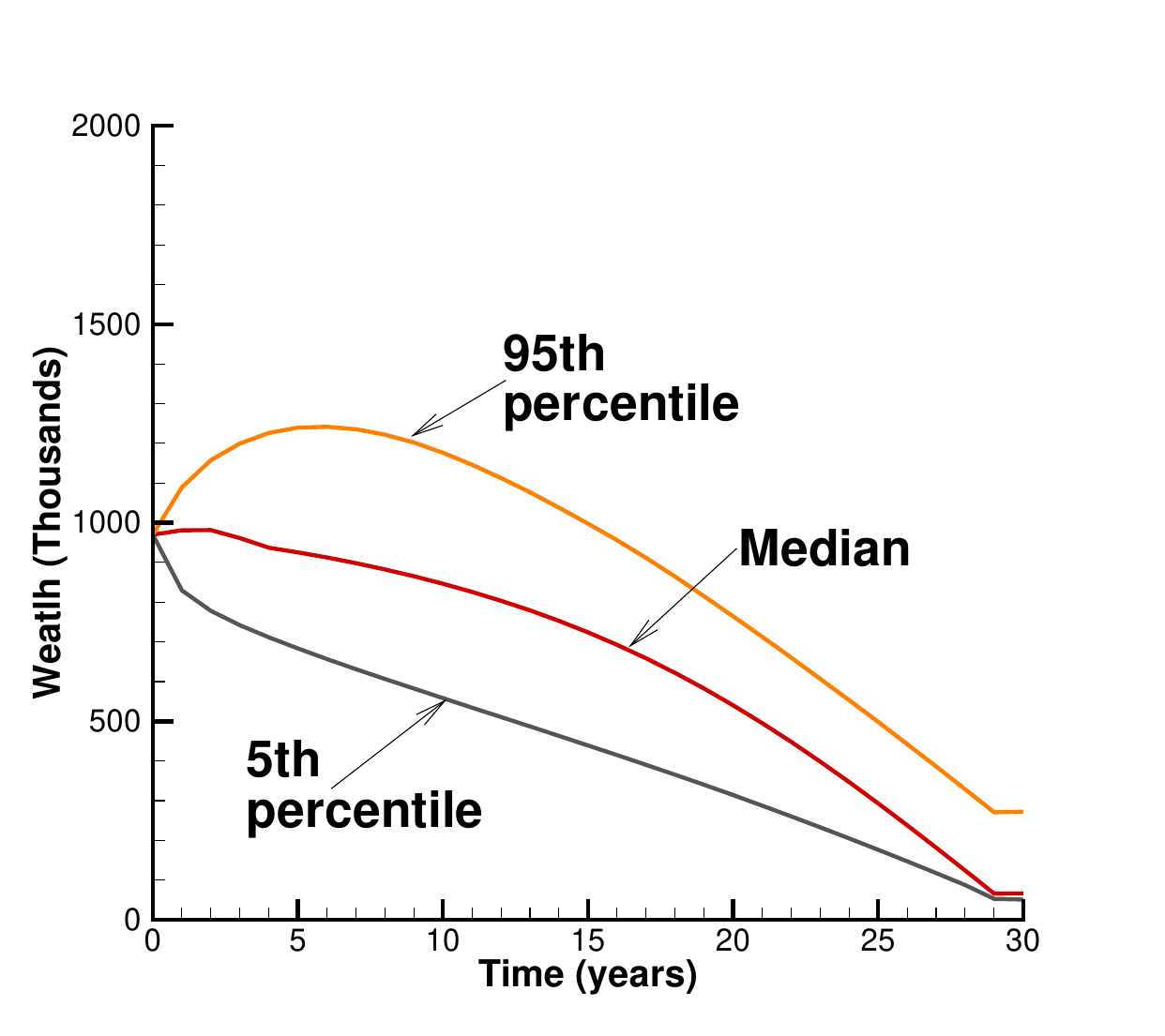}
\caption{Percentiles  wealth}
\label{percentile_wealth_pmax_zero_point_five_fig_boot}
\end{subfigure}
\begin{subfigure}[t]{.33\linewidth}
\centering
\includegraphics[width=\linewidth]{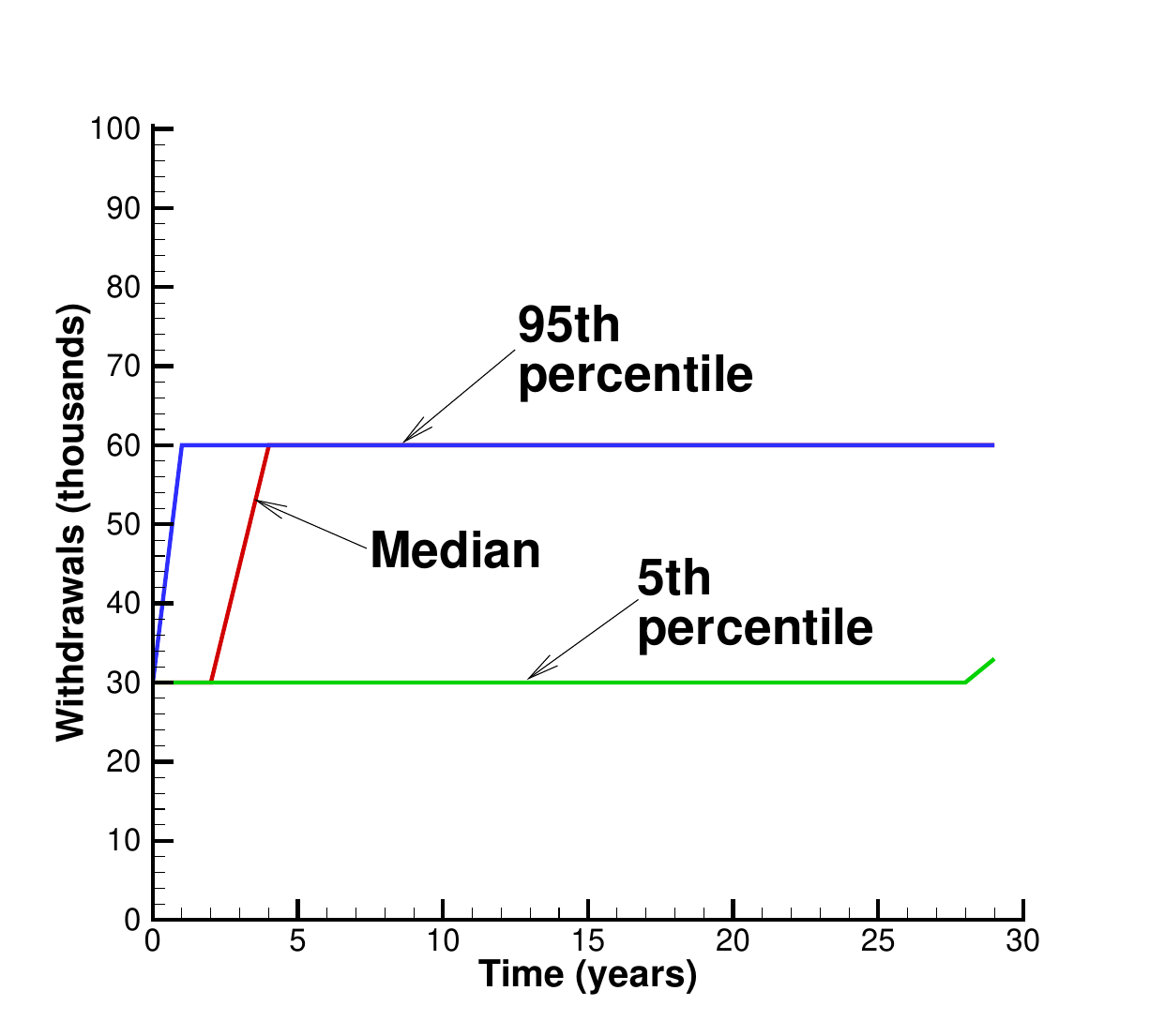}
\caption{Percentiles withdrawals}
\label{percentile_qplus_pmax_zero_point_five_fig_boot}
\end{subfigure}
}
\caption{
$p_{\max} = 0.5$. Percentiles of the controls.  Bootstrap results.
$EW = 50.3$, $ES = 6.7$ ($\kappa = 1.0$).
Median withdrawal at $q_{\max}$ at year 4.0.
Other information as in caption to Figure \ref{percentiles_one_point_three_boot_all}.
}
\label{percentiles_point_five_boot}
\end{figure}

\section{Limitations of the Methods Used}
We should note the following limitations of the PIDE-Fourier methods
used  in this Chapter.
\begin{itemize}
    \item PIDE methods (in general) are limited to problems with no more than
      three stochastic factors.  This excludes a number of problems of interest, for example portfolios
      having more than three risky assets, and adding stochastic mortality.

   \item Fourier based approaches require stochastic processes which
         have known, closed form expressions for the Fourier transform of the
         Green's function (equivalently the characteristic function).

   \item In order to ensure $\delta$-monotonicity, piecewise linear basis
         functions are used.  This caps the rate of convergence (for smooth
         functions) at second order in the grid spacing.   
         This contrasts with the 
       Chebyshev \citep{Gauss_2018,glau_2019} or Fourier based COS method \citep{Fang2008,Fang2009,Ruijter_2013}
         which, assuming smoothness, can achieve exponential convergence in some cases.
         But these techniques are not guaranteed to satisfy the $\delta$-monotonicity property needed to prove convergence.
         We also note that it is non-trivial to obtain high order
         convergence even for American options \citep{Ruijter_2013}, due to the jump in
         the second derivative across the exercise boundary.
         Careful tracking of the exercise boundary is necessary.  These methods
         would be difficult to implement for general optimal control problems.
         
\end{itemize}

At the same time these limitations are counterbalanced by the ability to obtain provably
convergent solutions to complex control problems, with realistic constraints.

Finally we should mention a recent approach for high dimensional stochastic control which 
is based on approximating the control by a neural network.  This machine learning
technique \citep{LiForsyth2019,
PvSForsythLi2022_stochbm,PvsForsythLi2023_NN,Chen_Mohib_2023} is 
purely data-driven, and makes no assumptions about the underlying
stochastic processes.  However, this also comes with the {\em caveat} that
convergence proofs require determination of  the global minimum
of a non-convex function, which is problematic in practice.

\section{Conclusions}
In this paper, we have developed a technique for solving the 
optimal control problem associated with decumulation of a defined
benefit (DC) pension plan.  The basic control algorithm at each rebalancing
time consists of (i) solution of an optimization problem and
(ii) advancing the solution to the next rebalancing time (going backwards).
We use a $\delta$-monotone scheme, based on Fourier methods, for
step (ii).  This method preserves order relations to $O(\delta)$, which is
an important property for numerical solution of optimal control.
We pay particular attention to ensuring that the Fourier wrap-around error
can be made arbitrarily small.

From a practical point of view, we have verified that the controls are robust
by testing the controls using block bootstrap resampling of historical data,
which makes no assumptions about stochastic processes for the underlying
stock and bond indexes.

Perhaps the most interesting result, in terms of investment strategies for
retirees, is the following.  If the maximum fraction in equities is restricted to be
less than 50\%, this strategy is only slightly
less efficient than allowing a maximum fraction of 130\% (that is,  using leverage).
This suggests that for most retirees, there is no need to undertake risky strategies
during the decumulation of a DC account.  

These optimal control
strategies are far more efficient than the typical four percent rule \citep{Bengen1994}.
The optimal policy  has an expected average withdrawal (real) of 
more than 5\% of initial capital over thirty years, with approximately 
zero expected shortfall at age 95 (based on historical data).  In contrast, the
strategy suggested in \citep{Bengen1994}, withdraws 4\% of initial
capital annually, but has an expected shortfall of  more than
35\% of initial capital at age 95.



\appendix

\section*{Appendices}

\appendix

\section{Wrap Around Error: Details}\label{wrap_appendix}

To gain some insight into the wrap-around problem, we consider a highly
simplified, one dimensional problem.
To avoid subscript clutter, in this section, we use the notation
\begin{eqnarray*}
    \tilde{g}({m- \ell}) \equiv \tilde{g}_{m- \ell}~;~  u^n(m) \equiv  u^n_m.
\end{eqnarray*}
Using the above notation, consider the discrete time advance convolution
\begin{eqnarray}
      u^n(m) & = & \Delta x \sum_{ \ell = -N^{\dagger}/2}^{N^{\dagger}/2-1}
                           \tilde{g}({m- \ell }) ~u^{n-1}({ \ell }) ~ , ~ \nonumber \\
               & & ~~~~~~~~~m=-N/2,\ldots, N/2-1 ~,~ N \leq N^{\dagger}~,
                          \quad
                          \label{wrap_1a}
\end{eqnarray}
where we assume periodic extension of $\tilde{g}$
\begin{eqnarray*}
     \tilde{g}( \ell \pm N^{\dagger} ) = \tilde{g}( \ell ) ~.
\end{eqnarray*}
In the above, if $N^{\dagger} > N$,  nodes corresponding to  $m < -N/2$ and $m > N/2-1$ are the padded nodes.
In this case, we assume that the padded nodal values $u^n(m)$,  
$ m < -N/2$ and $ m > N/2-1$, are determined by boundary data.

As an example of wrap-around error, we examine a worst case term in equation (\ref{wrap_1a}).  Consider the term in (\ref{wrap_1a}) corresponding to  $m=-N/2$, and $l =  N^{\dagger}/2 -1$, namely
\begin{eqnarray}
   \Delta x  ~\tilde{g}(-N/2 - N^{\dagger}/2 +1) ~u^{n-1} ({N^{\dagger}/2 -1 }).  \label{wrap_2}
\end{eqnarray}
By periodic extension, we shift the argument of $\tilde{g}(\cdot)$ by $N^{\dagger}$,
resulting in
\begin{eqnarray}
    \tilde{g}(-N/2 - N^{\dagger}/2 +1) = \tilde{g}(-N/2 - N^{\dagger}/2 +1 + N^{\dagger}) =  \tilde{g}( -N/2 + N^{\dagger}/2 +1),
    \label{wrap_error_1_c}
\end{eqnarray}
and hence, the term (\ref{wrap_2}) becomes
\begin{eqnarray}
   \Delta x ~ \tilde{g}(-N/2 + N^{\dagger}/2 +1) ~u^{n-1} ({N^{\dagger}/2 -1 }).
       \label{bad_term}
\end{eqnarray}
Hence, in this extreme case,  equation (\ref{wrap_1a}) becomes
\begin{eqnarray}
   u^{n}( -N/2) ~= ~ \Delta x ~ \tilde{g}(-N/2 + N^{\dagger}/2 +1) ~u^{n-1}({N^{\dagger}/2 -1 })
      ~+ \sum_{l=-N^{\dagger}/2}^{N^{\dagger}/2-2}({\mbox{ remaining terms }}). \nonumber \\
       \label{wrap_3}
\end{eqnarray}

\begin{example}[No padding: $N^{\dagger} = N$]
Suppose we do not use any padding, so that  $N^{\dagger} = N.$ In this case, equation (\ref{wrap_3}) becomes
\begin{eqnarray}
   u^{n}(-N/2) ~=~  \Delta x ~ \tilde{g}(1)~u^{n-1}({N/2 -1 })
      ~+ \sum_{l=-N/2}^{N/2-2}({\mbox{ remaining terms }}).
     \nonumber 
\end{eqnarray}
Since, in general, $\tilde{g}(1)$ is not small,  we can see that the term
$u^{n-1}({N/2 -1 })$ has a considerable effect on $u^{n}( -N/2)$,
which should not be the case.  We can see here that the periodic
extension of $\tilde{g}$ causes a {\em wrap-around} effect.
\end{example}

\begin{example}[Padding: $N^{\dagger} = 2N$]
If $N^{\dagger} = 2N$, then equation  (\ref{wrap_3}) becomes
\begin{eqnarray}
   u^{n}( -N/2)  ~=~  \Delta x ~ \tilde{g}( N/2 +1) ~u^{n-1}({N^{\dagger}/2 -1 })
                  ~+ \sum_{l=-N^{\dagger}/2}^{N^{\dagger}/2-2} ({\mbox{ other terms }}).
                   \label{wrap_5}
\end{eqnarray}
If we select $N$ sufficiently large so that $\tilde{g}( \pm N/2) \simeq 0$, 
then  the leading term in equation (\ref{wrap_5}) is small, and hence, wrap-around error is reduced.
\end{example}

\subsection{Wrap-around: a formal result}
From our previous examples, we can see that wrap-around may occur
in equation (\ref{wrap_1a}) if
\begin{eqnarray}
   (m- \ell ) < -N^{\dagger}/2 ~\mbox{ or } ~ (m- \ell) > N^{\dagger}/2 -1.  \nonumber 
\end{eqnarray}
This  leads us to the following formal definition of wrap-around error.
We show that, with $N^{\dagger} = 2N$, wrap-around error is sufficiently reduced.

\begin{definition}[Wrap-around error]
Assume $\{\tilde{g}(\cdot)\}$ is periodic with period $N^{\dagger}$
and $u^n(m)$, for $m < -N/2$ or $ m > N/2-1$, are determined by 
boundary data with $N^{\dagger} = 2N$.  Then the wrap-around error for equation (\ref{wrap_1a}),
at timestep $n$, denoted by $e_{\scalebox{.6}{wrap}}^{n}$, is
\begin{eqnarray}
   e_{\scalebox{.6}{wrap}}^{n} & = & \max_m \left\{ 
       \Delta x \sum_{\ell \in \mathbb{N}^{\dagger}} \Big\vert \tilde{g}({m- \ell}) ~u^{n-1}( {l})
                                                                    \Big\vert
                  ~ \biggl(
              {\bf{1}}_{ \{ (m- \ell) < - N^{\dagger}/2 \} } + {\bf{1}}_{ \{ (m- \ell) >  N^{\dagger}/2 -1\} }
                   \biggr)
           \right\}. \nonumber \\
      & &
         \label{wrap_error_1}
\end{eqnarray}
\end{definition}
We can now state the following result.
\begin{theorem}\label{wrap_theorem_1d}
Let $\tilde{g}(\cdot)$ be periodic with period $N^{\dagger}$ and
$u^n(m)$, for $m < -N/2$ or $ m > N/2-1$, be determined by boundary data with $N^{\dagger} = 2N$. 
Assume further that $u^{n}(m)$ is 
bounded, so for $0 \leq n \leq M$, there exists a positive constant $C$ such that\footnote{This
is essentially a stability condition.  See \citet{forsythlabahn2017} for a proof of
stability for the $\delta$-monotone method.}
\begin{eqnarray}
     | u^{n}(m) | \leq C, \quad m \in \mathbb{N}^{\dagger}~, \forall n~.
      \label{inf_norm_1}
\end{eqnarray}
If  $N$ is selected sufficiently large so that 
\begin{eqnarray}
    \Delta x \sum_{ \ell =-N^{\dagger}/2}^{ -N/2 -1}  | \tilde{g}({  \ell  })|
                                        +
    \Delta x \sum_{ \ell = {N/2} }^{ N^{\dagger}/2-1} | \tilde{g}({ \ell }  )|
                    ~  \leq~  \epsilon_e \Delta \tau~
             \label{wrap_condition_1d}
\end{eqnarray}
then the wrap-around error after $\mathcal{N}$ steps is bounded by $T C \epsilon_e$.
\end{theorem}

\begin{proof}
Applying property (\ref{inf_norm_1}) to equation (\ref{wrap_error_1}) gives
\begin{eqnarray}
   e_{\scalebox{.6}{wrap}}^{n} & \leq & C \max_m \biggl\{ \Delta x 
        \sum_{\ell = - N^{\dagger}/2}^{N^{\dagger}/2-1} ~ | \tilde{g}({m- \ell})|
                           ~
                         \biggl(
              {\bf{1}}_{ \{ (m- \ell ) < - N^{\dagger}/2 \} } + {\bf{1}}_{ \{ (m- \ell) >  N^{\dagger}/2 -1\} }
                   \biggr)
                    \biggr\}.
         \label{wrap_error_2}
\end{eqnarray}
Recall that $m \in \{ -N/2, \ldots, N/2-1 \}$, hence
the worst case values of $m$ on the righthand side of equation (\ref{wrap_error_2}) are $m= -N/2$ and $m  = N/2-1$. Thus
equation (\ref{wrap_error_2}) gives
\begin{eqnarray}
  e_{\scalebox{.6}{wrap}}^{n}  & \leq & C \Delta x \sum_{\ell = -N^{\dagger}/2}^{ N^{\dagger}/2-1} ~ |      
                \tilde{g}({N/2-1 - \ell })| {\bf{1}}_{ \{ (N/2 -1 - \ell) >  N^{\dagger}/2 -1\} }
                  \nonumber \\
       & & +~
                 C \Delta x \sum_{\ell = -N^{\dagger}/2}^{N^{\dagger}/2-1} ~ | \tilde{g}({-N/2 - \ell})|
                                      ~
                              {\bf{1}}_{ \{ (-N/2 - \ell) <  -N^{\dagger}/2 \} }.
         \label{wrap_error_3}
\end{eqnarray}
Also, since  $N = N^{\dagger}/2$ equation (\ref{wrap_error_3}) becomes
\begin{eqnarray}
e_{\scalebox{.6}{wrap}}^{n}  & \leq & C \Delta x \sum_{\ell = -N^{\dagger}/2}^{N^{\dagger}/2-1}
                       ~ | \tilde{g}({ N^{\dagger}/4-1 - \ell})|
                                       ~
                              {\bf{1}}_{ \{ ( N^{\dagger}/4  -1 - \ell ) >  N^{\dagger}/2 -1\} }
                  \nonumber \\
       & & +
                 ~C \Delta x \sum_{\ell = -N^{\dagger}/2}^{N^{\dagger}/2-1} ~ | \tilde{g}({-N^{\dagger}/4  -l})|
                                      ~
                              {\bf{1}}_{ \{ (-N^{\dagger}/4 -  \ell ) <  -N^{\dagger}/2 \} },
         \label{wrap_error_3a}
\end{eqnarray}
and eliminating the indicator functions gives
\begin{eqnarray}
e_{\scalebox{.6}{wrap}}^{n}  & \leq & C \Delta x \sum_{\ell =-N^{\dagger}/2}^{ -N^{\dagger}/4-1} 
              ~ | \tilde{g}({  N^{\dagger}/4 -1 -\ell})| ~
        +
                 C\Delta x \sum_{\ell = N^{\dagger}/4 + 1}^{N^{\dagger}/2-1} ~ | \tilde{g}(-N^{\dagger}/4  - \ell)|.
        \nonumber 
\end{eqnarray}
Shifting $\tilde{g}(\cdot)$ by $\pm N^{\dagger}$ so that the argument of $\tilde{g}(\cdot)$ is in
the range $[-N^{\dagger}/2, N^{\dagger}/2 -1]$, implies
\begin{eqnarray}
e_{\scalebox{.6}{wrap}}^{n}  & \leq & C\Delta x \sum_{\ell = -N^{\dagger}/2}^{ -N^{\dagger}/4-1} ~ 
           | \tilde{g}({  N^{\dagger}/4 -1 -\ell - N^{\dagger} })| ~
        +
                 C \Delta x \sum_{\ell= N^{\dagger}/4 + 1}^{N^{\dagger}/2-1} ~ 
              | \tilde{g}({ -N^{\dagger}/4  - \ell} + N^{\dagger}  )|
      \nonumber \\
   & & \nonumber  \\
 & = & C \Delta x \sum_{\ell =-N^{\dagger}/2}^{ -N^{\dagger}/4-1} ~ | \tilde{g}({  -3 N^{\dagger}/4 -1 - \ell  })| ~
        +
                 C \Delta x \sum_{\ell = N^{\dagger}/4 + 1}^{N^{\dagger}/2-1} ~ | \tilde{g}({ 3N^{\dagger}/4  - \ell}  )|.
         \label{wrap_error_5b}
\end{eqnarray}
Rearranging the indices, gives
\begin{eqnarray}
e_{\scalebox{.6}{wrap}}^{n}  & \leq & C \Delta x \sum_{ \ell =-N^{\dagger}/2}^{-N^{\dagger}/4-1} ~ | \tilde{g}({  \ell  })|
                                      ~
        +
                 C \Delta x \sum_{\ell = N^{\dagger}/4 + 1}^{ N^{\dagger}/2-1} ~ | \tilde{g}({ \ell }  )|,
         \label{wrap_error_5c}
\end{eqnarray}
which, since $N=N^{\dagger}/2$, implies that equation (\ref{wrap_error_5c}) satisfies
\begin{eqnarray}
e_{\scalebox{.6}{wrap}}^{n}  & \leq & C \Delta x \sum_{\ell =-N^{\dagger}/2}^{-N/2 -1} ~ | \tilde{g}({  \ell  })|
                                      ~
        +
                 C \Delta x \sum_{\ell = N/2 + 1}^{ N^{\dagger}/2-1} ~ | \tilde{g}({ \ell }  )| .
         \label{wrap_error_5d}
\end{eqnarray}
Since 
\begin{eqnarray}
     \Delta x \sum_{\ell = N/2 + 1}^{ N^{\dagger}/2-1} ~ | \tilde{g}({ \ell }  )| 
& \leq & \Delta x \sum_{\ell = N/2 }^{ N^{\dagger}/2-1} ~ | \tilde{g}({ \ell }  )|
      \label{strange_1}
\end{eqnarray}
then
\begin{eqnarray}
e_{\scalebox{.6}{wrap}}^{n}  & \leq & C \Delta x \sum_{\ell =-N^{\dagger}/2}^{-N/2 -1} ~ | \tilde{g}({  \ell  })|
                                      ~
        +
                 C \Delta x \sum_{\ell = N/2 }^{ N^{\dagger}/2-1} ~ | \tilde{g}({ \ell }  )|
                                     \label{wrap_error_sum}
              \\
    & = & C \epsilon_e \Delta \tau
         \label{wrap_error_5d_1}
\end{eqnarray}
where the last step follows from condition (\ref{wrap_condition_1d}).
Applying equation (\ref{wrap_error_5d_1}) recursively gives
the bound $T C \epsilon_e$.
\end{proof}

It might at first sight seem odd to weaken the error bound using equation (\ref{strange_1}).
However, this makes the final result easy to interpret.
Let
\begin{eqnarray}
   \mathcal{S}^{\dagger} = \{-N^{\dagger}/2, \ldots, +N^{\dagger}/2-1\} ~;~
    \mathcal{S} = \{-N/2, \ldots, +N/2-1\}
\end{eqnarray}
so that $\mathcal{S}^{\dagger}$ is the set of node indexes for the padded domain,
and $\mathcal{S}$ is the set of node indices for the original domain.
Consequently, condition (\ref{wrap_error_sum}) can be written compactly as
\begin{eqnarray}
   e_{\scalebox{.6}{wrap}}^{n} & \leq &
        C 
             \Delta x \sum_{\ell  \in  (\mathcal{S}^{\dagger} - \mathcal{S}) } | \tilde{g}({  \ell  })|
                    ~  \leq~  C \epsilon_e \Delta \tau~.
                   \label{wrap_compact_1d}
\end{eqnarray}

Proving Theorem \ref{wrap_theorem} basically follows the previous proof of the simpler case.  Divide  the region $\mathcal{D}^\prime - \mathcal{D}$ outside $\mathcal{D}$  into right and left $A_r, A_\ell$ and upper and lower $A_u, A_d$ regions as follows.
\begin{center}
\begin{tikzpicture}
\draw (0,0) -- (8, 0) -- (8,4) -- (0,4) -- (0,0);
\draw (2,1) -- (6, 1) -- (6, 3) -- (2, 3) -- (2,1);
\draw[dashed] (2,0) -- (2, 4);
\draw[dashed] (6,0) -- (6, 4);
\node at (4,2) {$\mathcal{D}$};
\node at (1,2) {$A_\ell$};
\node at (7,2) {$A_r$};
\node at (4,3.5) {$A_u$};
\node at (4,0.5) {$A_d$};
\end{tikzpicture}
\end{center}
Our goal is to give an upper bound to 
\begin{eqnarray*}
    \varepsilon^n_{\text{wrap}}
   & = & 
            \Delta x_1 \Delta x_2 \max_{ (\ell, m) \in \mathcal{D} }
           \sum_{ ( j ,k) \in \mathcal{D}^{\dagger} }
            \big\vert \tilde{g}_{\ell -j, m - k}^{\dagger} 
                      \plainv_{j,k}^{\dagger}(\tau^n) \big\vert
               ~{\bf{1}}_{ (\ell-j, m-k) \notin \mathcal{D}^{\dagger} }
   ~,
 \end{eqnarray*}
which by the boundedness assumptions of Theorem \ref{wrap_theorem} is the same as bounding
 \begin{eqnarray*}
   &  & \Delta x_1 \Delta x_2 \max_{ (\ell, m) \in \mathcal{D} }
           \sum_{ ( j ,k) \in \mathcal{D}^{\dagger}}
            \big\vert \tilde{g}_{\ell -j, m - k}^{\dagger}  \big\vert 
               ~{\bf{1}}_{ (\ell-j, m-k) \notin \mathcal{D}^{\dagger} }
        ~.
 \end{eqnarray*}

As in the 1d wraparound error, our worse case values occur at borders, in this case the four corners of $\mathcal{D}$. In the horizontal direction when the first components are $-N_1/2$ and $N_1/2-1$ we can use the same manipulations as done in the 1d case to obtain the wraparound error contribution from $A_\ell$ and $A_r$ and use an identical argument when the second argument is $-N_2/2$ and $N_2/2-1$ to get  the wraparound error contribution from $A_u$ and $A_d$. These manipulations reduce to 
\begin{eqnarray*}
   e_{\scalebox{.6}{wrap}}^{n} & \leq &
        C \cdot
           \Delta x_1  \Delta x_2 \sum_{(\ell, m)   \in  (\mathcal{D}^{\dagger} - \mathcal{D}) } 
               | \tilde{g}({  \ell, m  })|
                    ~  \leq~  C \epsilon_e \Delta \tau~
\end{eqnarray*}
for each time step and hence give the error bound $ e_{\scalebox{.6}{wrap}}^{\mathcal{N}} \leq T \epsilon_e C$ after $\mathcal{N}$ steps.

\begin{singlespace}
\bibliographystyle{chicago}

\end{singlespace}

\end{document}